\documentclass[numbers]{sigplanconf} \toappear{} \makeatletter
\def\@copyrightspace{\relax}
\makeatother

\usepackage[utf8]{inputenc}
\usepackage[T1]{fontenc}
\usepackage{microtype}

\newif\ifproofs
\proofstrue

\title{A Short Counterexample Property for Safety and Liveness
Verification of Fault-Tolerant Distributed Algorithms\thanks{
This is an extended version of the paper that will appear at POPL'17,
which can be accessed at: \url{http://dx.doi.org/10.1145/3009837.3009860}
}}

\authorinfo{Igor Konnov \and Marijana Lazi\'{c} \and Helmut Veith\thanks{We dedicate this article to the memory of Helmut Veith,
    who  passed
    away tragically after we finished the first draft together. 
    In addition to contributing to this work, Helmut initiated
      our long-term research program on verification
    of fault-tolerant distributed algorithms, which made this paper possible.} \and Josef Widder}
          {TU Wien, Austria}
          {\{konnov, lazic, veith, widder\}@forsyte.at}

\newcommand{\arxiv}[1]{
}

\usepackage{datetime}
\usepackage{enumitem}
\setlist{nolistsep}

\usepackage{mathtools}
\usepackage{amsfonts,amssymb,amsmath}

\usepackage{amsthm}
\usepackage{etoolbox}  \usepackage{enumitem} 

\newcommand{\url}[1]{\texttt{#1}}

\usepackage{listings}
\usepackage{environ}
\usepackage{csvsimple}
\usepackage{longtable} \usepackage[normalem]{ulem} 
\usepackage{fourier-orns} \usepackage{wasysym}        \usepackage{bbding} 
\usepackage{listings}
\lstdefinelanguage{promela}
  {morekeywords={do,od,init,proctype,for,if,fi,else,goto,byte,int,bool,bit,chan,mtype,atomic,d_step,nempty,empty,break,skip,active,ltl,symbolic,assume,some,all,card},
  morecomment=[s]{/*}{*/},mathescape=true,escapechar={@},
  basicstyle=\ttfamily\small,   commentstyle=\itshape\rmfamily\small\color{black!50},
  keywordstyle=\bfseries\small}
\lstdefinelanguage{distal}
  {morekeywords={case,class,extends,val,var,UPON,RECEIVING,START,WITH,
      TIMES,DO,IF,THEN,SEND,TO,ALL,False,True},
  morecomment=[s]{/*}{*/}, morecomment=[l]{//}, 
  mathescape=true,escapechar={@},
  basicstyle=\sffamily\small,
  commentstyle=\itshape\rmfamily\small,   numbers=left, numberstyle=\tiny,
  xleftmargin=2em, framexleftmargin=1.5em
}

\lstdefinelanguage{pseudo}
  {morekeywords={init,with,or,if,then,else,fi,and,not,while,do,od,distinct,
    case, goto,local,algorithm, function, for, each, times, from, to,
    variables, procedure, recursive, return},
  morecomment=[l]{//}, morecomment=[s]{/*}{*/},
  mathescape=true,escapechar={@},
  basicstyle=\sffamily\small,
  commentstyle=\itshape\rmfamily\small,
  keywordstyle=\sffamily\bfseries\small
}
\date{}

\usepackage{tikz}
\usetikzlibrary{arrows}
\usetikzlibrary{shapes}
\usetikzlibrary{calc}
\usetikzlibrary{positioning}
\usetikzlibrary{patterns}
\usetikzlibrary{decorations}
\usetikzlibrary{decorations.pathmorphing}
\usetikzlibrary{decorations.footprints}
\usetikzlibrary{decorations.markings}
\usetikzlibrary{shadows}
\usetikzlibrary{shapes.symbols}
\usetikzlibrary{decorations.pathreplacing}

\usepackage{array}
\newcounter{rowno}
\newtheorem{observation}{Observation}

\newtheorem{proposition}{Proposition}[section]
\newtheorem{definition}[proposition]{Definition}
\newtheorem{theorem}[proposition]{Theorem}
\newtheorem{lemma}[proposition]{Lemma}
\usepackage{thmtools}

\theoremstyle{definition}\newtheorem{example}[proposition]{Example}
\AtEndEnvironment{example}{\null\hfill$\triangleleft$}

\usepackage{color}
\newcommand{\highlight}[1]{\colorbox{gray!40}{#1}}

\usepackage{ifthen}

\newcommand{\recallthm}[2]{  {\medskip\noindent\bfseries Theorem~\ref{#1}.~}{\itshape #2}
}
\newcommand{\recallproposition}[2]{  {\medskip\noindent\bfseries Proposition~\ref{#1}.~}{\itshape #2}
}

\makeatletter{}\newcommand{\trrule}{{\mathit{rule}}}
\newcommand{\trfactor}{{\mathit{factor}}}

\newcommand{\concat}{\cdot}

\newcommand{\relc}{\prec_{\scriptscriptstyle{\mathit{C}}}}
\newcommand{\linrelc}{\prec^{\mathit{lin}}_{\scriptscriptstyle{\mathit{C}}}}

\newcommand{\Flowletter}{P}
\newcommand{\relf}{\prec_{\scriptscriptstyle{\mathit{\Flowletter}}}}
\newcommand{\erelf}{\sim_{\scriptscriptstyle{\mathit{\Flowletter}}}}
\newcommand{\nerelf}{\not\sim_{\scriptscriptstyle{\mathit{\Flowletter}}}}

\newcommand{\relftrans}{\prec^+_{\scriptscriptstyle{\mathit{\Flowletter}}}}

\newcommand{\ident}[1]{\textit{#1\/}\rule{0cm}{1ex}}
\gdef\dash---{\thinspace---\hskip.16667em\relax}
\gdef\ndash---{\thinspace--\hskip.16667em\relax}

\newcommand{\Natural}{{\mathbb N}}
\newcommand{\NatZero}{{\mathbb N}_0}
\newcommand{\Int}{{\mathbb Z}}

\newcommand{\paraset}{\Pi}
\newcommand{\globset}{\Gamma}
\newcommand{\ruleset}{{\mathcal R}}

\newcommand{\ruleclass}{\ruleset/{\sim}}            
\newcommand{\classof}[1]{[{#1}]}      \newcommand{\local}{{\mathcal L}}
\newcommand{\initlocal}{{\mathcal I}}

\newcommand{\numlocal}{{|\local|}}

\newcommand{\AdmP}{\mathbf{P}_{RC}}
\newcommand{\param}{\mathbf{p}}

\newcommand{\configs}{\Sigma}
\newcommand{\iconfigs}{I}
\newcommand{\transrel}{R}
\newcommand{\finpath}[2]{\textsf{path}(#1, #2)} \newcommand{\infpath}[2]{\textsf{path}(#1, #2)} 
\newcommand{\pc}{\ident{sv}}

\newcommand{\rcvd}{\ident{rcvd}}
\newcommand{\sent}{\ident{nsnt}}

\newcommand{\numparam}{{|\paraset|}}
\newcommand{\numglob}{{|\globset|}}

\newcommand{\rules}{{\mathcal R}}
\newcommand{\precondLE}{\varphi^{\scriptscriptstyle{\le}}}
\newcommand{\precondG}{\varphi^{\scriptscriptstyle{>}}}

\newcommand{\PrecondU}{\Phi^{\mathrm{rise}}}
\newcommand{\PrecondL}{\Phi^{\mathrm{fall}}}

\newcommand{\Ctx}{\Omega}
\newcommand{\CtxU}{\Omega^\mathrm{rise}}
\newcommand{\CtxL}{\Omega^{\mathrm{fall}}}

\newcommand{\statectx}{\omega}
\newcommand{\rulesliceclass}{\raise0ex\hbox{$(\slice{\ruleset}{\Ctx})$}\big/\lower.5ex\hbox{\hskip-.25em{{\scriptsize $\sim$}}}}
\newcommand{\update}{\vec{u}}

\newcommand{\syssize}{N}

\newcommand{\true}{\mathit{true}}

\newcommand{\fromstate}{{\mathit{from}}}
\newcommand{\tostate}{{\mathit{to}}}

\newcommand{\ResCond}{{\ident{RC}}}

\newcommand{\Sk}{\textsf{TA}}
\newcommand{\TA}{\Sk}

\newcommand{\Sys}{\textsf{Sys}}

\newcommand{\newreftheorem}[2]{\newenvironment{#1}[1]{\par\vspace{3mm}\noindent\textbf{#2~\ref{##1}.}
\em}{\rm}
}
\newreftheorem{reflemma}{Lemma}
\newreftheorem{refproposition}{Proposition}
\newreftheorem{reftheorem}{Theorem}
\newreftheorem{refcorollary}{Corollary}

\newcommand{\cpp}[1]{#1\scriptsize{\texttt{++}}}

\newcommand{\counters}{{\vec{\boldsymbol\kappa}}}
\newcommand{\vars}{\vec{g}}

\newcommand\gst{\sigma}

\newcommand{\IT}{\mathrm{V0}}
\newcommand{\RI}{\mathrm{V1}}

\newcommand{\movers}{\cite{CohenL98,Doeppner77,Lamport89pretendingatomicity,Elmas09,Flanagan:2005,KVW16:IandC}}

\newcommand\tlE{\textsf{\textbf{E}}\,}

\newcommand\ltlF{\textsf{\textbf{F}}\,}
\newcommand\ltlG{\textsf{\textbf{G}}\,}

\newcommand\ltlE{\,\textsf{\textbf{E}}\,}

\newcommand\LTL[0]{$\mbox{\textsf{LTL}}$}
\newcommand\ELTL[0]{$\mbox{\textsf{ELTL}}$}

\newcommand\ELTLTB{\textsf{ELTL}_\textsf{FT}}

\renewcommand\vec[1]{\mathbf{#1}}

\newcommand\tbh[1]{\textsf{\textbf{\scriptsize{#1}}}}

\newcommand\jrep[1]{\mathsf{rep}[#1]}

\newcommand\xrep[2]{\mathsf{srep}[#1,#2]}

\newcommand\slice[2]{#1{\raise-.5ex\hbox{\ensuremath|}}_{#2}}

\newcommand\xschema[0]{\mathsf{sschema}}

\newcommand\proj[2]{{#1}|_{#2}}

\newcommand{\multipl}{\mu}
\newcommand{\reprlive}{\mathsf{repr}[\psi , \gst, \tau]}
\newcommand{\sro}[3]{\mathsf{repr}_{\vee}[\psi , #2,#3]}

\newcommand{\srogen}{\sro \Ctx\gst\tau}
\newcommand{\sr}[3]{\mathsf{srep}[#2,#3]}
\newcommand{\srgen}{\sr \Ctx\gst\tau}
\newcommand{\bk}{\boldsymbol\kappa}

\newcommand{\critical}{\mathit{Locs}}
\newcommand{\gcritical}{n}
\newcommand{\gcri}{m}
\newcommand{\gsro}[4]{\mathsf{repr}_{\wedge\vee}[\psi ,\multist{#2}{#4},\multisch{#3}{#4}]}
\newcommand{\gsrogen}{\gsro \Ctx\gst\tau{\multipl}}

\newcommand{\typeall}{A}
\newcommand{\typebeg}{B}
\newcommand{\typeend}{C}
\newcommand{\typemid}{D}
\newcommand{\typebegend}{E}
\newcommand{\typenot}{F}

\newcommand{\nmodels}{\not\models}
\newcommand{\setconf}[2]{\textsf{Cfgs}(#1,#2)}   
\newcommand{\naming}{\eta}
\newcommand{\move}[2]{#1_{#2 { \scriptscriptstyle \leftarrow}}}
\newcommand{\movemove}[3]{#1_{#2 {\scriptscriptstyle \leftarrowtail} #3}}
\newcommand{\threads}{\Theta}

\newcommand{\newrest}[3]{\proj{#1}{#2,\Natural\setminus\{#3\}}}

\newcommand{\firststate}[1]{\mathrm{first}(#1)}
\newcommand{\laststate}[1]{\mathrm{last}(#1)}
\newcommand{\middlestate}[1]{\mathrm{middle}(#1)}
\newcommand{\multist}[2]{{#2}{#1}}
\newcommand{\multisch}[2]{{#2} {#1}}

\newcommand{\propfair}{\psi_\text{fair}}

\newcommand{\prop}{\mathit{prop}}

\newcommand{\loops}{\mathsf{loop}_{\mathsf{start}}}
\newcommand{\loope}{\mathsf{loop}_{\mathsf{end}}}

\newcommand{\canform}{\mathit{can}}
\newcommand{\syntree}{{\cal T}}
\newcommand{\pgraph}{{\cal G}}
\newcommand{\pgvertices}{{\cal V}_{\pgraph}}
\newcommand{\pgedges}{{\cal E}_{\pgraph}}
\newcommand{\tgraph}{{\cal H}}
\newcommand{\tgvertices}{{\cal V}_{\tgraph}}
\newcommand{\tgedges}{{\cal E}_{\tgraph}}

\newcommand{\imap}{\zeta}
\newcommand{\pnode}{\operatorname{\mathit{p-node}}}
\newcommand{\pcutpoint}{\operatorname{\mathit{p-cutpoint}}}
\newcommand{\EA}{\operatorname{\mathit{EA}}}

\newcommand{\proofintext}[1]{}

\begin{document}

\maketitle

\begin{abstract}
Distributed algorithms have many mission-critical applications ranging
     from  embedded systems and replicated databases to cloud
     computing.
Due to asynchronous communication, process faults, or network
     failures, these algorithms are difficult to design and verify.
Many algorithms achieve fault tolerance by using threshold guards
     that, for instance, ensure that a process waits until it has
     received an acknowledgment from a majority of its peers.
Consequently, domain-specific languages for fault-tolerant distributed
     systems offer language support for threshold guards.

We introduce an automated method for  model checking of safety and
     liveness of threshold-guarded distributed algorithms in systems
     where the number of processes and the fraction of faulty
     processes are parameters.
Our method is based on a \emph{short counterexample property}: if a
     distributed algorithm violates a temporal specification (in a
     fragment of \LTL), then
     there is a counterexample whose length is bounded and independent
     of the parameters.
We prove this property by (i) characterizing executions depending on
     the structure of the temporal formula, and (ii) using
     commutativity of transitions to accelerate and shorten
     executions.
We extended the ByMC toolset (Byzantine Model Checker) with our
     technique, and verified liveness and safety of 10 prominent
     fault-tolerant distributed algorithms, most of which were out of
     reach for existing techniques.
\end{abstract}



    \makeatletter{\renewcommand*{\@makefnmark}{}
    \footnotetext{Supported by: the Austrian Science~Fund~(FWF) through the
    National Research Network RiSE (S11403 and S11405), project PRAVDA
    (P27722), and Doctoral College LogiCS (W1255-N23); and by the Vienna Science and Technology
    Fund (WWTF) through project APALACHE (ICT15-103).}\makeatother}

\makeatletter{}\section{Introduction}\label{sec:intro}

\newcommand{\awesomepapers}{\cite{GleissenthallBR16,KillianABJV07,BielyD0S13,DragoiHZ16,LesaniBC16,Pel16,KVW15:CAV}}

Distributed algorithms have many applications in avionic and
     automotive embedded systems, computer networks, and the internet
     of things.
The central idea is to achieve dependability by  replication, and
     to ensure that all correct replicas behave as one, even if some of
     the replicas fail.
In this way, the correct operation of the system is more reliable than
     the correct operation of its parts.
Fault-tolerant algorithms typically have been used in applications
     where highest reliability is required because human life is at
     risk (e.g., automotive or avionic industries), and even unlikely
     failures of the system are not acceptable.
In contrast, in more mainstream applications like replicated
     databases, human intervention to restart the system from a
     checkpoint was often considered to be acceptable, so that
     expensive fault tolerance mechanisms were not used in
     conventional applications.
However, new application domains such as cloud computing provide a new
     motivation to study fault-tolerant algorithms: with the huge
     number of computers involved, faults are the norm~\cite{Netflix5}
     rather than an exception, so that fault tolerance becomes an
     economic necessity; and so does the correctness of fault
     tolerance mechanisms.
Hence, design, implementation, and verification of distributed systems
     constitutes an active research area~\awesomepapers{}.
Although distributed algorithms show complex behavior, and are
     difficult to understand for human engineers, there is only very
     limited tool support to catch logical errors in fault-tolerant
     distributed algorithms at design time.

The state of the art in the design of fault-tolerant 
     systems is exemplified by the recent work on Paxos-like
     distributed algorithms like Raft~\cite{Ongaro2014} or
     M$^2$PAXOS~\cite{Pel16}.
The  designers encode these 
     algorithms in TLA+~\cite{TLA}, and use the TLC model checker to
     automatically find bugs in small instances, i.e., in
     distributed systems containing, e.g., three  processes.
Large distributed systems (e.g., clouds) need guarantees for
     \emph{all} numbers of processes.
These guarantees are  typically given using hand-written
     mathematical proofs.
In principle, these proofs could be encoded and machine-checked using
     the TLAPS proof system~\cite{ChaudhuriDLM10}, PVS~\cite{LR93}, 
Isabelle~\cite{Charron-BostM09}, Coq~\cite{LesaniBC16},
Nuprl~\cite{RahliGBC15}, or similar
     systems; but this requires human expertise in the proof checkers
     and in the application domain, and a lot of effort.

Ensuring  correctness of the implementation is an open challenge:
As the implementations are done by hand~\cite{Ongaro2014,Pel16}, the
     connection between the specification and the implementation is
     informal, such that there is no formal argument about the
     correctness of the implementation.
To address the discrepancy between design, implementation, and
     verification,  Dr\u{a}goi et al.~\cite{DragoiHZ16} introduced a
     domain-specific language PSync which is used for two purposes:
     (i) it compiles into running code, and (ii) it is used for
     verification.
Their verification approach~\cite{DHVWZ14}, requires a developer to
     provide invariants, and similar verification conditions.
While this approach requires less human intervention than writing
     machine-checkable proofs, coming up with invariants of
     distributed systems requires considerable human ingenuity.
The Mace~\cite{KillianABJV07} framework is based on a similar idea,
     and is an extension to C++.
While being fully automatic, their approach to correctness is
     light-weight in that it uses a tool that explores random walks to
     find (not necessarily all) bugs, rather than actually verifying
     systems.

\begin{figure}\label{fig:st}
\begin{lstlisting}[language=distal]
case class EchoMsg extends Message

class ReliableBroadcastOnce
      extends DSLProtocol {
 val n = ALL.size         // nr. processes
 val t = ALL.size / 3 - 1 // max. faults
 var accept: Boolean = False

 UPON RECEIVING START WITH v DO {@\label{line:init}@
  IF v == 1 THEN // check the initial value
  SEND EchoMsg TO ALL
 }
 UPON RECEIVING EchoMsg TIMES t + 1 DO {@\label{line:tp1}@
  SEND EchoMsg TO ALL // >= 1 correct
 }
 UPON RECEIVING EchoMsg TIMES n - t DO {@\label{line:nmt}@
  accept = True // almost all correct
 }
}
\end{lstlisting}
\caption{Code example of a distributed algorithm in DISTAL~\cite{BielyD0S13}.
  A distributed system consists of $n$ processes,
    at most $t < n/3$ of which are Byzantine faulty. The correct ones
    execute the code, and no assumptions is made about the faulty processes.
}
\end{figure}

In this paper we focus on automatic verification methods for
     programming constructs that are typical for fault-tolerant
     distributed algorithms.
Figure~\ref{fig:st} is an example of
a distributed algorithm in the domain-specific language
     DISTAL~\cite{BielyD0S13}.
It encodes the core of the reliable broadcast protocol
     from~\cite{ST87:abc}, which is used as building block of many
     fault-tolerant distributed systems.
Line~\ref{line:tp1} and Line~\ref{line:nmt}  use so-called ``threshold
     guards'' that check whether a given number of messages from
     distinct senders arrived at the receiver.
As threshold guards are the central algorithmic idea for fault
     tolerance, domain-specific languages such as DISTAL or PSync have
     constructs for them (see~\cite{DragoiHZ16} for an 
     overview of domain-specific languages and formalization
     frameworks for distributed systems).
For instance, the code in Figure~\ref{fig:st} works for systems with
     $n$ processes among which~$t$ can fail, with $t < n/3$ as
     required for Byzantine fault tolerance~\cite{LSP80}.
In such systems, waiting for messages from $n-t$ processes ensures
     that if all correct processes send messages, then faulty
     processes cannot prevent progress.
Similarly, waiting for $t+1$ messages ensures that at least one
     message was sent by a correct process.
Konnov et al.~\cite{KVW15:CAV} introduced an automatic method to
     verify safety of algorithms with threshold guards.
Their method is parameterized in that it verifies distributed
     algorithms for all values of parameters ($n$ and $t$) that
     satisfy a resilience condition ($t < n/3$).
This work bares similarities to the classic work on reduction for
     parallel programs by Lipton~\cite{Lipton75}.
Lipton proves statements like ``all $P$ operations on a semaphore are
     left movers with respect to operations on other processes.'' 
He proves that given a run that ends in a given state,
     the same state is reached by the run in which the $P$ operation
     has been moved.
Konnov et al.~\cite{KVW15:CAV} do a similar analysis for
     threshold-guarded operations, in which they analyze the relation
     between statements from Figure~\ref{fig:st} like ``\texttt{send
     EchoMsg}'' and  ``\texttt{UPON RECEIVING EchoMsg TIMES t + 1}''
     in order to determine which statements are movable.
From this, they develop an offline partial order reduction that
     together with acceleration~\cite{BardinFLP08,KVW16:IandC} reduced
     reachability checking to complete bounded
     model checking using SMT.
In this way, they automatically check safety of fault-tolerant
     algorithms.

However, for  fault-tolerant distributed algorithms liveness is as
     important as safety: This comes from the celebrated impossibility result by Fischer, Lynch,
     and Paterson~\cite{FLP85} that states that a fault-tolerant
     consensus algorithm cannot ensure both safety and liveness in
     asynchronous systems.
It is folklore that designing a safe fault-tolerant
     distributed algorithm is trivial: \emph{just do nothing}; e.g., by never committing transactions, one cannot commit
     them in inconsistent order.
Hence, a technique that verifies only safety  may establish the
     ``correctness'' of a distributed algorithm that never does
     anything useful.
To achieve trust in correctness of a distributed algorithm,
    we need tools that verify both safety and liveness.

As exemplified by~\cite{FarzanKP16}, liveness verification of
     parameterized distributed and concurrent systems is still a
     research challenge.
Classic work on parameterized model checking by German and
     Sistla~\cite{GS1992} has several restrictions on the
     specifications ($\forall i.\, \phi(i)$) and the computational
     model (rendezvous), which are incompatible with fault-tolerant
     distributed algorithms.
In fact, none of the approaches
     (e.g.,~\cite{CTV2008,EN95,EmersonK03LICS,PXZ02}) surveyed
     in~\cite{2015Bloem} apply to the algorithms we consider.
More generally, in the parameterized case, going from safety to
     liveness is not straightforward.
There are systems where safety is decidable and liveness is
     not~\cite{Esparza99}.

\newcommand{\casest}{\cite{CT96,ST87:abc,BrachaT85,MostefaouiMPR03,Raynal97,Gue02,DobreS06,BrasileiroGMR01,SongR08}}

\paragraph{Contributions.} We generalize the approach by Konnov
    et al.~\cite{KVW16:IandC,KVW15:CAV} to liveness by 
presenting a framework and a model checking tool that takes as input a
    description of a distributed algorithm (in our variant~\cite{GKSVW14:SFM} of Promela~\cite{H2003}) and specifications
    in a
    fragment of linear temporal logic. 
It then shows correctness for all parameter values (e.g., $n$ and $t$) that
    satisfy the required resilience condition (e.g., $t < n/3$), or reports
    a counterexample:
    \begin{enumerate}

\item As in the classic result by Vardi and Wolper~\cite{VW86}, we
     observe that it is sufficient to search for counterexamples that
     have the form of a lasso, i.e., after a finite prefix an
     infinite loop is entered.
Based on this, we analyze specifications automatically, in order to
     enumerate possible shapes of lassos depending on temporal
     operators~$\ltlF$ and~$\ltlG$ and evaluations of threshold guards.

\item We automatically do offline partial order reduction
    using the algorithm's description.
For this, we introduce a more refined mover analysis for threshold
     guards and temporal properties.
We extend Lipton's reduction method~\cite{Lipton75}
     (re-used and extended by many others~\movers), so that we maintain
     invariants, which allows us to go beyond reachability and
     verify specifications with the temporal operators $\ltlF$ and
     $\ltlG$.

\item By combining acceleration~\cite{BardinFLP08,KVW16:IandC} with
     Points~1 and~2, we obtain a short counterexample property, that is,
     that infinite executions (which may potentially be
     counterexamples) have ''equivalent'' representatives of bounded
     length.
The bound depends on the process code and is independent of the
     parameters.
The equivalence is understood in terms of temporal logic
     specifications that are satisfied by the original executions and
     the representatives, respectively.
We show that the length of the representatives increases mildly
     compared to reachability checking in~\cite{KVW15:CAV}.
This implies a so-called completeness threshold~\cite{KroeningOSWW11}
     for threshold-based algorithms and our fragment of~\LTL{}.

\item Consequently, we only have to check a reasonable number of SMT queries
        that encode parameterized and bounded-length representatives of
        executions.
We show that if the parameterized system violates a temporal property, then SMT
reports a counterexample for one of the queries. 
We prove that otherwise the specification holds for all system sizes.

\item Our theoretical results and our implementation push
        the boundary of liveness verification for fault-tolerant
        distributed algorithms. 
While prior results~\cite{JohnKSVW13:fmcad} scale just to two out of ten
    benchmarks from~\cite{KVW15:CAV}, we verified safety and liveness
of all ten. 
These benchmarks originate from distributed algorithms~\casest{} that
    constitute the core of important services such as replicated state
    machines.
\end{enumerate}

From a theoretical viewpoint, we introduce new concepts and conduct
extensive proofs (the proofs can be found in~\cite{KLVW16:arxiv}) 
for Points~1 and~2.
From a practical viewpoint, we have built a complete framework for
     model checking of fault-tolerant distributed algorithms that use
     threshold guards, which constitute the central programming
     paradigm for dependable distributed systems.

\makeatletter{}
\section{Representation of Distributed Algorithms}\label{sec:Ab}

\subsection{Threshold Automata}
\label{sec:TA}

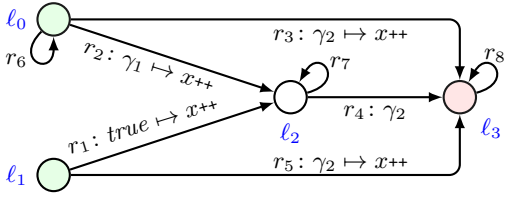
\begin{figure}[t]
    \begin{center}
    {\makeatletter{}
\tikzstyle{node}=[circle,draw=black,thick,minimum size=4.3mm,inner sep=0.75mm,font=\normalsize]
\tikzstyle{init}=[circle,draw=black!90,fill=green!10,thick,minimum size=4.3mm,inner sep=0.75mm,font=\normalsize]
\tikzstyle{final}=[circle,draw=black!90,fill=red!10,thick,minimum size=4.3mm,inner sep=0.75mm,font=\normalsize]
\tikzstyle{rule}=[->,thick]
\tikzstyle{post}=[->,thick,rounded corners,font=\normalsize]
\tikzstyle{pre}=[<-,thick]
\tikzstyle{cond}=[rounded
  corners,rectangle,minimum
  width=1cm,draw=black,fill=white,font=\normalsize]
\tikzstyle{asign}=[rectangle,minimum
  width=1cm,draw=black,fill=gray!5,font=\normalsize]

\tikzset{every loop/.style={min distance=5mm,in=140,out=113,looseness=2}}
\begin{tikzpicture}[>=latex, thick,scale=0.9, every node/.style={scale=01}]

\node[] at (0, 1.15) [init,label=left:\textcolor{blue}{$\ell_0$}]         (0) {};
 \node[] at (0, -1.15) [init,node,label=left:\textcolor{blue}{$\ell_1$}]   (1) {};

 \node[] at (3.5, 0) [node,label=below:\textcolor{blue}{$\ell_2$}]        (2) {};
 \node[] at (6, 0) [final,label=below right:\textcolor{blue}{$\ell_3$}]    (3) {};

\draw[post] (0) to[]
    node[rotate=-17, align=center,anchor=east, midway,yshift=-.23cm,xshift=.75cm]
    {$r_2 \colon \gamma_1 \mapsto \cpp{x}$} (2);
\draw[post] (1) to[] node[rotate=17, anchor=north,yshift=.48cm,xshift=-.25cm]
    {$r_1 \colon \mathit{true} \mapsto \cpp{x}$~ ~}(2);
\draw[post] (0) -| node[anchor=south,yshift=-.45cm,xshift=1cm, pos=.25] (xpp)
    {$r_3 \colon \gamma_2 \mapsto \cpp{x}$} (3);
\draw[post] (2)to[]
    node[align=center,anchor=north, midway]
    {$r_4 \colon \gamma_2$} (3);
\draw[post] (1) -| node[anchor=north, pos=.25,yshift=.35cm,xshift=1cm] (xget)
    {$r_5 \colon \gamma_2 \mapsto \cpp{x}$} (3);
\draw[rule] (0) to[out=225,in=270,looseness=8]
	    node[align=center,anchor=east,midway]{$r_6$} (0);
\draw[rule] (2) to[out=15,in=60,looseness=8]
	    node[align=center,anchor=north,midway,xshift=.2cm,yshift=.3cm]{$r_7$} (2);
\draw[rule] (3) to[out=15,in=60,looseness=8]
	    node[align=center,anchor=north,midway,yshift=.4cm]{$r_8$} (3);

\end{tikzpicture}

 }
    \end{center}
\caption{The threshold automaton corresponding to Figure~\ref{fig:st} 
with $\gamma_1 \colon x \ge (t+1) - f$ and $\gamma_2 \colon x \ge
(n-t) - f$
 over parameters
     $n$, $t$, and $f$, representing the number of processes, the upper
     bound on the faulty processes (used in the code), and the actual
     number of faulty processes. The negative number $-f$ in the threshold is
     used to  model the environment, and captures that at most $f$ of
     the received messages may have been sent by faulty processes.
}
\label{fig:stunningexample}
\end{figure}

As internal representation in our tool, and in the theoretical work of
     this paper, we use \emph{threshold automata} (TA) defined
     in~\cite{KVW16:IandC}.
The TA that corresponds to the DISTAL code from  Figure~\ref{fig:st}
     is given in Figure~\ref{fig:stunningexample}.
The threshold automaton represents the local control flow of a single
     process, where arrows represent local transitions that are
     labeled with $\varphi \mapsto \mathsf{act}$: Expression~$\varphi$
     is a  threshold guard and the action $\mathsf{act}$ may increment
     a shared variable.

\begin{example}\label{ex:machina} The TA from
Figure~\ref{fig:stunningexample} is quite similar to the code in
     Figure~\ref{fig:st}: if \texttt{START} is called with $v=1$ this
     corresponds to the initial local state $\ell_1$, while otherwise
     a process starts in~$\ell_0$.
Initially a process has not sent any messages.
The local state $\ell_2$ in Figure~\ref{fig:stunningexample} 
	captures that the process has sent
     \texttt{EchoMsg} and  \texttt{accept} evaluates to false, 
     while $\ell_3$ captures that the
     process has sent \texttt{EchoMsg} and \texttt{accept} evaluates
     to true.
The syntax of Figure~\ref{fig:st}, although checking how
     many messages of some type are received, hides bookkeeping
     details and the environment, e.g., message buffers.
For our verification technique, we need to make such issues explicit:
     The shared variable $x$ stores the number of correct processes
     that have sent \texttt{EchoMsg}.
Incrementing~$x$ models that \texttt{EchoMsg} is sent when the
     transition is taken.
Then, execution of Line~\ref{line:init} corresponds to the transition
     $r_1$.
Executing Line~\ref{line:tp1} is captured by $r_2$: the check whether
     $t+1$ messages are received is captured by the fact that $r_2$
     has the guard $\gamma_1$, that is, $x \ge (t+1) - f$.
Intuitively, this guard checks whether sufficiently many processes
     have sent \texttt{EchoMsg} (i.e., increased~$x$), and  takes into
     account that at most~$f$ messages may have been sent by  
     faulty processes.
Namely, if we observe the guard in the equivalent form  $x + f\ge
     t+1$, then we notice that it evaluates to true when the total
     number of received \texttt{EchoMsg}  messages from correct
     processes ($x$) and potentially received messages from faulty processes (at most $f$),  is at least
     $t+1$, which corresponds to the guard of Line~\ref{line:tp1}.
Transition~$r_4$ corresponds to Line~\ref{line:nmt},  $r_3$
     captures that Line~\ref{line:init} and Line~\ref{line:nmt} are
     performed in one protocol step, and $r_5$ captures
     Line~\ref{line:tp1} and Line~\ref{line:nmt}.
\end{example}

While the example shows that the code in a domain-specific
     language and a TA are quite close, it should be noted that in
     reality, things are slightly more involved.
For instance, the DISTAL runtime takes care of the bookkeeping of sent
     and received messages (waiting queues at different network
     layers, buffers, etc.), and just triggers the high-level protocol
     when a threshold guard evaluates to true.
This typically requires counting the number of
     received messages.
While these local counters are present in the implementation, they
     are abstracted in the~TA.
For the purpose of this paper we do not need to get into the details.
Discussions on data abstraction and automated generation of TAs from
     code similar to DISTAL can be found in~\cite{KVW16:psi}.

We recall the necessary definitions introduced in~\cite{KVW16:IandC}.
A threshold automaton is a tuple~$\Sk = (\local, \initlocal, \globset,
     \paraset, \ruleset,\ResCond)$ whose components are defined as follows:
The  \emph{local states} and the  \emph{initial states} are in the
     finite sets $\local$ and $\initlocal\subseteq\local$,
     respectively. 
For simplicity, we identify local states with natural numbers, i.e., $\local= \{1, \dots, |\local| \}$.
\emph{Shared variables} and \emph{parameter variables} range
     over~$\NatZero$ and are in the finte sets $\globset$
     and~$\paraset$, respectively.
The  \emph{resilience condition}~$\ResCond$ is a formula over
     parameter variables in linear integer arithmetic, and the
     \emph{admissible parameters} are
     $\AdmP = \{ \param\in \NatZero^\numparam \colon \param \models
     \ResCond \}$.
After an example for resilience conditions, we will conclude the
definition of a threshold automaton by defining~$\ruleset$ as the finite set of rules.

\begin{example}\label{ex:rc}
The admissible parameters and  resilience conditions are motivated
     by fault-tolerant distributed algorithms: Let $n$
     be the number of processes, $t$ be the assumed number of faulty
     processes, and in a run, $f$ be the actual number of faults.
For these parameters, the famous result by Pease, Shostak and
     Lamport~\cite{LSP80} states that agreement can be solved iff the
     resilience condition $n > 3t \wedge t \ge f \ge 0$ is satisfied.
Given such constraints, the set $\AdmP$ is infinite, and in
     Section~\ref{sec:countsys} we will see that this results in an
     infinite state system.
\end{example}

A  \emph{rule} is a tuple $(\fromstate, \tostate, \precondLE,
     \precondG, \update)$, where $\fromstate$ and $\tostate$ are from
     $\local$, and capture from which local state to which a process
     moves via that rule.
A rule can only be executed if $\precondLE$ and $\precondG$ are true; both
are conjunction of guards.
Each guard consists of a shared variable $x \in \globset$,
     coefficients $a_0, \dots, a_\numparam \in \Int$, and parameter
     variables $p_1, \dots, p_\numparam \in \paraset$ so that $x \ge a_0 +
     \sum\nolimits_{i=1}^{\numparam} a_i \cdot p_i$ is
a \emph{lower guard}
 and
     $x < a_0 + \sum\nolimits_{i=1}^{\numparam} a_i \cdot p_i $ 
     is an \emph{upper guard}. Then,
$\PrecondU$ and $\PrecondL$ are the sets of lower and upper
guards.\footnote{Compared to~\cite{KVW15:CAV}, we use the more
  intuitive notation of $\PrecondU$ and $\PrecondL$: lower guards
  can only change from false to true (rising), while upper guards can only
  change from true to false (falling); cf.\ Proposition~\ref{prop:mono}.}
Rules may increase shared variables using an update vector $\update
     \in \NatZero^\numglob$ that is added to the vector of shared
     variables.
Finally,~$\ruleset$ is the finite set of rules.

\begin{example} A rule corresponds to an edge in
     Figure~\ref{fig:stunningexample}.
The pair $(\fromstate,\tostate)$ encodes the edge while
     $(\precondLE,\precondG,\update)$ encodes the edge label.
For example, rule~$r_2$ would be $(\ell_0,\ell_2,\gamma_1,\top,1)$.
Thus, a rule corresponds to a (guarded) statement from
     Figure~\ref{fig:st} (or combined statements as discussed in
     Example~\ref{ex:machina}).
\end{example}

The above definition of TAs is quite general.
It allows loops, increase of shared variables in loops, etc.
As has been observed in~\cite{KVW16:IandC}, if one does not restrict
     increases on shared variables, the resulting systems may produce
     runs that visit infinitely many states, and there is little hope for
     a complete verification method.
Hence, Konnov et al.~\cite{KVW15:CAV} analyzed the TAs of the
     benchmarks~\casest{}: They observed that some states have self-loops (corresponding to
     busy-waiting for messages to arrive) and in the case of failure
     detector based algorithms~\cite{Raynal97} there are loops that
     consist of  at most two rules.
None of the rules in loops increase shared variables. 
In our theory, we allow more general TAs than actually found in the
     benchmarks. In more detail, we make the following assumption: 

\paragraph{Threshold automata for fault-tolerant distributed algorithms.}

As in~\cite{KVW16:IandC}, we  assume that if a rule~$r$ is in a loop,
     then $r.\update = \vec{0}$.
In addition, we use the restriction that all the cycles of a TA are
     simple, i.e., between any two locations in a cycle there exists
     exactly one node-disjoint directed path (nodes in cycles may have
     self-loops).
We conjecture that this restriction can be relaxed as
     in~\cite{KVW15:CAV}, but this  is orthogonal to our~work.

\begin{example} In the TA from
     Figure~\ref{fig:stunningexample} we use the shared variable~$x$
     as the number of correct processes that have sent a message.
One easily observes that the rules that update~$x$ do not belong to
     loops.
Indeed, all the benchmarks~\casest{} share this structure.
This is because at the algorithmic level, all these algorithms are
based on the 
     reliable communication assumption (no message loss and no
     spurious message generation/duplication), and not much is gained by
     resending the same message.
In these  algorithms a process checks whether
     sufficiently many processes (e.g., a majority) have sent a
     message to signal that they are in some specific local state.
Consequently, a receiver would ignore duplicate messages from the same
     sender.
In our analysis we exploit this characteristic of distributed
     algorithms with threshold guards, and make the corresponding
     assumption that processes do not send (i.e., increase~$x$) from
     within a loop.
Similarly, as a process cannot make the sending of a message undone,
     we assume that shared variables are never decreased.
So, while we need these assumptions to derive our results, they are
     justified by our application domain.
\end{example}

\subsection{Counter Systems}
\label{sec:countsys}

A threshold automaton models a single process.
Now the question arises how we define the composition of multiple
     processes that will result in a distributed system.
Classically, this is done by parallel composition and interleaving
     semantics: A state of a distributed system that consists of $n$
     processes is modeled as $n$-dimensional vector of local states.
The transition to a successor state is then defined by
     non-deterministically picking a process, say $i$, and changing the
     $i$th component of the $n$-dimensional vector according to the
     local transition relation of the process.
However, for our domain of threshold-guarded algorithms, we do not
     care about the precise $n$-dimensional vector so that we use a
     more efficient encoding:  It is well-known that the system state
     of specific distributed or concurrent systems can be represented
     as a counter  system~\cite{Lub84,PXZ02,AlbertiGP16,KVW16:IandC}: 
instead of recording for some local state~$\ell$, which
     processes are in~$\ell$, we are only interested in \emph{how many
     processes are in~$\ell$}.
In this way, we can efficiently encode transition systems in SMT with
     linear integer arithmetics.
Therefore, we formalize the semantics of the threshold automata by
     counter systems.

Fix a threshold automaton $\Sk$, a function (expressible as linear combination
    of parameters) $\syssize\colon \AdmP \rightarrow \NatZero$ that determines
    the number of modeled processes, and admissible parameter values $\param
    \in \AdmP$.
A counter system $\Sys(\TA)$ is defined as a transition system
    $(\configs,\iconfigs,\transrel)$, with configurations $\configs$ and
    $\iconfigs$ and transition relation~$\transrel$ defined below.

\begin{definition}\label{def:config}
A configuration $\sigma=(\counters,\vars,\param)$ consists of a vector
     of \emph{counter values} $\sigma.\counters \in
     \NatZero^\numlocal$, a vector of
     \emph{shared variable values} $\sigma.\vars \in
     \NatZero^\numglob$, and a vector of \emph{parameter values}
     $\sigma.\param = \param$.
The set $\configs$ contains all configurations.
The initial configurations are in set~$\iconfigs$, and
     each initial configuration $\sigma$ satisfies
 $\sigma.\vars = \vec{0}$, $\sum_{i
     \in \initlocal}  \sigma.\counters[i] = \syssize(\param)$, and
     $\sum_{i \not\in \initlocal} \sigma.\counters[i] = 0$.
\end{definition}

\begin{example}\label{ex:config}
The safety property from Example~\ref{ex:rc},
     refers to an initial configuration that satisfies 
     resilience condition  $n > 3t \wedge t \ge f \ge 0$, e.g., $4 >
     3\cdot 1 \wedge 1 \ge 0 \ge 0$ such that $\sigma.\param =
     (4,1,0)$.
In our encodings we typically have
$\syssize$ is the function $(n,t,f) \mapsto n-f$.
Further, $\sigma.\counters[\ell_0] = \syssize(\param)= n-f = 4$ and
     $\sigma.\counters[\ell_i] = 0$, for $\ell_i
     \in\local\setminus\{\ell_0\}$, and the shared variable
     $\sigma.\vars=0$.
\end{example}

A \emph{transition} is a pair $t=(\trrule,\trfactor)$ of a rule  and a
     non-negative integer called the \emph{acceleration factor}.
For $t=(\trrule,\trfactor)$ we write $t.\update$  for
     $\trrule.\update$, etc. A transition $t$ is \emph{unlocked}
     in~$\sigma$ if $ \forall k \in \{0, \dots, t.\trfactor - 1 \}.\;
     (\sigma.\counters,\sigma.\vars + k\cdot t.\update, \sigma.\param)
     \models t.\precondLE \wedge t.\precondG.\nonumber $ A transition
     $t$ is \emph{applicable (or enabled)} in~$\sigma$, if it is
     unlocked, and $\sigma.\counters[t.\fromstate] \ge t.\trfactor$,
     or $t.\trfactor=0$.

\begin{example}\label{ex:accel}
This notion of applicability contains acceleration and is central for
     our approach.
Intuitively, the value of the factor corresponds to how
many times the rule is executed by different processes. In this way, we
can subsume steps by an arbitrary number of
     processes into one transition.
Consider Figure~\ref{fig:stunningexample}.
If for some $k$, $k$ processes are in location $\ell_1$, then in
     classic modeling it takes $k$ transitions to move these processes
     one-by-one to~$\ell_2$.
With acceleration, however, these $k$ processes can
     be moved to $\ell_2$ in one step, independently of $k$.
In this way, the bounds we compute will be independent of the
     parameter values.
However, assuming $x$ to be a shared variable and $f$ being a
     parameter that captures the number of faults, our
     (crash-tolerant) benchmarks include rules like ``$x<f \mapsto
     \cpp{x}$'' for local transition to a special ``crashed'' state.
The above definition ensures that at most $f-x$
     of these transitions are accelerated into one transition (whose
     factor thus is at most $f-x$).
This precise treatment of threshold guards
     is crucial for fault-tolerant distributed
     algorithms.
The central contribution of this paper is to show how acceleration can
     be used to shorten schedules while maintaining specific temporal
     logic properties.
\end{example}

\begin{definition}\label{def:TofSigma}
The configuration $\sigma'$ is the result of applying the enabled
     transition $t$ to $\sigma$, if 
\begin{enumerate}
\item $\sigma'.\vars = \sigma.\vars + t.\trfactor \cdot t.\update$
\item $\sigma'.\param = \sigma.\param$
\item if $t.\fromstate \ne t.\tostate$ then
$\sigma'.\counters[t.\fromstate]=\sigma.\counters[t.\fromstate]
  - t.\trfactor$,
 $\sigma'.\counters[t.\tostate]=\sigma.\counters[t.\tostate] +
  t.\trfactor$, and\\
 $\forall \ell \in \local \setminus \{t.\fromstate, t.\tostate \}.\;
\sigma'.\counters[\ell]=\sigma.\counters[\ell]$.

\item if $t.\fromstate =  t.\tostate$ then $\sigma'.\counters =
  \sigma.\counters.$
\end{enumerate}
In this case we use the notation $\sigma' =
     t(\sigma)$.
\end{definition}

\begin{example} \label{ex:guards}  Let us again consider
Figure~\ref{fig:stunningexample} with $n=4$, $t=1$, and $f=1$.
We consider the initial configuration where $\sigma.\counters[\ell_1]
     = n-f = 3$ and $\sigma.\counters[\ell_i] = 0$, for $\ell_i
     \in\local\setminus\{\ell_0\}$.
The guard of rule $r_5$, $\gamma_2 \colon x \ge (n-t) - f = 2$,
     initially evaluates to false because $x=0$.
The guard of rule $r_1$ is true, so that any transition $(r_1,
     \trfactor)$ is unlocked.
As $\sigma.\counters[\ell_1] =  3$, all transitions $(r_1,
     \trfactor)$, for $0\le\trfactor\le 3$ are applicable.
If the transition $(r_1,2)$ is applied to the initial configuration,
     we obtain that $x=2$ so that, after the application, $\gamma_2$
     evaluates to true.
Then $r_5$ is unlocked and the transitions $(r_5,1)$ and $(r_5,
     0)$ are applicable as $\sigma.\counters[\ell_1]  =  1$.
Since  $\gamma_2$ checks for greater or equal,
     once it becomes true it remains true.
Such monotonic behavior is given for all guards, as has already been
     observed in~\cite[Proposition 7]{KVW16:IandC}, and is a crucial property.
\end{example}

The transition relation $\transrel$
     is defined as follows: Transition $(\sigma, \sigma')$ belongs to
     $\transrel$ iff there is a rule $r\in\rules$ and a factor $k\in\Natural_0$
     such that $\sigma' = t(\sigma)$ for $t=(r,k)$.
A \emph{schedule} is a sequence of
     transitions.
For a schedule~$\tau$ and an index~$i: 1 \le i \le |\tau|$, by $\tau[i]$
     we denote the $i$th transition of $\tau$, and by $\tau^i$ we
     denote the prefix $\tau[1], \dots, \tau[i]$ of~$\tau$.
A schedule $\tau = t_1, \dots, t_m$ is \emph{applicable} to
     configuration $\sigma_0$, if there is a sequence of
     configurations $\sigma_1,\dots, \sigma_m$ with $\sigma_i =
     t_{i} (\sigma_{i-1})$ for $1 \le i \le m$.
A schedule $t_1, \dots, t_m$ where  $t_i.\trfactor=1$ for $0< i\le m$
     is called \emph{conventional}.
If there is a $t_i.\trfactor>1$, then a schedule is \emph{accelerated}.
By $\tau \concat \tau'$ we denote the concatenation
     of two schedules $\tau$ and~$\tau'$.

We will reason about schedules in Section~\ref{sec:repr} for our mover
     analysis, which is naturally expressed by swapping neighboring
     transitions in a schedule.
To reason about temporal logic properties, we need to reason about the
     configurations that are ``visited'' by a schedule.
For that we now introduce paths.

A finite or infinite sequence $\gst_0, t_1, \gst_1, \dots, t_{k-1}, \gst_{k-1},
     t_k, \dots$ of alternating configurations and transitions is called a
     \emph{path}, if for every transition~$t_i$, $i\in\Natural$, in 
     the sequence, holds that~$t_i$ is enabled in~$\gst_{i-1}$, and
     $\gst_i = t_i(\gst_{i-1})$.
For a configuration~$\gst_0$ and a finite schedule~$\tau$ applicable
     to~$\gst_0$, by $\finpath{\gst_0}{\tau}$ we denote $\gst_0, t_1, \gst_1,
     \dots, t_{|\tau|}, \gst_{|\tau|}$ with 
     $\gst_i = t_i(\gst_{i-1})$, for $1 \le i \le |\tau|$.
Similarly, if $\tau$ is an infinite schedule applicable to~$\gst_0$,
      then $\infpath{\gst_0}{\tau}$ represents an infinite sequence
      $\gst_0, t_1, \gst_1, \dots, t_{k-1}, \gst_{k-1}, t_k, \dots$
      where $\gst_i = t_i(\gst_{i-1})$, 
        for all $i > 0$.

The evaluation of the threshold guards solely
     defines whether certain rules are unlocked.
As was discussed in
     Example~\ref{ex:guards}, along a path, the evaluations of guards
     are monotonic.
The set of upper guards that evaluate to false and lower guards that 
	evaluate to true\dash---called
     the context\dash---changes  only finitely many times.
A schedule can thus be understood as an alternating  sequence of
     schedules without context change, and context-changing transitions.
We will recall the definitions of context etc.\ from~\cite{KVW15:CAV} 
in Section~\ref{sec:structguards}. We
     say that a schedule $\tau$ is \emph{steady}  for a configuration
     $\gst$, if every configuration of $\finpath{\gst}{\tau}$ has
     the same context.

Due to the resilience conditions and admissible parameters, our counter
     systems are in general infinite state.
The following proposition establishes an important property for
     verification.

\begin{proposition}\label{prop:finite}
Every (finite or infinite) path visits finitely many configurations.
\end{proposition}
\begin{proof}
By Definition~\ref{def:TofSigma}(3), if a transition $t$ is applied to a
configuration $\sigma$, then the sum of the counters remains unchanged,
that is,
$\sum_{\ell\in\local} \sigma.\counters[\ell] = \sum_{\ell\in\local} 
t(\sigma).\counters[\ell]$. By
repeating this argument, the sum of the counters remains stable in a
path. By Definition~\ref{def:TofSigma}(2) the parameter values also
remain stable in a path. 

By Definition~\ref{def:TofSigma}(1), it remains to show that in each
     path eventually the shared variable $\vars$ stop increasing.
Let us fix a rule $r=(\fromstate, \tostate, \precondLE, \precondG,
     \update)$ that increases $\vars$. By the definition of a
     transition, applying some transition $(r,\trfactor)$ decreases
     $\counters[r.\fromstate]$ by $\trfactor$.
As by assumption on TAs,
     $r$ is not in a cycle, $\counters[r.\fromstate]$ is increased only
     finitely often, namely, at most $\syssize(\param)$
     times.
As there are only finitely many rules in a TA, the proposition
     follows.
\end{proof}

\makeatletter{}

\section{Verification Problems:
    Parameterized Reachability vs.  Safety \& Liveness.}
\label{sec:reach-and-live}

In this section we will discuss the verification problems for
fault-tolerant distributed algorithms. A central challenge is to handle
resilience conditions precisely.

\begin{example} \label{ex:reach}
The safety property (unforgeability)
     of~\cite{ST87:abc} expressed in terms of
     Figure~\ref{fig:stunningexample} means that no process should
     ever enter~$\ell_3$ if initially all processes are in~$\ell_0$,
    given that $n > 3t \wedge t \ge f \ge 0$.
We can
     express this in the counter system: under the
     resilience condition $n > 3t \wedge t \ge f \ge 0$,  given an
     initial configuration~$\sigma$,  with $\sigma.\counters[\ell_0] =
     n-f$, to verify safety, we have to establish the absence of a
     schedule $\tau$ that satisfies $\sigma'=\tau(\sigma)$ and
     $\sigma'.\counters[\ell_3] > 0$.

In order to be able to answer this question, we have to deal with these
     resilience conditions  precisely: Observe that $\ell_3$ is
     unreachable, as all outgoing transitions from $\ell_0$ contain
     guards that evaluate to false initially, and since all processes
     are in $\ell_0$ no process ever increases~$x$.
A slight modification of $t \ge f$ to $t+1 \ge f$ in the resilience
     condition changes the result, i.e., one fault too many breaks the
     system.
For example, if $n=4$, $t=1$, and $f=2$, then the new resilience
     condition holds, but as the guard~$\gamma_1: x \ge (t+1) -f$ is
     now initially true, then one correct process can fire the
     rule~$r_2$ and increase~$x$.
Now when $x=1$, the guard $\gamma_2: x\ge (n-t) -f$ becomes true, so
     that the process can fire the rule~$r_4$ and reach the
     state~$\ell_3$.
This tells us that unforgeability is not satisfied in the system where
     the resilience  condition is $n > 3t \wedge t+1 \ge f \ge 0$.
\end{example}

This is the verification question studied in~\cite{KVW15:CAV},
which can be formalized as follows:

\begin{definition}[Parameterized reachability]
Given a threshold automaton~$\TA$ and a Boolean formula~$B$ over
    $\{\counters[i] = 0 \mid i \in \local\}$,
    check whether there are parameter values~$\param \in \AdmP$,
    an initial configuration $\gst_0 \in \iconfigs$ with~$\gst_0.\param =
    \param$ and a finite schedule~$\tau$ applicable to~$\gst_0$ such
    that~$\tau(\gst_0) \models B$.
\end{definition}

As shown in~\cite{KVW15:CAV}, if such a schedule exists, then
     there is also a schedule of bounded length.
In this paper, we do not limit ourselves to reachability, but consider
 specifications of
\emph{counterexamples to safety and liveness} of FTDAs from the
    literature. 
We observe that such specifications use a simple subset of linear temporal
    logic that contains only the temporal operators~$\ltlF$ and~$\ltlG$.

\begin{example}\label{ex:fair}
Consider a liveness property from the distributed algorithms
literature
    called correctness~\cite{ST87:abc}:
\begin{equation}\label{eq:correctness}
     \ltlG\ltlF \propfair \rightarrow (\counters[\ell_0] = 0 \rightarrow
     \ltlF \counters[\ell_3] \ne 0).
\end{equation}
Formula~$\propfair$ expresses the reliable communication
     assumption of distributed algorithms~\cite{FLP85}.
In this example, $\propfair \equiv \counters[\ell_1] = 0 \wedge  (x\ge
     t+1 \rightarrow \counters[\ell_0] = 0 \wedge \counters[\ell_1] =
     0) \wedge (x\ge n-t \rightarrow \counters[\ell_0] = 0 \wedge
     \counters[\ell_2] = 0)$.
Intuitively, $\ltlG\ltlF \propfair$ means that all processes in
     $\ell_1$ should  eventually leave this state, and if sufficiently
     many messages of type~$x$ are sent ($\gamma_1$ or $\gamma_2$
     holds true), then all processes eventually receive them.
If they do so, they have to eventually fire rules~$r_1$,~$r_2$,~$r_3$,
     or~$r_4$ and thus leave locations~$\ell_0$,~$\ell_1$,
     and~$\ell_2$.
Our approach is based on possible shapes of  \emph{counterexamples}.
Therefore, we consider the negation of the
specification~(\ref{eq:correctness}), that is, $
  \ltlG\ltlF \propfair \wedge  \counters[\ell_0] = 0
    \wedge \ltlG \counters[\ell_3] = 0$. In the following we define
    the logic that can express such counterexamples.
\end{example}

The fragment of LTL limited to~$\ltlF$ and~$\ltlG$ was studied
    in~\cite{EtessamiVW02,KroeningOSWW11}. 
We further restrict it to the logic that we call \emph{Fault-Tolerant
    Temporal Logic} ($\ELTLTB$), whose 
syntax is shown in Table~\ref{table:eltlft-syntax}. 
The formulas derived from \emph{cform}\dash---called counter
    formulas\dash---restrict counters, while the formulas
    derived from \emph{gform}\dash---called guard formulas\dash---restrict
    shared variables. 
The formulas derived from \emph{pform} are propositional formulas. 
The temporal operators $\ltlF$ and $\ltlG$ follow the standard
    semantics~\cite{CGP1999,BK08}, that is, for a configuration~$\gst$ and an
    infinite
    schedule~$\tau$, it holds that ~$\infpath{\gst}{\tau} \models \varphi$, if:

\begin{enumerate}
    \item $\gst \models \varphi$, when $\varphi$ is a propositional formula,

    \item $\exists \tau',\tau'':
        \tau = \tau' \concat \tau''.\ \infpath{\tau'(\gst)}{\tau''} \models \psi$, when $\varphi=\ltlF \psi$,

    \item $\forall \tau',\tau'':
        \tau = \tau' \concat \tau''.\ \infpath{\tau'(\gst)}{\tau''} \models \psi$, when $\varphi=\ltlG \psi$.
\end{enumerate}

\begin{table}
  \begin{align*}
    \psi &::=
    \mathit{pform} \ |\ \ltlG \psi\ |\ \ltlF \psi \ |\  \psi \wedge \psi\\
    \mathit{pform} &::= \mathit{cform} \ |\ \mathit{gform} \vee \mathit{cform}\\
    \mathit{cform} &::= \bigvee_{\ell \in \mathit{Locs}} \counters[\ell] \ne
    0 \ |\ \bigwedge_{\ell \in \mathit{Locs}} \counters[\ell] = 0 \ |\ \mathit{cform} \wedge
    \mathit{cform} \\
    \mathit{gform} &::= \mathit{guard}
    \ |\ \neg \mathit{gform} \ |\ \mathit{gform} \wedge \mathit{gform}
  \end{align*}

\caption{The syntax of~$\ELTLTB$-formulas: $\mathit{pform}$ defines
        propositional formulas, and $\psi$ defines temporal formulas. 
We assume that $\mathit{Locs} \subseteq \local$ and $\mathit{guard} \in
    \PrecondU \cup \PrecondL$.} \label{table:eltlft-syntax}
\end{table}

To stress that the formula should be satisfied by \emph{at least
    one path}, we prepend $\ELTLTB$-formulas with the existential path
    quantifier~$\ltlE$.
We use the shorthand notation~$\mathit{true}$ for a valid
    propositional formula, e.g.,~$\bigwedge_{i \in \emptyset} \counters[i] =
    0$.
We also denote with~$\ELTLTB$ the set of all formulas that can be
    written using the logic~$\ELTLTB$.

We will reason about invariants of the finite
subschedules, and
 consider a propositional formula~$\psi$.
Given a configuration~$\gst$, a finite schedule~$\tau$ applicable
     to~$\sigma$, and~$\psi$, by 
     $\setconf{\gst}{\tau}\models \psi$
     we denote
     that~$\psi$ holds in every configuration $\sigma'$ visited by the
     path~$\finpath{\gst}{\tau}$.
In other words, for every prefix~$\tau'$ of~$\tau$, we have that 
     $\tau'(\gst)\models \psi$.

\begin{definition}[Parameterized unsafety \& non-liveness]\label{def:pmcp}
Given a threshold automaton~$\TA$ and an $\ELTLTB{}$ formula $\psi$, check
    whether there are parameter values $\param \in \AdmP$, an initial
    configuration $\gst_0 \in \iconfigs$ with $\gst_0.\param = \param$, and an
    infinite schedule~$\tau$ of~$\Sys(\TA)$ applicable to~$\gst_0$ such that
    $\infpath{\gst_0}{\tau} \models \psi$.
\end{definition}

\paragraph{Complete bounded model checking.}

We solve this problem  by showing how to
     reduce it to bounded model checking while guaranteeing
     completeness.
To this end, we have to construct a bounded-length encoding of
     infinite schedules. In more detail:

\begin{itemize}
\item We observe that if  $\infpath{\gst_0}{\tau} \models \psi$,
     then  there is an initial state~$\gst$ and two finite schedules
     $\vartheta$ and $\rho$ (of unknown length) that can be used to
     construct an infinite (lasso-shaped) schedule $\vartheta \cdot
     \rho^\omega$, such that $\infpath{\gst}{\vartheta \cdot \rho^\omega}
     \models \psi$ (Section~\ref{sec:lasso}).

\item Now given $\vartheta$ and $\rho$, we prove that we can use a
        $\psi$-specific reduction, to cut $\vartheta$ and $\rho$ into
        subschedules $\vartheta_1,\dots,\vartheta_m$ and $\rho_1,\dots,\rho_n$,
        respectively so that the subschedules satisfy subformulas of $\psi$
        (Sections~\ref{sec:shapelasso}, \ref{sec:enumerating-lassos}
        and~\ref{sec:structguards}).

\item We use an offline partial order reduction, specific to the
  subformulas of $\psi$,
 and acceleration to construct
     representative schedules $\jrep{\vartheta_i}$ and
     $\jrep{\rho_j}$  that satisfy the required $\ELTLTB$
     formulas that are satisfied $\vartheta_i$ and~$\rho_j$, respectively for $1 \le i
     \le m$ and $1 \le j \le n$.
Moreover,  $\jrep{\vartheta_i}$ and $\jrep{\rho_j}$ are 
     fixed sequences of rules, where bounds on the lengths of the
     sequences are known (Section~\ref{sec:repr}).

\item These fixed sequence of rules can be used to encode a query to
  the SMT solver (Section~\ref{sec:smtencodings}). We ask whether there is an applicable schedule in
  the counter system that
  satisfies the sequence of rules and~$\psi$ (Section~\ref{sec:subexp}). If the SMT solver
  reports a contradiction, there exists no counterexample.
\end{itemize}

Based on these theoretical results, our tool implements the 
     high-level verification algorithm from Figure~\ref{fig:pseudo} (in the comments we give the
     sections that are concerned with the respective step): 

\begin{figure}
\begin{lstlisting}[language=pseudo,numbers=none,columns=fullflexible]
algorithm parameterized_model_checking($\TA$, $\varphi$): // see Def. @\ref{def:pmcp}@
 $\pgraph$ := cut_graph ($\varphi$)         /* Sect. 4 */
 $\tgraph$ := threshold_graph($\TA$) /* Sect. 5 */
 for each $\prec$ in topological_orderings($\pgraph \cup \tgraph$) do // e.g., using @\cite{Canfield1995}@
  check_one_order($\TA$, $\varphi$, $\pgraph$, $\tgraph$, $\prec$) /* Sect. 6-7 */
  if SMT_sat() then report the SMT model as a counterexample
\end{lstlisting}
\caption{A high-level description of the verification algorithm. For details
of \texttt{check\_one\_order}, see Section~\ref{sec:one-order} and Figure~\ref{fig:pseudosmt}.}
\label{fig:pseudo}
\end{figure}

\makeatletter{}\section{Shapes of Schedules that Satisfy~$\ELTLTB$}\label{sec:counterexamples}

We characterize all possible shapes of
     lasso schedules that satisfy an
     $\ELTLTB{}$-formula~$\varphi$.
These shapes are characterized by so-called  \emph{cut points}:  We
     show that every lasso satisfying~$\varphi$ has a fixed number of
     cut points, one cut point per a subformula of~$\varphi$ that
     starts with~$\ltlF$.
The configuration in the cut point of a subformula~$\ltlF \psi$ must
     satisfy~$\psi$, and all configurations between two cut points
     must satisfy certain propositional formulas, which are extracted
     from the subformulas of~$\varphi$ that start with~$\ltlG$.
Our notion of a cut point is motivated by extreme appearances of
     temporal operators~\cite{EtessamiVW02}.

\begin{figure}
    \begin{center}
        \makeatletter{}\begin{tikzpicture}[x=1cm,y=1cm,font=\scriptsize,>=latex];
    \tikzstyle{node}=[circle,fill=black,minimum size=0.1cm,inner sep=0cm];
    \tikzstyle{cut}=[cross out,thick,draw=red!90!black,
        minimum size=0.15cm,inner sep=0mm,outer sep=.1mm];
    \tikzstyle{path}=[-];
    \tikzstyle{Gfin}=[-, very thick, blue];

  \draw[rounded corners] (-.2,-.8) rectangle (8.2, 1);
    
  \begin{scope}[xshift=0cm, yshift=0cm]
    \node at (0.1, .7) { \normalsize\textbf{(a)} };

    \node[node] (0) at (0, 0) {};

            \foreach \x/\n in {.6/A, 1.2/B, 1.8/C, 2.4/D, 3.0/E, 3.6/F}
        \node[cut,label={above:$\n$}] (\n) at (\x,0) {};

    \draw[path] (0) -- (A);
    \draw[path] (A) -- (B);
    \draw[path] (B) -- (C);
    \draw[path] (C) -- (D);
    \draw[Gfin] (D) -- (E);
    \draw[Gfin] (E) -- (F);
    \draw[->] (F) edge[bend right=55,looseness=1.7] (D);

    \draw[|<->|] ($(A)+(0,-.5)$)
        --node[midway, fill=white, text=black]
        {$b$} ($(F)+(0,-.5)$);
    \node at ($(A)+(0,-.25)$) {$a$};
    \node at ($(B)+(0,-.25)$) {$d$};
    \node at ($(C)+(0,-.25)$) {$e$};
    \node at ($(E)+(0,-.25)$) {$c$};
  \end{scope}
    
  \begin{scope}[xshift=4.3cm, yshift=0cm]
    \node at (0.1, .7) { \normalsize\textbf{(b)} };

    \node[node] (0) at (0, 0) {};

            \foreach \x/\n in {.6/A, 1.8/B, 1.2/C, 2.4/D, 3.0/E, 3.6/F}
        \node[cut,label={above:$\n$}] (\n) at (\x,0) {};

    \draw[path] (0) -- (A);
    \draw[path] (A) -- (C);
    \draw[path] (C) -- (B);
    \draw[path] (B) -- (D);
    \draw[Gfin] (D) -- (E);
    \draw[Gfin] (E) -- (F);
    \draw[->] (F) edge[bend right=55,looseness=1.7] (D);

    \draw[|<->|] ($(A)+(0,-.5)$)
        --node[midway, fill=white, text=black]
        {$b$} ($(F)+(0,-.5)$);
    \node at ($(A)+(0,-.25)$) {$a$};
    \node at ($(B)+(0,-.25)$) {$d$};
    \node at ($(C)+(0,-.25)$) {$e$};
    \node at ($(E)+(0,-.25)$) {$c$};
  \end{scope}
    
  \begin{scope}[xshift=0cm, yshift=-2cm]
    \node at (0.1, .7) { \normalsize\textbf{(c)} };

    \node[node] (0) at (0, 0) {};

            \foreach \x/\n in {1.2/A, 1.8/B, 2.4/C, .6/D, 3.0/E, 3.6/F}
        \node[cut,label={above:$\n$}] (\n) at (\x,0) {};

    \draw[path] (0) -- (D);
    \draw[Gfin] (D) -- (A);
    \draw[Gfin] (A) -- (B);
    \draw[Gfin] (B) -- (C);
    \draw[Gfin] (C) -- (E);
    \draw[Gfin] (E) -- (F);
    \draw[->] (F) edge[bend right=55] (D);

    \draw[|<->|] ($(D)+(0,-.5)$)
        --node[midway, fill=white, text=black]
        {$b$} ($(F)+(0,-.5)$);
    \node at ($(A)+(0,-.25)$) {$a$};
    \node at ($(B)+(0,-.25)$) {$d$};
    \node at ($(C)+(0,-.25)$) {$e$};
    \node at ($(E)+(0,-.25)$) {$c$};
  \end{scope}
    
  \begin{scope}[xshift=4.3cm, yshift=-2cm]
    \node at (0.1, .7) { \normalsize\textbf{(d)} };

    \node[node] (0) at (0, 0) {};

            \foreach \x/\n in {3.0/A, 2.4/B, 1.8/C, .6/D, 1.2/E, 3.6/F}
        \node[cut,label={above:$\n$}] (\n) at (\x,0) {};

    \draw[path] (0) -- (D);
    \draw[Gfin] (D) -- (E);
    \draw[Gfin] (E) -- (C);
    \draw[Gfin] (C) -- (B);
    \draw[Gfin] (B) -- (A);
    \draw[Gfin] (A) -- (F);
    \draw[->] (F) edge[bend right=55] (D);

    \draw[|<->|] ($(D)+(0,-.5)$)
        --node[midway, fill=white, text=black]
        {$b$} ($(F)+(0,-.5)$);
    \node at ($(A)+(0,-.25)$) {$a$};
    \node at ($(B)+(0,-.25)$) {$d$};
    \node at ($(C)+(0,-.25)$) {$e$};
    \node at ($(E)+(0,-.25)$) {$c$};
  \end{scope}

  \begin{scope}[xshift=0cm, yshift=-4cm]
    \node at (0.1, .7) { \normalsize\textbf{(e)} };

    \node[node] (0) at (0, 0) {};

            \foreach \x/\n in {.6/A, 2.4/B, 1.8/C, 1.2/D, 3.0/E, 3.6/F}
        \node[cut,label={above:$\n$}] (\n) at (\x,0) {};

    \draw[path] (0) -- (A);
    \draw[path] (A) -- (D);
    \draw[Gfin] (D) -- (C);
    \draw[Gfin] (C) -- (B);
    \draw[Gfin] (B) -- (E);
    \draw[Gfin] (E) -- (F);
    \draw[->] (F) edge[bend right=55] (D);

    \draw[|<->|] ($(A)+(0,-.5)$)
        --node[midway, fill=white, text=black]
        {$b$} ($(F)+(0,-.5)$);
    \node at ($(A)+(0,-.25)$) {$a$};
    \node at ($(B)+(0,-.25)$) {$d$};
    \node at ($(C)+(0,-.25)$) {$e$};
    \node at ($(E)+(0,-.25)$) {$c$};
  \end{scope}
 
  \begin{scope}[xshift=4.3cm, yshift=-4cm]
    \node at (2, 0) { \normalsize{(and 15 more...)} };
  \end{scope}

\end{tikzpicture}

    \end{center}

    \caption{The shapes of lassos that satisfy
            the formula $\ltlE \ltlF (a \wedge \ltlF d \wedge \ltlF e
        \wedge \ltlG b \wedge \ltlG \ltlF c)$. The crosses show cut points for:
    (A)~formula $\ltlF (a \wedge \ltlF d \wedge \ltlF e
        \wedge \ltlG b \wedge \ltlG \ltlF c)$,
            (B)~formula $\ltlF d$, (C)~formula $\ltlF e$,
            (D)~loop start, (E)~formula~$\ltlF c$, and
            (F)~loop end.
    }
    \label{fig:lasso-shapes}
\end{figure}

\begin{example}\label{ex:many-shapes-of-lassos}
Consider the $\ELTLTB$ formula
    $\varphi \equiv \ltlE \ltlF (a \wedge \ltlF d \wedge \ltlF e
            \wedge \ltlG b \wedge \ltlG \ltlF c)$, where $a, \dots, e$
    are propositional
    formulas, whose structure is not of interest in this section.
Formula~$\varphi$ is satisfiable by certain paths that have lasso shapes,
    i.e., a path consists of a finite prefix and a loop, which is repeated
    infinitely.
These lassos may differ in the actual occurrences of the propositions
and the start of the loop: For instance, at some point,~$a$ holds,
    and since then~$b$ always holds, then~$d$ holds at some point,
    then~$e$ holds at some point,
    then the loop is entered, and $c$ holds infinitely often inside the loop.
This is the case~(a) shown in Figure~\ref{fig:lasso-shapes}, where the
    configurations in the cut points~$A$, $B$, $C$, and~$D$ must satisfy
    the propositional formulas~$a$, $d$, $e$, and~$c$ respectively,
    and the configurations between~$A$ and $F$ must satisfy the propositional
    formula~$b$. 
This example does not restrict the propositions between the initial state and
    the cut point~A, so
    that this lasso shape, for instance, also captures the path where $b$ holds
    from the beginning. 
There are 20 different lasso shapes for~$\varphi$, five of them are
    shown in the figure.
We construct lasso shapes that are sufficient
    for finding a path satisfying an $\ELTLTB$ formula.
In this example, it is sufficient to consider lasso shapes~(a) and~(b),
    since the other shapes can be constructed from~(a) and~(b) by unrolling
    the loop several times.
\end{example}

\subsection{Restricting Schedules to Lassos}
\label{sec:lasso}

In the seminal paper~\cite{VW86}, Vardi and Wolper showed that if a
    finite-state transition system~$M$ \emph{violates} an~\LTL{}
    formula\dash---which requires \emph{all paths} to satisfy the
    formula\dash---then there is a path in~$M$ that (i)~violates the formula
    and (ii)~has lasso shape. 
As our logic~$\ELTLTB{}$ specifies counterexamples to the properties of
    fault-tolerant distributed algorithms, we are interested in this result in
    the following form: if the transition system \emph{satisfies} an~$\ELTL{}$
    formula\dash---which requires \emph{one path} to satisfy the
    formula\dash---then~$M$ has a path that (i)~\emph{satisfies} the formula
    and (ii)~has lasso shape.

As observed above, counter systems are infinite state.
Consequently, one cannot apply the results of~\cite{VW86} directly.
However,  using Proposition~\ref{prop:finite}, we show that a similar
     result holds for counter systems of threshold automata
     and~$\ELTLTB$:

\newcommand{\proplassoscheda}{Given a threshold automaton~$\TA$ and an $\ELTLTB{}$ formula~$\varphi$, if
    $\Sys(\TA) \models \ltlE \varphi$, then there are an initial
    configuration~$\gst_1 \in \iconfigs$ and a schedule $\tau \concat
    \rho^\omega$ with the following properties:
        \begin{enumerate}
        \item the path satisfies the formula:
                $\infpath{\gst_1}{\tau \concat \rho^\omega} \models \varphi$,
        \item application of~$\rho$ forms a cycle:
                $\rho^k(\tau(\gst_1)) = \tau(\gst_1)$ for~$k \ge 0$.
    \end{enumerate}}

\begin{proposition}\label{prop:lasso-sched}
\proplassoscheda
\end{proposition}

Although in \cite{KLVW16:arxiv} we use B\"uchi automata to prove
     Proposition~\ref{prop:lasso-sched}, we do not use B\"uchi
     automata in this paper.
Since $\ELTLTB$ uses only the temporal operators $\ltlF$ and $\ltlG$,
     we found it much easier to reason about the structure of
     $\ELTLTB{}$ formulas directly (in the spirit
     of~\cite{EtessamiVW02}) and then apply path reductions, rather
     than constructing the synchronous product of a B\"uchi automaton
     and of a counter system and then finding proper path reductions.

Although Proposition~\ref{prop:lasso-sched} guarantees
     counterexamples of lasso shape, it is not sufficient
     for model checking: (i) counter systems are infinite state, so
     that state enumeration may not terminate, and (ii)
     Proposition~\ref{prop:lasso-sched} does not provide us with
     bounds on the length of the lassos needed for bounded
     model checking.
In the next section, we show how to split a lasso schedule in finite
     segments and to find constraints on lasso schedules that satisfy
     an $\ELTLTB{}$ formula.
In Section~\ref{sec:repr} we then construct shorter (bounded length)
     segments.

\subsection{Characterizing Shapes of Lasso Schedules}
\label{sec:shapelasso}

    \begin{figure}
        \centering
        \makeatletter{}        \begin{tikzpicture}[font=\small,
                level distance=7mm,
                level 1/.style={sibling distance=30mm},
                level 2/.style={sibling distance=17mm, level distance=11mm},
                level 3/.style={sibling distance=12mm}]
            \tikzstyle{form}=[rectangle,draw,rounded corners,
                minimum height=.5cm];

            \node[form,label={[xshift=0mm]east:[0]}]
                {$\canform(\varphi)$ }
            child[sibling distance=15mm,
                edge from parent path={(\tikzparentnode\tikzparentanchor)
                    edge [bend right]
                (\tikzchildnode\tikzchildanchor)}] {
                node[form,label={[xshift=-5mm]south east:[0.0]}] {$a$}
                            }
            child[sibling distance=20mm] {
                    node[form,label={[xshift=-7mm]south east:[0.1]}]
                    {$\ltlF (d \wedge \ltlG \true)$}
                child[sibling distance=20mm] {
                    node[form,label={south:[0.1.0]}] {$d$}
                }
                child[sibling distance=2mm] {
                    node[form,label={[xshift=0mm]south:[0.1.1]}] {$\ltlG \true$}
                }
            }
        child[sibling distance=17mm] {
            node[form,label={[xshift=-6mm]south east:[0.2]}]
                {$\ltlF (e \wedge \dots)$ }
            child[sibling distance=12mm] {
                node[form,label={[xshift=-5mm]south east:[0.2.0]}] {$e$}
            }
            child[sibling distance=6mm] {
                node[form,label={[xshift=0mm]south:[0.2.1]}] {$\ltlG \true$}
            }
        }
            child[sibling distance=25mm] {
                node[form,label={[xshift=-6mm]north east:[0.3]}]
                    {$\ltlG (b \wedge \ltlF (c \wedge \ltlG \true)
                        \wedge \ltlG \true)$}
                child[sibling distance=15mm] {
                    node[form,label={south:[0.3.0]}] {$b$}
                }
                child[sibling distance=14mm] {
                    node[form,label={[xshift=8mm]south:[0.3.1]}]
                        {$\ltlF (c \wedge \ltlG \true)$}
                    child { node[form] {$c$} }                     child {
                        node[form] {$\ltlG \true$}                     }
                }
                child[sibling distance=16mm] {
                    node[form,label={south:[0.3.2]}] {$\ltlG \true$}
                }
            };
                                            \end{tikzpicture}

        \caption{A canonical syntax tree of the $\ELTLTB$ formula
        $\varphi \equiv \ltlF (a \wedge \ltlF d \wedge \ltlF e
        \wedge \ltlG b \wedge \ltlG \ltlF c)$
        considered in Example~\ref{ex:many-shapes-of-lassos}.
        The labels $[w]$ denote identifiers of the tree nodes.
                }

        \label{Fig:syntax-tree}
    \end{figure}

We now construct a cut graph of an $\ELTLTB$ formula:
Cut graphs constrain the orders in which subformulas that start with the
     operator~$\ltlF$ are witnessed by configurations.
The nodes of a cut graph correspond to cut points, while the edges
     constrain the order between the cut points.
Using cut points, we give necessary and sufficient conditions for a
     lasso to satisfy an~$\ELTLTB$ formula in 
     Theorems~\ref{thm:witness-soundness}
     and~\ref{thm:witness-completeness}.
Before defining cut graphs, we give the technical definitions of
    canonical formulas and canonical syntax trees.

\begin{definition}
    We inductively define canonical~$\ELTLTB$ formulas:

    \begin{itemize}
\item if $p$ is a propositional formula, then the formula~$p \wedge \ltlG
        \true$ is a canonical formula of rank 0,

\item if $p$ is a propositional formula and formulas $\psi_1, \dots, \psi_{k}$
    are canonical formulas (of any rank) for some $k \ge 1$, then the formula
        $p \wedge \ltlF \psi_1 \wedge \dots \wedge \ltlF \psi_k \wedge \ltlG
        \true$ is a canonical formula of rank 1,

\item if $p$ is a propositional formula and formulas $\psi_1, \dots,
    \psi_k$ are canonical formulas (of any rank) for
        some $k \ge 0$, and $\psi_{k+1}$ is a canonical formula of rank 0 or 1, then
        the formula $p \wedge \ltlF \psi_1 \wedge \dots \wedge \ltlF \psi_k
        \wedge \ltlG \psi_{k+1}$ is a canonical formula of rank 2.
    \end{itemize}
\end{definition}

\begin{example}
Let $p$ and $q$ be propositional formulas. 
The formulas $p \wedge \ltlG \true$ and $\true \wedge \ltlF (q \wedge \ltlG
    \true) \wedge \ltlG (p \wedge \ltlG \true)$ are canonical, while the
    formulas $p$, $\ltlF q$, and $\ltlG p$ are not canonical.
Continuing Example~\ref{ex:many-shapes-of-lassos}, the canonical version of the
    formula $\ltlF (a \wedge \ltlF d \wedge \ltlF e
            \wedge \ltlG b \wedge \ltlG \ltlF c)$ is the formula
 $\ltlF (a \wedge \ltlF (d \wedge \ltlG \true) \wedge \ltlF (e \wedge \ltlG \true)
    \wedge \ltlG (b \wedge \ltlF (c \wedge \ltlG \true) \wedge \ltlG \true))$.
\end{example}

We will use formulas in the following canonical form in order to
simplify presentation.

\begin{observation}
The properties of canonical $\ELTLTB$ formulas:
\begin{enumerate}
\item Every canonical formula consists of canonical subformulas
        of the form $p \wedge \ltlF
        \psi_1 \wedge \dots \wedge \ltlF \psi_k \wedge \ltlG \psi_{k+1}$ for
        some $k \ge 0$, for a propositional formula~$p$, canonical formulas
        $\psi_1, \dots, \psi_k$, and a formula~$\psi_{k+1}$ that is either
        canonical, or equals to~$\true$.

    \item If a canonical formula contains a subformula
        $\ltlG(\dots \wedge \ltlG \psi)$, then $\psi$ equals $\true$.
\end{enumerate}
\end{observation}

\begin{proposition}\label{thm:canonical-form}
    There is a function~$\canform: \ELTLTB \rightarrow \ELTLTB$
    that produces for each formula $\varphi \in \ELTLTB$ an equivalent
    canonical formula~$\canform(\varphi)$.
\end{proposition}

For an $\ELTLTB$ formula, there may be several equivalent canonical formulas,
    e.g., $p \wedge \ltlF (q \wedge \ltlG \true) \wedge \ltlF (p \wedge \ltlG
    \true) \wedge \ltlG \true$ and $p \wedge \ltlF (p \wedge \ltlG \true)
    \wedge \ltlF (q \wedge \ltlG \true) \wedge \ltlG \true$ differ in the
    order of $\ltlF$-subformulas. 
With the function~$\canform$ we fix one such a formula.

\paragraph{Canonical syntax trees.} 

The canonical syntax tree of the formula introduced in
     Example~\ref{ex:many-shapes-of-lassos} is shown in
     Figure~\ref{Fig:syntax-tree}.
With $\NatZero^*$ we denote the set of all finite words over natural
     numbers\dash---these words are used as node identifiers.

\begin{definition}
The \emph{canonical syntax tree} of a formula $\varphi \in \ELTLTB$ is
the set $\syntree(\varphi) \subseteq \ELTLTB \times \NatZero^*$
constructed inductively as follows:

\begin{enumerate}
\item The tree contains the root node
    labeled with the canonical formula~$\canform(\varphi)$ and id~$0$, that
        is, $\left<\canform(\varphi), 0\right> \in \syntree(\varphi)$.

\item Consider a tree node $\left<\psi, w\right> \in \syntree(\varphi)$ such that
        for some canonical formula $\psi' \in \ELTLTB$ one of the following
        holds: (a)~$\psi = \psi' = \canform(\varphi)$, or (b)~$\psi = \ltlF
\psi'$, or (c)~$\psi = \ltlG \psi'$. 

If $\psi'$ is $p \wedge \ltlF \psi_1 \wedge \dots \wedge \ltlF \psi_k \wedge
    \ltlG \psi_{k+1}$ for some $k \ge 0$, then the tree $\syntree(\varphi)$
    contains a child node for each of the conjuncts of~$\psi'$, that is,
    $\left<p, w.0\right> \in \syntree(\varphi)$, as well as $\left<\ltlF
    \psi_i, w.i\right> \in \syntree(\varphi)$ and $\left<\ltlG \psi_j,
    w.j\right> \in \syntree(\varphi)$ for $1 \le i \le k$ and $j=k+1$.

\end{enumerate}
\end{definition}

\begin{observation}
The canonical syntax tree $\syntree(\varphi)$ of an $\ELTLTB$ formula~$\varphi$
    has the following properties:

\begin{itemize}
\item Every node $\left<\psi, w\right>$ has the unique identifier~$w$,
        which encodes the path to the node from the root.

\item Every intermediate node is labeled with a temporal operator $\ltlF$ or
        $\ltlG$ over the conjunction of the formulas in the children nodes.
        
\item The root node is labeled with the formula~$\varphi$ itself, and $\varphi$
        is equivalent to the conjunction of the root's children formulas,
        possibly preceded with a temporal operator $\ltlF$ or $\ltlG$.

\end{itemize}
\end{observation}

The temporal formulas that appear under the operator~$\ltlG$
    have to be dealt with by the loop part of a lasso. 
To formalize this, we say that a node with id $w \in \NatZero^*$ is
    \emph{covered} by a $\ltlG$-node, if $w$ can be split into two words $u_1,
    u_2 \in \NatZero^*$ with $w = u_1.u_2$, and there is a formula $\psi \in
    \ELTLTB$ such that $\left<\ltlG \psi, u_1\right> \in {\syntree(\varphi)}$.

    \begin{figure}
        \centering
        \makeatletter{}\begin{tikzpicture}[font=\small,>=latex];
    \tikzstyle{node}=[circle,fill=black,minimum size=1mm,inner sep=0cm];
    
    \node[node,label={below:$[0]$}] (init) at (0, 0) {};
    \node[node,label={[xshift=0mm]above:$[0.1]$}] (01) at (1.5,.75) {};
    \node[node,label={[xshift=0mm]above:$[0.2]$}] (02) at (1.5,-.75) {};
    \node[node,label={below right:$\loops$}] (ls) at (3,0) {};
    \node[node,label={[xshift=0mm]above:$[0.3.1]$}] (031) at (4.5,.75) {};
    \node[node,label={below right:$\loope$}] (le) at (6, 0) {};

    \draw[->] (ls) edge (le);
    \draw[->] (init) edge (01);
    \draw[->] (init) edge (02);

    \draw[->] (ls) -- (031);
    \draw[->] (031) -- (le);

    \draw[->] (init) -- (ls);
    \draw[->] (01) -- (ls);
    \draw[->] (02) -- (ls);

\end{tikzpicture}

        \caption{The cut graph of the canonical syntax tree
            in Figure~\ref{Fig:syntax-tree}}
            \label{Fig:prec-graph}
    \end{figure}
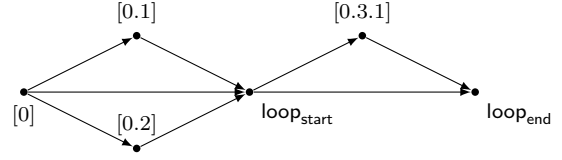

\paragraph{Cut graphs.} Using the canonical syntax
tree~$\syntree(\varphi)$ of a formula~$\varphi$, we capture in a
so-called \emph{cut graph}
     the possible orders in which formulas $\ltlF \psi$ should be
     witnessed by configurations of a lasso-shaped path.
We will then use the occurrences of the formula $\psi$
     to cut the lasso into bounded finite schedules.

\begin{example}\label{ex:cut}
Figure~\ref{Fig:prec-graph} shows the cut graph of the canonical
     syntax tree  in Figure~\ref{Fig:syntax-tree}.
It consists of tree node ids for subformulas starting with $\ltlF$,
     and two special nodes for the start and the end of the
     loop.
In the cut graph, the node with id 0 precedes the node with id 0.1,
     since at least one configuration satisfying $(a \wedge \ltlF
     (d \wedge \dots) \wedge \dots)$ should occur on a path before (or
     at the same moment as) a state satisfying $(d \wedge
     \dots)$.
Similarly, the node with id 0 precedes the node with id
     0.2.
The nodes with ids 0.1 and 0.2 do not have to precede each
     other, as the formulas~$d$ and~$e$ can be satisfied in either
     order.
Since the nodes with the ids $0$, $0.1$, and $0.2$ are not covered by
     a $\ltlG$-node, they both precede the loop start.
The loop start precedes the node with id $0.3.1$, as this node is
     covered by a $\ltlG$-node.
\end{example}

\begin{definition}\label{def:cutgraph}
    The \emph{cut graph} $\pgraph(\varphi)$ of an $\ELTLTB$ formula
    is a directed acyclic graph $(\pgvertices, \pgedges)$
    with the following properties:
\begin{enumerate}
    
\item The set of nodes $\pgvertices=\{ \loops, \loope \} \cup \{w \in
    \NatZero^* \mid \exists \psi.\ \left<\ltlF \psi, w\right> \in \syntree(\varphi)\}$
        contains the tree ids that label $\ltlF$-formulas
        and two special nodes $\loops$ and $\loope$, which denote the start and
        the end of the loop respectively.

\item The set of edges~$\pgedges$ satisfies the following constraints:

    \begin{enumerate}

    \item Each tree node $\left<\ltlF \psi, w\right> \in \syntree(\varphi)$
        that is \emph{not} covered by a $\ltlG$-node precedes the loop start, i.e.,
        $(w, \loops) \in \pgedges$.

    \item For each tree node $\left<\ltlF \psi, w\right> \in \syntree(\varphi)$
        covered by a $\ltlG$-node:
\begin{itemize}
\item the loop start precedes $w$, i.e., $(\loops, w) \in \pgedges$, and
\item  $w$ precedes the loop end, i.e., $(w, \loope) \in \pgedges$.
\end{itemize} 

\item For each pair of tree nodes $\left<\ltlF \psi_1, w\right>, \left<\ltlF
    \psi_2, w.i\right> \in \syntree(\varphi)$ not covered by a $\ltlG$-node, we
        require $(w, w.i) \in \pgedges$.

\item For each pair of tree nodes $\left<\ltlF \psi_1, w_1\right>, \left<\ltlF
        \psi_2, w_2\right> \in \syntree(\varphi)$ that are both covered by a
        $\ltlG$-node, we require either $(w_1, w_2) \in \pgedges$, or $(w_2, w_1)
        \in \pgedges$ (but not both).

    \end{enumerate}

\end{enumerate}
\end{definition}

\begin{definition}\label{def:cutfun}
Given a lasso~$\tau \cdot \rho^\omega$ and a cut graph
    $\pgraph(\varphi)=(\pgvertices, \pgedges)$, we call a function $\imap:
    \pgvertices \to \{0, \dots, |\tau| + |\rho| - 1\}$ a 
    \emph{cut function}, if the following holds:
       
    \begin{itemize}

    \item $\imap(\loops) = |\tau|$ and
            $\imap(\loope) = |\tau| + |\rho| - 1$,

    \item if $(v, v') \in \pgedges$, then $\imap(v) \le \imap(v')$.

\end{itemize}

\end{definition}

We call the indices $\{\imap(v) \mid v \in \pgvertices \}$ the \emph{cut
    points}. 
Given a schedule~$\tau$ and an index~$k: 0 \le k < |\tau| + |\rho|$, we say
    that the index~$k$ \emph{cuts} $\tau$ into $\pi'$ and $\pi''$, if $\tau =
    \pi' \cdot \pi''$ and $|\pi'| = k$.

Informally, for a tree node $\left<\ltlF \psi, w\right> \in
     \syntree(\varphi)$, a cut point $\imap(w)$ witnesses satisfaction
     of~$\ltlF \psi$, that is, the formula~$\psi$ holds at the
     configuration located at the cut point.
It might seem that Definitions~\ref{def:cutgraph} and~\ref{def:cutfun}
     are too restrictive.
For instance, assume that the node $\left<\ltlF \psi, w\right>$ is not
     covered by a $\ltlG$-node, and there is a lasso schedule~$\tau
     \cdot \rho^\omega$ that satisfies the formula~$\varphi$ at a
     configuration~$\gst$.
It is possible that the formula~$\psi$ is witnessed only by a cut
     point inside the loop.
At the same time, Definition~\ref{def:cutfun} forces~$\imap(w) \le
     \imap(\loops)$.
We show that this problem is resolved by unwinding the loop~$K$ times
     for some $K \ge 0$, so that there is a cut function for the lasso
     with the prefix $\tau \cdot \rho^K$ and the loop~$\rho$:

\newcommand{\propunwinding}{Let $\varphi$ be an $\ELTLTB$ formula, $\gst$ be a configuration and $\tau
    \cdot \rho^\omega$ be a lasso schedule applicable to~$\gst$ such that
    $\infpath{\gst}{\tau \cdot \rho^\omega} \models \varphi$ holds. 
There is a constant~$K \ge 0$ and a cut function~$\imap$
    such that for every $\left<\ltlF \psi, w\right> \in
    \pgraph(\syntree(\varphi))$ if~$\imap(w)$ cuts $(\tau \cdot \rho^K) \cdot
    \rho$ into $\pi'$ and~$\pi''$, then~$\psi$ is satisfied at the cut point,
        that is,
        $\infpath{\pi'(\gst)}{\pi'' \cdot \rho^\omega} \models \psi$.}

\begin{proposition}\label{prop:unwinding}
\propunwinding
\end{proposition}

\begin{proof}[Proof sketch]
The detailed proof is given in~\cite{KLVW16:arxiv}.
We will present the required constant~$K \ge 0$ and the cut
     function~$\imap$.
To this end, we use extreme appearances of
     $\ltlF$-formulas~(cf.~\cite[Sec.~4.3]{EtessamiVW02}) and use them
     to find~$\imap$.
An extreme appearance of a formula~$\ltlF \psi$ is the
     furthest point in the lasso that still witnesses~$\psi$.
There might be a subformula that is required to be witnessed in the
prefix,
 but in $\tau
    \cdot \rho^\omega$it
     is only witnessed by the loop.
To resolve this, we replace $\tau$ by a a longer prefix $\tau \cdot \rho^K$, by
 unrolling the loop~$\rho$ several times; more precisely, $K$ times, 
   where $K$ is the number
     of nodes that should precede the lasso start.
In other words, if all extreme appearances of the nodes happen to be in
     the loop part, and they appear in the order that is against the
     topological order of the graph~$\pgraph(\syntree(\varphi))$, we
     unroll the loop~$K$ times (the number of nodes that have to be in
     the prefix) to find the prefix, in which the nodes respect the
     topological order of the graph.
In the unrolled schedule we can now find extreme appearances of the
required subformulas in the
     prefix.
\end{proof}

We show that to satisfy an~$\ELTLTB$ formula, a lasso should
    (i)~satisfy propositional subformulas of $\ltlF$-formulas in the
    respective cut points, and (ii)~maintain the propositional formulas of
    $\ltlG$-formulas from some cut point on. 
This is formalized as a witness. 

In the following definition, we use a short-hand notation for propositional
    subformulas: given an~$\ELTLTB$-formula $\psi$ and its canonical
    form~$\canform(\psi) = \psi_0 \wedge
    \ltlF \psi_1 \wedge \dots \wedge \ltlF \psi_k \wedge \ltlG \psi_{k+1}$,
    we use the notation~$\prop(\psi)$ to denote the formula~$\psi_0$.

\begin{definition}\label{def:prop-witness}
Given a configuration~$\gst$, a lasso $\tau \cdot \rho^\omega$ applicable
    to~$\gst$, and an $\ELTLTB$ formula $\varphi$, a cut
    function~$\imap$ of $\pgraph(\syntree(\varphi))$ is a \emph{witness}
    of~$\infpath{\gst}{\tau \cdot \rho^\omega} \models \varphi$, if the
    three conditions hold:

\begin{enumerate}[label={\bfseries\emph{(C\arabic*)}}, leftmargin=.75cm]
        \item\label{assume:root}
            For $\canform(\varphi) \equiv \psi_0 \wedge \ltlF \psi_1 \wedge \dots
            \wedge \ltlF \psi_k \wedge \ltlG \psi_{k+1}$:
\begin{enumerate}
\item            $\gst \models \psi_0$, and
\item            $\setconf{\gst}{\tau \cdot \rho} \models \prop(\psi_{k+1})$.
\end{enumerate}
        \item\label{assume:Gfin-prefix}
            For $\left<\ltlF \psi, v\right> \in \syntree(\varphi)$
            with $\imap(v) < |\tau|$, if $\imap(v)$ cuts $\tau \cdot \rho$
            into $\pi'$ and $\pi''$ and
            $\psi \equiv \psi_0 \wedge \ltlF \psi_1 \wedge \dots
            \wedge \ltlF \psi_k \wedge \ltlG \psi_{k+1}$,
            then:
\begin{enumerate}
\item $\pi'(\gst) \models \psi_0$, and 
\item            $\setconf{\pi'(\gst)}{\pi''} \models \prop(\psi_{k+1})$.
\end{enumerate}

        \item\label{assume:Gfin-loop}
            For $\left<\ltlF \psi, v\right> \in \syntree(\varphi)$
            with $\imap(v) \ge |\tau|$, if $\imap(v)$ cuts $\tau \cdot \rho$
            into $\pi'$ and $\pi''$ and
            $\psi \equiv \psi_0 \wedge \ltlF \psi_1 \wedge \dots
            \wedge \ltlF \psi_k \wedge \ltlG \psi_{k+1}$,
            then:
\begin{enumerate}
\item $\pi'(\gst) \models \psi_0$, and 
\item            $\setconf{\tau(\gst)}{\rho} \models
  \prop(\psi_{k+1})$.
\end{enumerate}

    \end{enumerate}
\end{definition}

Conditions~(a) require that propositional formulas hold
     in a configuration, while conditions~(b) require that
     propositional formulas hold on a finite suffix.
Hence, to ensure that a cut function constitutes a witness, one has to
     check the configurations of a \emph{fixed number of finite} paths
     (between the cut points).
This property is crucial for the path reduction (see
     Section~\ref{sec:repr}).
Theorems~\ref{thm:witness-soundness}
     and~\ref{thm:witness-completeness} show that the existence of a
     witness is a sound and complete criterion for the existence of a
     lasso satisfying an~$\ELTLTB$ formula.

\newcommand{\thmwitnesssoundness}{Let~$\gst$ be a configuration, $\tau \cdot \rho^\omega$ be a lasso applicable
    to~$\gst$, and $\varphi$ be an $\ELTLTB$ formula.
If there is a witness of $\infpath{\gst}{\tau \cdot \rho^\omega} \models \varphi$,
    then the lasso~$\tau \cdot \rho^\omega$ satisfies~$\varphi$,
    that is~$\infpath{\gst}{\tau \cdot \rho^\omega} \models \varphi$.}

\begin{theorem}[Soundness]\label{thm:witness-soundness}
\thmwitnesssoundness
\end{theorem}

\newcommand{\thmwitnesscompleteness}{Let $\varphi$ be an $\ELTLTB$ formula, $\gst$ be a configuration and $\tau
    \cdot \rho^\omega$ be a lasso applicable to~$\gst$ such that
    $\infpath{\gst}{\tau \cdot \rho^\omega} \models \varphi$ holds. 
    There is a witness of $\infpath{\gst}{(\tau \cdot \rho^K) \cdot \rho^\omega}
    \models \varphi$ for some~$K \ge 0$.}

\begin{theorem}[Completeness]\label{thm:witness-completeness}
\thmwitnesscompleteness
\end{theorem}

Theorem~\ref{thm:witness-soundness} is proven for subformulas of
     $\varphi$ by structural induction on the intermediate nodes of the
     canonical syntax tree.
In the proof of Theorem~\ref{thm:witness-completeness} we use
     Proposition~\ref{prop:unwinding} to prove the points of
     Definition~\ref{def:prop-witness}.
(The detailed proofs are given in~\cite{KLVW16:arxiv}.) 

\subsection{Using Cut Graphs to Enumerate Shapes of Lassos}
    \label{sec:enumerating-lassos}

Proposition~\ref{prop:lasso-sched} and Theorem~\ref{thm:witness-completeness}
    suggest that in order to find a schedule that satisfies an $\ELTLTB$
    formula~$\varphi$, it is sufficient to look for lasso schedules that can be cut in
    such a way that the configurations at the cut points and the configurations
    between the cut points satisfy certain propositional formulas. 
In fact, the cut points as defined by cut functions
    (Definition~\ref{def:cutfun}) are \emph{topological orderings} of the cut
    graph~$\pgraph(\syntree(\varphi))$.
Consequently, by enumerating the topological
    orderings of the cut graph~$\pgraph(\syntree(\varphi))$ we can enumerate
    the \emph{lasso shapes}, among which there is a lasso schedule
    satisfying~$\varphi$ (if~$\varphi$ holds on the counter
    system). 
These shapes differ in the order, in which $\ltlF$-subformulas of~$\varphi$ are
    witnessed.
For this, one can use fast generation algorithms,
    e.g.,~\cite{Canfield1995}.

    \begin{example}
Consider the cut graph in Figure~\ref{Fig:prec-graph}. 
The ordering of its vertices $0, 0.1, 0.2, \loops, 0.3.1, \loope$ corresponds
    to the lasso shape~(a) shown in Figure~\ref{fig:lasso-shapes}, while the
    ordering $\loops, 0, 0.2, 0.1, \loops, 0.3.1, \loope$ corresponds to the
    lasso shape~(b). 
These are the two lasso shapes that one has to analyze, and they are
the result of our construction using the cut graph.
The other 18 lasso shapes in the figure are not required, and not
constructed by our method.
    \end{example}

From this observation, we conclude that given a topological ordering $v_1,
\dots, v_{|\pgvertices|}$ of the cut graph $\pgraph(\syntree(\varphi)) =
    (\pgvertices, \pgedges)$, one has to look for a lasso schedule that
    can be written as an alternating sequence of configurations~$\gst_i$
    and schedules~$\tau_j$:
\begin{equation}\label{eq:lasso-sequence}
    \gst_0,\tau_0,\gst_1,\tau_1, \dots, \gst_\ell, \tau_\ell,
        \dots, \gst_{|\pgvertices|-1}, \tau_{|\pgvertices|}, \gst_{|\pgvertices|},
\end{equation}
where $v_\ell = \loops$, $v_{|\pgvertices|} = \loope$, and $\gst_\ell =
    \gst_{|\pgvertices|}$.
Moreover, by Definition~\ref{def:prop-witness}, the sequence of
    configurations and schedules should satisfy~\ref{assume:root}--\ref{assume:Gfin-loop}, e.g.,
 if a node~$v_i$ corresponds to the formula $\ltlF
    (\psi_0 \wedge \dots \wedge \ltlG \psi_{k+1})$ and this formula
    matches Condition~\ref{assume:Gfin-prefix}, then the following
    should hold:

\begin{enumerate}
\item Configuration~$\gst_i$ satisfies the propositional formula:
      $\gst_i \models \psi_0$.

\item All configurations visited by the schedule $\tau_i \concat \dots \concat
        \tau_{|\pgvertices|}$ from the configuration~$\gst_i$ satisfy the
        propositional formula~$\prop(\psi_{k+1})$.
        Formally, $\setconf{\gst_i}{\tau_i \concat \dots \concat
        \tau_{|\pgvertices|}} \models \prop(\psi_{k+1})$.
\end{enumerate}

One can write an SMT query for the sequence
    (\ref{eq:lasso-sequence}) satisfying
    Conditions~\ref{assume:root}--\ref{assume:Gfin-loop}.
However, this approach has two problems:
\begin{enumerate}
    \item The order of rules in schedules $\tau_0, \dots, \tau_{|\pgvertices|}$
        is not fixed. Non-deterministic choice of rules complicates
        the SMT query.

    \item To guarantee completeness of the search, one requires a bound
        on the length of schedules $\tau_0, \dots, \tau_{|\pgvertices|}$.
\end{enumerate}

For reachability properties these issues were addressed
     in~\cite{KVW15:CAV}  by showing that one only has to consider
     specific orders of the rules;  so-called representative schedules.
To lift this technique to~$\ELTLTB$, we are left with two issues: 

\begin{enumerate}
\item The shortening technique applies to steady schedules, i.e., the schedules
that do not change evaluation of the guards. 
Thus, we have to break the schedules $\tau_0, \dots, \tau_{|\pgvertices|}$ into
steady schedules. 
This issue is addressed in Section~\ref{sec:structguards}.

\item The shortening technique preserves state reachability, e.g.,
    after shortening of~$\tau_i$, the resulting schedule still reaches
    configuration~$\gst_{i+1}$. But it may violate an invariant
    such as $\setconf{\gst_i}{\tau_i \concat \dots \concat
        \tau_{|\pgvertices|}} \models \prop(\psi_{k+1})$.
    This issue is addressed in Section~\ref{sec:repr}.
\end{enumerate}

\makeatletter{}\section{Cutting Lassos with Threshold Guards}\label{sec:structguards}

We introduce threshold graphs to cut a lasso into steady
    schedules, in order to apply the shortening technique of
    Section~\ref{sec:repr}. 
Then, we combine the cut graphs and threshold graphs to cut a lasso into smaller
    finite segments, which can be first shortened and then checked with the
    approach introduced in Section~\ref{sec:enumerating-lassos}. 

Given a configuration~$\gst$, its context~$\statectx(\gst)$ is the set that
    consists of the lower guards unlocked in~$\gst$ and the upper guards locked
    in~$\gst$, i.e., $\statectx(\gst) = \CtxU \cup \CtxL$, where $\CtxU = \{ g
    \in \PrecondU \mid \gst \models g \}$ and $\CtxL = \{ g \in \PrecondL \mid
\gst \nmodels g \}$. 
As discussed in Example~\ref{ex:guards} on page~\pageref{ex:guards},
since the shared variables are never decreased, the contexts in a path are
    monotonically non-decreasing:

\begin{proposition}[Prop.~3 of \cite{KVW15:CAV}]\label{prop:mono}
    If a transition~$t$ is enabled in a configuration~$\gst$, then
    $\statectx(\gst) \subseteq \statectx(t(\gst))$.
\end{proposition}

\begin{example}
Continuing Example~\ref{ex:guards}, which
    considers the TA in Figure~\ref{fig:stunningexample}. 
Both threshold guards~$\gamma_1$ and~$\gamma_2$ are false in the initial
    state~$\gst$. 
Thus, $\statectx(\gst)=\emptyset$. 
The transition $t=(r_1, 1)$ unlocks the guard~$\gamma_1$, i.e.,
    $\statectx(t(\gst)) = \{\gamma_1\}$.
\end{example}

As the transitions of the counter system~$\Sys(\TA)$ never decrease shared
    variables, the loop of a lasso schedule must be steady:

\begin{proposition}\label{prop:lasso-is-steady}
    For each configuration~$\gst$ and a schedule~$\tau \concat
    \rho^\omega$, if $\rho^k(\tau(\gst)) = \tau(\gst)$ for $k \ge 0$,
    then the loop $\rho$ is steady for~$\tau(\gst)$, that is,
    $\statectx(\rho(\tau(\gst))) = \statectx(\tau(\gst))$.
\end{proposition}

In~\cite{KVW15:CAV}, Proposition~\ref{prop:mono} was used to cut a
     finite path into segments, one per context.
We introduce threshold graphs and their topological orderings to apply
     this idea to lasso schedules.

\begin{definition}\label{def:tgraph}
   A \emph{threshold graph} is~$\tgraph(\TA)= (\tgvertices, \tgedges)$
   such that:

    \begin{itemize}
\item The vertices set~$\tgvertices$ contains the threshold guards
        and the special node~$\loops$, i.e., $\tgvertices =
        \PrecondU \cup \PrecondL \cup \{\loops\}$.

\item There is an edge from a guard~$g_1 \in \PrecondU$ to a guard~$g_2 \in
        \PrecondU$, if $g_2$ cannot be unlocked before~$g_1$, i.e., $(g_1, g_2)
        \in \tgedges$, if for each configuration $\gst \in \configs$, $\gst
        \models g_2$ implies $\gst \models g_1$.

\item There is an edge from a guard~$g_1 \in \PrecondL$ to a guard~$g_2 \in
        \PrecondL$, if $g_2$ cannot be locked before~$g_1$, i.e., $(g_1, g_2)
        \in \tgedges$, if for each configuration $\gst \in \configs$, $\gst
        \nmodels g_2$ implies $\gst \nmodels g_1$.

    \end{itemize}
\end{definition}

Note that the conditions in Definition~\ref{def:tgraph} can be easily checked
    with an SMT solver, for all configurations.

\begin{example}
The threshold graph of the TA in Figure~\ref{fig:stunningexample} has
    the vertices $\tgvertices=\{\gamma_1, \gamma_2,\loops\}$ and the edges
    $\tgedges = \{(\gamma_1, \gamma_2)\}$.
\end{example}

\begin{figure}
    \begin{center}
        \makeatletter{}\begin{tikzpicture}[x=1cm,y=1cm,font=\scriptsize,>=latex];
    \tikzstyle{node}=[circle,fill=black,minimum size=0.1cm,inner sep=0cm];
    \tikzstyle{cut}=[cross out,thick,draw=red!90!black,
        minimum size=0.15cm,inner sep=0mm,outer sep=.1mm];
    \tikzstyle{path}=[-];
    \tikzstyle{Gfin}=[-, very thick, blue];
    
  \begin{scope}[xshift=0cm, yshift=0cm]
    \node at (0.1, .3) { \normalsize\textbf{(a)} };

    \node[node] (0) at (0, 0) {};

            \foreach \x/\n in {.75/A, 1.5/B, 2.25/C, 3.0/D, 3.75/E}
        \node[cut] (\n) at (\x,0) {}; 
    \draw[path] (0) -- (A);
    \draw[path] (A) -- (B);
    \draw[path] (B) -- (C);
    \draw[Gfin] (C) -- (D);
    \draw[Gfin] (D) -- (E);
    \draw[->] (E) edge[bend right=60] (C);

    \draw[|<->|] ($(0)+(0,-.5)$)
        --node[midway, fill=white, text=black]
        {$\kappa[\ell_3]=0$} ($(E)+(0,-.5)$);
    \node at ($(A)+(0,-.25)$) {$\gamma_1$};
    \node at ($(B)+(0,-.25)$) {$\gamma_2$};
    \node at ($(D)+(0,-.25)$) {$\psi_{\mathit{fair}}$};
  \end{scope}
    
  \begin{scope}[xshift=4.3cm, yshift=0cm]
    \node at (0.1, .3) { \normalsize\textbf{(b)} };

    \node[node] (0) at (0, 0) {};

            \foreach \x/\n in {1.1/A, 2.25/C, 3.0/D, 3.75/E}
        \node[cut] (\n) at (\x,0) {}; 
    \draw[path] (0) -- (A);
    \draw[path] (A) -- (C);
    \draw[Gfin] (C) -- (D);
    \draw[Gfin] (D) -- (E);
    \draw[->] (E) edge[bend right=60] (C);

    \draw[|<->|] ($(0)+(0,-.5)$)
        --node[midway, fill=white, text=black]
        {$\kappa[\ell_3]=0$} ($(E)+(0,-.5)$);
    \node at ($(A)+(0,-.25)$) {$\gamma_1$};
    \node at ($(D)+(0,-.25)$) {$\psi_{\mathit{fair}}$};
  \end{scope}
    
  \begin{scope}[xshift=0cm, yshift=-1.5cm]
    \node at (0.1, .3) { \normalsize\textbf{(c)} };

    \node[node] (0) at (0, 0) {};

            \foreach \x/\n in {2.25/C, 3.0/D, 3.75/E}
        \node[cut] (\n) at (\x,0) {}; 
    \draw[path] (0) -- (D);
    \draw[Gfin] (C) -- (D);
    \draw[Gfin] (D) -- (E);
    \draw[->] (E) edge[bend right=60] (C);

    \draw[|<->|] ($(0)+(0,-.5)$)
        --node[midway, fill=white, text=black]
        {$\kappa[\ell_3]=0$} ($(E)+(0,-.5)$);
    \node at ($(D)+(0,-.25)$) {$\psi_{\mathit{fair}}$};
  \end{scope}
  
\end{tikzpicture}

    \end{center}

    \caption{The shapes of lassos to check the correctness property
            in Example~\ref{ex:fair}. Recall that $\gamma_1$ and $\gamma_2$
            are the threshold guards, defined as $x \ge t+1-f$ and $x \ge n-t-f$
            respectively.
    }
    \label{fig:ta-lasso-shapes}
\end{figure}
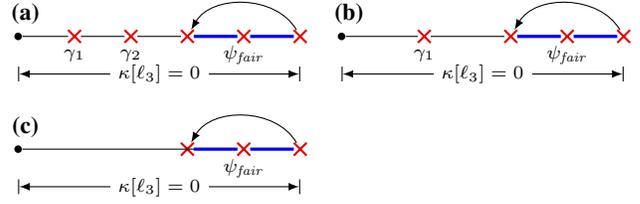

Similar to Section~\ref{sec:enumerating-lassos}, we consider a topological
    ordering $g_1, \dots, g_\ell, \dots, g_{|\tgvertices|}$ of the vertices
    of the threshold graph.
The node~$g_\ell = \loops$ indicates the point where a loop should start, and
    thus by Proposition~\ref{prop:lasso-is-steady}, after that point the context
    does not change. 
Thus, we consider only the subsequence~$g_1, \dots, g_{\ell-1}$ and split the
    path $\infpath{\gst}{\tau \concat \rho}$ of a lasso schedule~$\tau \concat
    \rho^\omega$ into an alternating sequence of configurations~$\gst_i$ and
    schedules~$\tau_0$ and $t_j \cdot \tau_j$, for $1 \le j < \ell$,
    ending up with the loop~$\rho$
    (starting in~$\gst_{\ell-1}$ and ending in~$\gst_\ell = \gst_{\ell-1}$):
\begin{equation}
    \gst_0,\tau_0,\gst_1, (t_1\concat\tau_1) , \dots,
    \gst_{\ell-2}, (t_{\ell-1}\concat\tau_{\ell-1}), \gst_{\ell-1}, \rho,
    \gst_\ell \label{eq:lasso-thresh-seq}
\end{equation}

In this sequence, the transitions $t_1, \dots, t_{\ell-1}$ change the
     context, and the schedules~$\tau_0, \tau_1, \dots, \tau_{\ell-1},
     \rho$ are steady.
Finally, we interleave a topological ordering of the vertices of the
     cut graph with a topological ordering of the vertices of the
     threshold graph.
More precisely, we use a topological ordering of the vertices of the
union of the cut graph and the threshold graph.
We use the resulting sequence to cut a lasso schedule
     following the approach in Section~\ref{sec:enumerating-lassos}
     (cf.
Equation~(\ref{eq:lasso-sequence})).
By enumerating all such interleavings, we obtain all lasso shapes.
Again, the lasso is a sequence of steady schedules and
     context-changing transitions.

\begin{example}
Continuing Example~\ref{eq:correctness} given on page~\pageref{eq:correctness},
    we consider the lasso shapes that satisfy the $\ELTLTB{}$ formula
    $\ltlG\ltlF \propfair \wedge  \counters[\ell_0] = 0 \wedge \ltlG
    \counters[\ell_3] = 0$.
Figure~\ref{fig:ta-lasso-shapes} shows the lasso shapes that have to be
    inspected by an SMT solver.
In case~(a), both threshold guards~$\gamma_1$ and~$\gamma_2$ are eventually
    changed to true, while the counter~$\kappa[\ell_3]$ is never increased
    in a fair execution.
For $n=3t$, this is actually a counterexample to the correctness property explained in
    Example~\ref{eq:correctness}.
In cases~(b) and~(c) at most one threshold guard is eventually changed to true,
    so these lasso shapes cannot produce a counterexample.
\end{example}

In the following section, we will show how to shorten steady schedules, while
    maintaining Conditions~\ref{assume:root}--\ref{assume:Gfin-loop} of
    Definition~\ref{def:prop-witness}, required to satisfy the~$\ELTLTB$
    formula.

\makeatletter{}\section{The Short Counterexample Property}\label{sec:repr}

Our verification approach focuses on counterexamples, and
as discussed in Section~\ref{sec:reach-and-live}, negations of
     specifications are expressed in~$\ELTLTB{}$.
In the case of reachability properties, counterexamples are finite
     schedules reaching a bad  state from an initial state.
An efficient method for finding counterexamples to reachability
 can be found in~\cite{KVW15:CAV}.
It is based on the short counterexample property.
Namely, it was proven that for each threshold automaton, there is a
     constant~$d$ such that if there is a schedule that reaches a bad
     state, then there must also exist an accelerated schedule that
     reaches that state in at most~$d$ transitions (i.e., $d$ is
     the diameter of the counter system).
The proof in~\cite{KVW15:CAV} is based on the following three steps:
\begin{enumerate}
 \item each finite schedule (which may or may not be a counterexample),
    can be divided into a few steady schedules,

 \item for each of these steady schedules they find a representative, 
    i.e., an accelerated schedule of bounded length,
    with the same starting and ending configurations 
    as the original schedule,

 \item at the end, all these representatives are concatenated in 
    the same order as the original steady schedules.
\end{enumerate}

This result guarantees that the system is correct if no
     counterexample to reachability properties is found using bounded
     model checking with bound~$d$.
In this section, we extend the technique from Point~2 from
     reachability properties to $\ELTLTB{}$ formulas.
The central result regarding Point~2  is the following proposition
     which is a specialization of~\cite[Prop.~7]{KVW15:CAV}:

  \begin{proposition}\label{prop:srep-ex}
Let $\Sk = (\local, \initlocal, \globset,
     \paraset, \ruleset,\ResCond)$ be a threshold automaton.
For every configuration~$\gst$ and every steady
    schedule~$\tau$ applicable to~$\gst$, there exists a steady schedule
    $\xrep{\gst}{\tau}$ with the following properties:
$\xrep{\gst}{\tau}$ is applicable
        to $\gst$, $\xrep{\gst}{\tau}(\gst) = \tau(\gst)$, and
        $|\xrep{\gst}{\tau}| \le 2 \cdot |\ruleset|$.
\end{proposition}

We observe that the proposition talks about the first
     configuration~$\sigma$ and the last one $\tau(\sigma)$, while it
     ignores intermediate configurations.
However, for $\ELTLTB{}$ formulas, one has to consider all
     configurations in a schedule, and not just the first and the last
     one.

  \begin{figure}[t]
\centering
\scalebox{1.}{\makeatletter{}\tikzstyle{trans}=[->,line width=0.5mm]

\begin{tikzpicture}
 \node at (-1.5, 0) { {\Large $\tau_{\mathrm{up}}$:} };
 \node (s1) at (0,0) {$\gst_1$};
 \node (s2) at (3,0) {$\gst_2$};
 \node (s3) at (6,0) {$\gst_3$};
 
 \draw[trans] (s1) to[]
    node(k1)[align=center,anchor=south, midway]
    {$\cpp{\counters[\ell]}$}  (s2);
 \draw[trans] (s2) to[]
    node(k2)[align=center,anchor=south, midway]
    {$\counters[\ell]\scriptsize{\texttt{-}\texttt{-}}$}  (s3);

 \node at (-1.5, -1.0) { {\Large $\tau_{\mathrm{down}}$:} };
 \node (s1') at (0,-1.0) {$\gst_1$};
 \node (s2') at (3,-1.0) {$\gst'_2$};
 \node (s3') at (6,-1.0) {$\gst_3$};
 
 \draw[trans] (s1') to[]
    node(k3)[align=center,anchor=north, midway]
    {$\counters[\ell]\scriptsize{\texttt{-}\texttt{-}}$}   (s2');
 \draw[trans] (s2') to[]
    node(k4)[align=center,anchor=north, midway]
    {$\cpp{\counters[\ell]}$} (s3');
    
 \draw[->,dashed,line width=0.36mm] ($(0.1,-0.5)+(k1)$) -- ($(0.1,0.5)+(k4)$);
 \draw[->,dashed,line width=0.36mm] ($(-0.1,-0.5)+(k2)$) -- ($(-0.1,0.5)+(k3)$);
 
\end{tikzpicture}
 }
\caption{Changing the order of transitions can violate $\ELTLTB{}$ formulas. If $\sigma_1.\counters[\ell]=1$, then for the upper 
  schedule~$\tau_{\text{up}}$ holds that 
  $\setconf{\gst_1}{\tau_{\text{up}}}\models
  \counters[\ell] >  0$, while for the lower one this is not the case,
  because $\gst_2'\not\models \counters[\ell] >  0$.}
\label{fig:swapping}
\end{figure}

\begin{example} \label{ex:swap}
Figure~\ref{fig:swapping} shows the result of swapping transitions.
The approaches by~\cite{Lipton75} and~\cite{KVW15:CAV} are only
     concerned with the first and last configurations: they use the
     property that after swapping transitions, $\gst_3$ is still
     reached from~$\gst_1$.
The arguments used in~\cite{Lipton75,KVW15:CAV} do not care about the
     fact that the resulting path visits a different intermediate
     state ($\gst'_2$ instead of $\gst_2$).
However,   if $\sigma_1.\counters[\ell]=1$, then
     $\gst_2.\counters[\ell] > 0$, while $\gst'_2.\counters[\ell] =
     0$.
Hence, swapping transitions may change the evaluation of
$\ELTLTB{}$ formulas, e.g., $\ltlG (\counters[\ell] >  0)$.
\end{example}

\begin{figure*}[t]
\begin{center}
\makeatletter{} \tikzstyle{inx}=[circle,draw=black!90,fill=black!10,
    thick,minimum size=4.8mm,inner sep=0.75mm,font=\normalsize]
\tikzstyle{outx}=[circle,draw=black!90,fill=white,
    thick,minimum size=4.8mm,inner sep=0.5mm,font=\normalsize]
\tikzstyle{rule}=[->,thick]
\tikzstyle{post}=[->,thick,rounded corners,font=\normalsize]
\tikzstyle{comment}=[color=slateblue,font=\normalsize]
\tikzset{every loop/.style={min distance=5mm,in=140,out=113,looseness=2}}
\tikzstyle{token}=[draw,fill,circle,minimum size=0.8mm,inner sep=0.5mm, color=black!60]
\tikzstyle{box}=[rounded corners=.2cm,opacity=.9,very thick,fill=none]
\tikzstyle{trans}=[->,line width=0.4mm]
\tikzstyle{transt}=[->,line width=.8mm]
\tikzstyle{crit}=[draw=black!90,fill,diamond,minimum size=2mm, fill=black!10,
		  rounded corners=.02cm,text width=0.22cm]
\tikzstyle{translasso}=[->,thick]
\tikzstyle{state}=[fill,circle,minimum size=1mm,inner sep=0mm]
\tikzstyle{path}=[->,thick,dashed]
\tikzset{every loop/.style={min distance=5mm,in=140,out=113,looseness=2}}

\begin{tikzpicture}[=latex]

\begin{scope}[shift={(-7.3,0)}]
 \node[] at (0,0) [outx,label=left:\textcolor{blue}{$\ell_0$}] (0) {};
 \node[] at ($(0)+(0,-1.8)$) [outx,label=left:\textcolor{blue}{$\ell_1$}] (1) {};
 \node[] at ($(1)+(1.7,0.9)$) [inx,label=below:\textcolor{blue}{$\ell_2$}] (2) {};
 \node[] at ($(2)+(1.5,0)$) [outx,label=below right:\textcolor{blue}{$\ell_3$}] (3) {};

\draw[post] (0) to[]
    node[align=center,anchor=east, midway,yshift=-.1cm]
    {$r_2$} (2);
\draw[post] (1) to[] node[anchor=north,yshift=.3cm,xshift=-.3cm]
    {$r_1$}(2);
\draw[post] (0) -| node[anchor=south, pos=.25] (xpp)
    {$r_3$} (3);
\draw[post] (2)to[]
    node[align=center,anchor=north, midway]
    {$r_4$} (3);
\draw[post] (1) -| node[anchor=north, pos=.25,yshift=0cm,xshift=.3cm] (xget)
    {$r_5$} (3);

\draw[rule] (0) to[out=225,in=270,looseness=8]
	    node[align=center,anchor=east,midway]{$r_6$} (0);
\draw[rule] (2) to[out=15,in=60,looseness=8]
	    node[align=center,anchor=north,midway,yshift=.4cm]{$r_7$} (2);
\draw[rule] (3) to[out=15,in=60,looseness=8]
	    node[align=center,anchor=north,midway,yshift=.4cm]{$r_8$} (3);

\fill[box,draw=black!25] ($(1)+(-1,-0.5)$) rectangle ($(3)+(0.9,1.4)$);

\node[token] at ($(2)+(0.08,-0.08)$) {};
\node[token] at ($(0)+(-0.08,-0.08)$) {};
\node[token] at ($(1)+(-0.08,0.08)$) {};

\node[] at ($(2)+(0,-1.7)$) {Configuration $\gst_1$};
\end{scope}

\begin{scope}[shift={(4,0)}]
 \node[] at (0,0) [outx,label=left:\textcolor{blue}{$\ell_0$}] (0) {};
 \node[] at ($(0)+(0,-1.8)$) [outx,label=left:\textcolor{blue}{$\ell_1$}] (1) {};
 \node[] at ($(1)+(1.7,0.9)$) [inx,label=below:\textcolor{blue}{$\ell_2$}] (2) {};
 \node[] at ($(2)+(1.5,0)$) [outx,label=below right:\textcolor{blue}{$\ell_3$}] (3) {};

\draw[post] (0) to[]
    node[align=center,anchor=east, midway,yshift=-.1cm]
    {$r_2$} (2);
\draw[post] (1) to[] node[anchor=north,yshift=.3cm,xshift=-.3cm]
    {$r_1$}(2);
\draw[post] (0) -| node[anchor=south, pos=.25] (xpp)
    {$r_3$} (3);
\draw[post] (2)to[]
    node[align=center,anchor=north, midway]
    {$r_4$} (3);
\draw[post] (1) -| node[anchor=north, pos=.25,yshift=0cm,xshift=.3cm] (xget)
    {$r_5$} (3);

\draw[rule] (0) to[out=225,in=270,looseness=8]
	    node[align=center,anchor=east,midway]{$r_6$} (0);
\draw[rule] (2) to[out=15,in=60,looseness=8]
	    node[align=center,anchor=north,midway,yshift=.4cm]{$r_7$} (2);
\draw[rule] (3) to[out=15,in=60,looseness=8]
	    node[align=center,anchor=north,midway,yshift=.4cm]{$r_8$} (3);

\fill[box,draw=black!25] ($(1)+(-1,-0.5)$) rectangle ($(3)+(0.9,1.4)$);

\node[token] at ($(3)+(0.08,-0.08)$) {};
\node[token] at ($(2)+(-0.08,-0.08)$) {};
\node[token] at ($(3)+(-0.08,0.08)$) {};

\node[] at ($(2)+(0,-1.7)$) (Cc) {Configuration $\gst_2$};
\end{scope}

\begin{scope}[local bounding box=scope1]

\node[coordinate] (f1) at (-2.5,-0.15) {};
\node[coordinate] (f2) [right = 1cm of f1] {};
\node[coordinate] (f3) [right = 1cm of f2] {};
\node[coordinate] (f4) [right = 1cm of f3] {};
\node[coordinate] (f5) [right = 1cm of f4] {};
\node[coordinate] (f6) [right = 1cm of f5] {};

\node[coordinate] (s1) [below = 1.3cm of f1] {};
\node[coordinate] (s2) [right = 1cm of s1] {};
\node[coordinate] (s3) [right = 1cm of s2] {};
\node[coordinate] (s4) [right = 1cm of s3] {};
\node[coordinate] (s5) [right = 1cm of s4] {};
\node[coordinate] (s6) [right = 1cm of s5] {};

 \draw[trans] (f1) to[] 
	  node[align=center,anchor=south,midway]{{$r_1$}}(f2);
 \draw[transt] (f2) to[] 
	  node[align=center,anchor=south,midway]{{$r_6$}} (f3);
 \draw[trans] (f3) to[] 
	  node[align=center,anchor=south,midway]{{$r_4$}} (f4);
 \draw[transt] (f4) to[]  
	  node[align=center,anchor=south,midway]{{$r_2$}}(f5);
 \draw[trans] (f5) to[] 
	  node[align=center,anchor=south,midway]{{$r_4$}} (f6);

 \draw[transt] (s1) to[]  
	  node[align=center,anchor=north,midway]{{$r_6$}}(s2);
 \draw[transt] (s2) to[] 
	  node[align=center,anchor=north,midway]{{$r_2$}} (s3);
 \draw[trans] (s3) to[]  
	  node[align=center,anchor=north,midway]{{$r_1$}}(s4);
 \draw[trans] (s4) to[] 
	  node[align=center,anchor=north,midway]{{$r_4$}} (s5);
 \draw[trans] (s5) to[] 
	  node[align=center,anchor=north,midway]{{$r_4$}} (s6);

\draw[-,dashed] (f2) to[] (s1);
\draw[-,dashed] (f3) to[] (s2);
\draw[-,dashed] (f4) to[] (s2);
\draw[-,dashed] (f5) to[] (s3);

 \draw [decorate,decoration={brace,amplitude=9pt}]
($(s3)+(0,-0.5)$) -- ($(s1)+(0,-0.5)$) node (a) [black,midway,yshift=-17pt] 
{One thread};
 \draw [decorate,decoration={brace,amplitude=9pt}]
($(s6)+(0,-0.5)$) -- ($(s3)+(0,-0.5)$) node (b) [black,midway,yshift=-17pt] 
{All other threads};
\end{scope}

\end{tikzpicture}
 
\end{center}
\caption{Example of constructing a representative schedule by moving a
  thread to the beginning. The number of dots in the local states correspond to
  counter values, i.e.,
  $\sigma_1.\counters[\ell_0]= \sigma_1.\counters[\ell_1]=
  \sigma_1.\counters[\ell_2] = 1$
and $\sigma_1.\counters[\ell_3] =0$.}
\label{fig:movingthread}
\end{figure*}
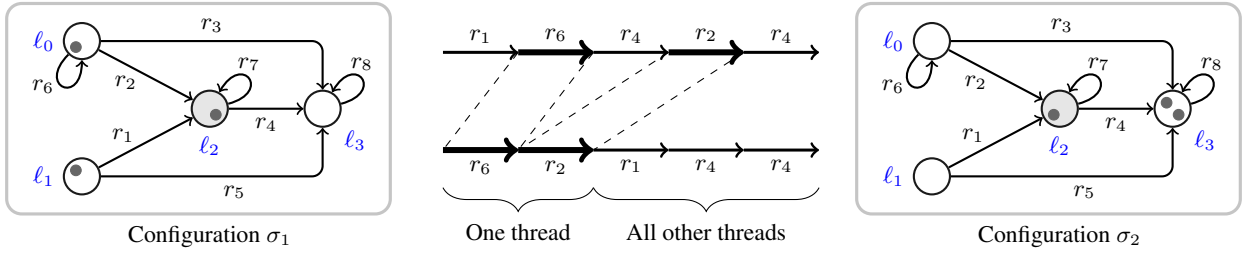

Another challenge in verification of  $\ELTLTB{}$
     formulas 
     is that counterexamples to liveness properties are infinite
     paths.
As discussed in Section~\ref{sec:counterexamples}, we consider
     infinite paths of lasso shape $\vartheta \concat \rho^\omega$.
For a finite part of a schedule, $\vartheta \concat \rho$, satisfying
     an $\ELTLTB{}$ formula, we show the existence of a new schedule,
     $\vartheta' \concat \rho'$, of bounded length satisfying the
     same formula as the original one.
Regarding the shortening, our  approach uses a similar idea as the one
     from~\cite{KVW15:CAV}.
We follow modified steps from reachability analysis:   

\begin{enumerate}
\item We split $\vartheta\concat\rho$ into several steady schedules,
     using cut points introduced in Sections~\ref{sec:counterexamples}
     and~\ref{sec:structguards}.
The cut points depend not only on threshold guards, but also on the
     $\ELTLTB{}$ formula~$\varphi$ representing  the negation of a specification
     we want to check.
Given such a steady schedule~$\tau$, each configuration of~$\tau$
satisfies a set of propositional subformulas of~$\varphi$, which are
covered by the operator~$\ltlG$ in $\varphi$.

\item For each of these steady schedules we find a representative,
  that is, an accelerated schedule of bounded length that satisfies
     the necessary propositional subformulas as in the original schedule (i.e., not just
     that starting and ending configurations coincide).

\item We  concatenate the obtained representatives in the
     original order.
\end{enumerate}

In \cite{KLVW16:arxiv}, we present the mathematical details for obtaining
     these representative schedules, and prove different cases that
     taken together establish our following main theorem:

\begin{theorem}\label{thm:main}
Let $\Sk = (\local, \initlocal, \globset,
     \paraset, \ruleset,\ResCond)$ be  a threshold automaton, and
let $\mathit{Locs}\subseteq\local$ be a set of locations.
Let~$\gst$ be a confi\-gu\-ration, let~$\tau$ be a steady conventional schedule
    applicable to~$\gst$, and let~$\psi$ be one of the following formulas:
    $$\bigvee_{\ell\in \mathit{Locs}} \counters[\ell]\neq 0, \;\mbox{ or }\;
\bigwedge_{\ell\in \mathit{Locs}} \counters[\ell]= 0.$$
If all configurations visited by $\tau$ from $\sigma$
satisfy~$\psi$, i.e.,
$\setconf\gst\tau \models \psi$, then
    there is a steady representative schedule~$\reprlive$ with the following properties:
\begin{enumerate}
        \item[a)] The representative is applicable, and ends in the same final state:\\ $\reprlive$ is applicable to~$\gst$,
		  and $\reprlive(\gst)=\tau(\gst)$,
        \item[b)] The representative has bounded length: $|\reprlive|\leq 6\cdot |\ruleset |$,
        \item[c)] The representative maintains the formula $\psi$. In
          other words, $\setconf{\gst}{\reprlive} \models \psi$,
        \item[d)] The representative is a concatenation of
          three representative schedules $\mathsf{srep}$ from Proposition~\ref{prop:srep-ex}:\\
        there exist~$\tau_1$,~$\tau_2$ and~$\tau_3$, (possibly empty) subschedules of~$\tau$,
        such that $\tau_1\cdot\tau_2\cdot\tau_3$ is applicable to~$\gst$, and it holds that 
        $(\tau_1\cdot\tau_2\cdot\tau_3)(\gst)=\tau(\gst)$, 
        and $\reprlive=\sr\Ctx{\gst}{\tau_1}\cdot \sr\Ctx{\tau_1(\gst)}{\tau_2}\cdot
        \sr\Ctx{(\tau_1\cdot\tau_2)(\gst)}{\tau_3}$.
\end{enumerate}
  \end{theorem}

Our approach is slightly different in the case when the formula~$\psi$
     has a more complex form:  $\bigwedge_{1\le\gcri\le
     \gcritical}  \bigvee_{\ell\in \critical_\gcri}
     \counters[\ell]\neq 0$, for $\critical_\gcri \subseteq \local$,
     where $1\le \gcri\le\gcritical$ and $\gcritical\in\Natural$.
In this case, our proof requires the schedule $\tau$ to have
     sufficiently large counter values.
To ensure that there is an infinite schedule with  sufficiently large
     counter values,  we first prove that if a counterexample exists
     in a small system, there also exists one in a larger system, that
     is, we consider configurations where each counter is multiplied
     with a constant \emph{finite multiplier}~$\multipl$.
For resilience conditions that do not correspond to parameterized
     systems (i.e., fix the system size to, e.g., $n=4$) or
     pathological threshold automata, such multipliers may not exist.
However, all our benchmarks have multipliers, and existence of
     multipliers can easily be checked using simple queries to SMT
     solvers in preprocessing.
This additional restriction leads to slightly smaller bounds on the
     lengths of representative schedules:

\newcommand{\thmandor}{Fix a threshold automaton~$\Sk = (\local, \initlocal, \globset,
     \paraset, \ruleset,\ResCond)$ that has a finite multiplier~$\mu$, and a confi\-gu\-ration $\gst$.
For an  $\gcritical\in\Natural$, fix sets of locations
    $\mathit{Locs}_\gcri \subseteq \local$ for $1\le \gcri\le\gcritical$.
If $\psi =
   \bigwedge_{1\le\gcri\le \gcritical} \bigvee_{\ell\in \mathit{Locs}_\gcri} \counters[\ell]\neq 0,$ 
then  for every steady conventional schedule $\tau$, 
    applicable to $\gst$,
    with $\setconf \gst\tau \models \psi$,
    there exists a schedule $\gsrogen$ with the following properties:
\begin{enumerate}
        \item[a)]  The representative is applicable and ends in the same final
          state:\\
$\gsrogen$ is a steady schedule applicable to $\multist{\gst}{\multipl}$, 
            and $\gsrogen (\multist{\gst}{\multipl})=\multisch{\tau}{\multipl}(\multist{\gst}{\multipl})$,
        \item[b)]  The representative has bounded length: $|\gsrogen|\leq  4\cdot |\ruleset |$,
        \item[c)]   The representative maintains the formula
          $\psi$. In other words, $\setconf{\multist{\gst}{\multipl}}\gsrogen \models \psi$,
        \item[d)]  The representative is a concatenation of
          two representative schedules $\mathsf{srep}$ from Proposition~\ref{prop:srep-ex}:\\  
         $\gsrogen=\sr\Ctx{\mu\gst}{\tau}\cdot \sr\Ctx{\tau(\mu\gst)}{\multisch{\tau}{(\multipl-1)}}$.
\end{enumerate}}

\begin{theorem}\label{thm:andor}
\thmandor
\end{theorem}
  
The main technical challenge for proving Theorems~\ref{thm:main}
     and~\ref{thm:andor} is that we want to swap transitions and
     maintain $\ELTLTB{}$ formulas  at the same time.
As discussed in Example~\ref{ex:swap}, simply applying the ideas from
     the reachability analysis in~\cite{Lipton75,KVW15:CAV} is not
     sufficient.

We address this challenge by more refined swapping strategies
     depending on the property $\psi$ of Theorem~\ref{thm:main}.
For instance, the intuition behind $\bigvee_{\ell\in \mathit{Locs}}
     \counters[\ell]\neq 0$ is that in a given distributed algorithm,
     there should always be at least one process in one of the states
     in $\mathit{Locs}$.
Hence, we would like to consider individual processes, but in the
     context of counter systems.
Therefore, we introduce a mathematical notion we call a \emph{thread},
     which is a schedule that can be executed by an individual
     process.
A thread is then characterized depending on whether it starts in
     $\mathit{Locs}$, ends in $\mathit{Locs}$, or visits
     $\mathit{Locs}$ at some intermediate step.
Based on this characterization, we show that $\ELTLTB{}$ formulas are
     preserved if we move carefully chosen threads to the beginning of
     a steady schedule (intuitively, this corresponds to $\tau_1,$ and
     $\tau_2$ from Theorem~\ref{thm:main}).
Then, we replace the threads, one by one, by their representative
     schedules from Proposition~\ref{prop:srep-ex}, and append another
     representative schedule for the remainder of the schedule.
In this way, we then obtain the representative schedules in
     Theorem~\ref{thm:main}(d).

\begin{example}\label{ex:slides}
We
     consider the~$\Sk$ in Figure~\ref{fig:stunningexample}, and show
     how a schedule $\tau=(r_1,1),(r_6,1),(r_4,1),(r_2,1),(r_4,1)$
     applicable to~$\gst_1$, with~$\tau(\gst_1) = \gst_2$ can be
     shortened.
Figure~\ref{fig:movingthread} follows this example where $\tau$ is the
     upper schedule.
Assume that $\setconf{\sigma_1}{\tau}\models\counters[\ell_2]\ne 0$,
     and  that we want to construct a shorter schedule that produces a
     path that satisfies the same formula.

In our theory, subschedule $(r_1,1),(r_4,1)$ is a thread of $\gst_1$
     and~$\tau$ for two reasons: (1) the counter of the starting local
     state of~$(r_1,1)$ is greater than~$0$, i.e.,
     $\sigma_1.\counters[\ell_0]=1$, and (2) it is a sequence of rules
     in the control flow of the threshold automaton, i.e., it
     starts from~$\ell_0$, then uses $(r_1,1)$ to go to local state
     $\ell_2$ and then $(r_4,1)$ to arrive at $\ell_3$.
The intuition of (2) is that a thread corresponds to a process
     that executes the threshold automaton.
Similarly, $(r_6,1),(r_2,1)$ and $(r_4,1)$ are also threads of
     $\gst_1$ and~$\tau$.
In fact, we can show that each schedule can be decomposed into threads.
Based on this, we analyze which local states are visited when a thread
     is executed.
  
Our formula $\setconf{\sigma_1}{\tau}\models\counters[\ell_2]\ne 0$ talks about~$\ell_2$.
Thus, we are interested in a thread that ends at~$\ell_2$,
     because after executing this thread, intuitively there will
     always be at least one process in~$\ell_2$, i.e., the
     counter~$\counters[\ell_2]$ will be nonzero, as required.
Such a thread will be moved to the beginning.
We find that thread~$(r_6,1),(r_2,1)$ meets this requirement.
Similarly, we are also interested in a thread that starts from~$\ell_2$.
Before we execute such a thread, at least one process must always be
     in~$\ell_2$, i.e.,~$\counters[\ell_2]$ will be nonzero.
For this, we single out the thread~$(r_4,1)$, as it starts
     from~$\ell_2$.

Independently of the actual positions of these threads within a
     schedule, our condition $\counters[\ell_2]\ne 0$ is true
     \emph{before}~$(r_4,1)$ starts, and \emph{after}~$(r_6,1),(r_2,1)$ ends.
Hence, we move the thread~$(r_6,1),(r_2,1)$ to the beginning,
     and obtain a schedule that ensures our condition in all visited
     configurations; cf.\ the lower schedule in Figure~\ref{fig:movingthread}.
Then we replace the thread~$(r_6,1),(r_2,1)$, by a representative schedule
     from Proposition~\ref{prop:srep-ex},  and the remaining part~$(r_1,1)$,
     $(r_4,1)$, $(r_4,1)$,
     by another one.
Indeed in our example, we could merge $(r_4,1),(r_4,1)$ into one accelerated
     transition $(r_4,2)$ and obtain a schedule which is shorter
     than~$\tau$ while maintaining $\counters[\ell_2]\ne 0$.
\end{example}

\makeatletter{}\section{Application of the Short Counterexample Property and 
Experimental Evaluation}

\subsection{SMT Encoding}
\label{sec:smtencodings}

We use the theoretical results from the previous section to give an
     efficient encoding of lasso-shaped executions in SMT with linear
     integer arithmetic.
The definitions of counter systems in Section~\ref{sec:countsys}
     directly tell us how to encode paths of the counter system.
Definition~\ref{def:config} describes a configuration $\sigma$ as
     tuple $(\counters,\vars,\param)$, where each component is encoded
     as a vector of SMT integer variables.
Then, given a path $\gst_0, t_1, \gst_1, \dots, t_{k-1}, \gst_{k-1},
     t_k, \dots \gst_{k}$ of length~$k$, by $\counters^i$,
     $\vars^i$, and $\param^i$ we denote the values of the vectors that
     correspond to $\sigma_i$, for $0\le i \le k$.
As the parameter values do not change, we use one copy of the
     variables $\param$ in our SMT encoding.
By $\counters^i_\ell$, for $1 \le \ell \le |\local|$, we denote the
     $\ell$th component of $\counters^i$, that is, the counter
     corresponding to the number of processes in local
     state~$\ell$ after the $i$th iteration.
Definition~\ref{def:config} also gives us the constraint
on the initial states, namely:
\begin{equation}\label{eq:smt-init}
    \mathrm{init}(0) \equiv \sum_{\ell \in \initlocal}
\counters^0_\ell = \syssize(\param) \wedge
     \sum_{\ell \not\in \initlocal} \counters^0_\ell = 0 \wedge
     \vars^0=\vec{0} \wedge RC(\param)
\end{equation}

\begin{example}
\label{ex:init}
The TA from Figure~\ref{fig:stunningexample} has four
local states $\ell_0$, $\ell_1$, $\ell_2$, $\ell_3$ among which $\ell_0$ and
$\ell_1$ are the initial states.
In this example, $\syssize(\param)$ is $n-f$, and the resilience
condition requires that there are less than a third of the processes
faulty, i.e., $n>3t$. We obtain
$
\mathrm{init}(0) \equiv \counters^0_0 +
\counters^0_1 = n-f \wedge  \counters^0_2 +
\counters^0_3 = 0 \wedge
     x^0=0 \wedge n > 3t \wedge t\ge f \wedge f \ge 0
$.
The constraint is in linear integer arithmetic.
\end{example}

Further, Definition~\ref{def:TofSigma} encodes the transition
     relation.
A transition is identified by a rule and an acceleration factor.
A rule is identified by threshold guards $\precondLE$ and $\precondG$,
     local states $\mathit{from}$ and $\mathit{to}$ between which processes
     are moved, and by $\update$, which defines the increase of shared
     variables.
As according to Section~\ref{sec:structguards} only a fixed number of
     transitions change the context and thus may change the evaluation
     of $\precondLE$ and $\precondG$, we do not encode $\precondLE$
     and $\precondG$ for each rule.
In fact, we check the guards $\precondLE$ and $\precondG$ against a
     fixed number of configurations, which correspond to the cut
     points defined by the threshold guards.
The acceleration factor~$\delta$ is indeed the only variable in a
     transition, and the SMT solver has to find assignments of these
     factors.
Then this transition from the $i$th to the $(i+1)$th configuration is
     encoded using rule $r=(\fromstate, \tostate, \precondLE,
     \precondG,\update)$ as follows: 
\begin{eqnarray}
    T(i,r) &\equiv& \mathit{Move}(\mathit{from}, \mathit{to}, i)
        \wedge \mathit{IncShd}(\update, i) \label{eq:smt-rule}\\
        \mathit{Move}(\ell, \ell', i) &\equiv& 
\ell \ne \ell'
  \rightarrow \counters^{i}_\ell - \counters^{i+1}_\ell
  = \delta^{i+1} = \counters^{i+1}_{\ell'} - \counters^{i}_{\ell'} \notag\\
    &\wedge& \ell = \ell'
\rightarrow \big( \counters^{i}_\ell = \counters^{i+1}_\ell
\wedge \counters^{i+1}_{\ell'} = \counters^{i}_{\ell'} \big) \notag\\
        &\wedge&
  \bigwedge_{ s \in \local \setminus \{\ell, \ell'\}}
    \counters^{i}_s =  \counters^{i+1}_s \notag\\
    \mathit{IncShd}(\update, i) &\equiv&\vars^{i+1} - \vars^i = \delta^{i+1} \cdot \update \notag
\end{eqnarray}

Given a schedule~$\tau$, we encode in linear integer arithmetic the paths that
    follow this schedule from an initial state as follows: $$E(\tau) \equiv
    \mathrm{init}(0) \wedge T(0,r_1) \wedge T(1,r_2)  \wedge \dots$$ We can now
    ask the SMT solver for assignments of the parameters as well as the factors
    $\delta^1, \delta^2, \dots$ in order to check whether a path with this
    sequence of rules exists. 
Note that some factors can be equal to~0, which means that the corresponding
    rule does not have any effect (because no process executes it). 
If $\tau$ encodes a lasso shape, and the SMT solver reports a satisfying
    assignment, this assignment is a counterexample. 
If the SMT solver reports unsat on all lassos discussed in
    Section~\ref{sec:structguards}, then there does not exists a counterexample
    and the algorithm is verified.

\begin{example}
In Example~\ref{ex:fair} we have seen the fairness requirement
$\propfair$, which is a property of a configuration that can be
encoded as
$\mathrm{fair}(i) \equiv \counters^i_1 = 0 \wedge 
     (x^i\ge t+1 \rightarrow \counters^i_0 = 0 \wedge \counters^i_1 = 0) \wedge
     (x^i\ge n-t \rightarrow \counters^i_0 = 0 \wedge
\counters^i_2 = 0)$, which is a formula in linear integer
arithmetic. Then, e.g., $\mathrm{fair}(5)$ encodes that the
fifth configuration satisfies the predicate. Such state
formulas can be added as conjunct to the formula $E(\tau)$
that encodes a path.
\end{example}

As discussed in Sections~\ref{sec:counterexamples}
    and~\ref{sec:structguards} we have to encode
    lassos of the form $\vartheta \cdot
     \rho^\omega$ starting from an initial configuration~$\sigma$. We
     immediately obtain a finite representation by encoding the
     fixed length execution $E(\vartheta \cdot
     \rho)$ as above, and adding the constraint that applying $\rho$
     returns to the start of the lasso loop, that is,
     $\vartheta(\sigma) = \rho(\vartheta(\sigma))$. In SMT this is
     directly encoded as equality of integer variables.

\subsection{Generating the SMT Queries}\label{sec:one-order}

\begin{figure}[t]
\begin{lstlisting}[language=pseudo,numbers=left,numberstyle=\scriptsize,
    columns=fullflexible]
variables fn, fs; // the current configuration number and the loop start
// Try to find a witness lasso for: a threshold automaton $\TA$,
// an $\ELTLTB$ formula $\varphi$, a cut graph $\pgraph$, a threshold graph $\tgraph$, and
// a topological order $\prec$ on the nodes of $\pgraph \cup \tgraph$.
procedure check_one_order($\TA$, $\varphi$, $\pgraph$, $\tgraph$, $\prec$):
  fn := 0; fs := 0;
  SMT_start(); // start (or reset) the SMT solver
  assume($\canform(\varphi) = \psi_0 \wedge \ltlF \psi_1 \wedge \dots \wedge \ltlF \psi_k \wedge \ltlG \psi_{k+1}$);
  SMT_assert($\counters^{0}, \vars^{0}, \param \models \mathsf{init}(0) \wedge \psi_0 \wedge \psi_{k+1}$); // see Equation @\ref{eq:smt-init}@
  $v_0 := \min_\prec (\pgvertices \cup \tgvertices)$; // the minimal node w.r.t. the linear order $\prec$
  check_node($\pgraph$, $\tgraph$, $\prec$, $v_0$, $\psi_{k+1}$, $\emptyset$);

// Try to find a witness lasso starting with the node $v$ and the context $\Ctx$,
// while preserving the invariant $\psi_{inv}$.
recursive procedure check_node($\pgraph$, $\tgraph$, $\prec$, $v$, $\psi_{inv}$, $\Ctx$):
  if not SMT_sat() then:
    return no_witness;
  case (a) $v \in \pgvertices \setminus \{\loops, \loope\}$:
    find $\psi$ s.t. $\left<\ltlF \psi, v\right> \in \syntree(\varphi)$; // $v$ labels a formula in the syntax tree
    assume($\psi = \psi_0 \wedge \ltlF \psi_1 \wedge \dots \wedge \ltlF \psi_k \wedge \ltlG \psi_{k+1}$);
    SMT_assert($\counters^{\mathsf{fn}}, \vars^{\mathsf{fn}}, \param \models \psi_0$);@\label{line:Fwitness}@
    push_segment($\psi_{inv} \wedge \psi_{k+1}$);@\label{line:FGinv}@
    $v' := \min_\prec (\pgvertices \cup \tgvertices) \cap \{w: v \prec w \}$; // the next node after $v$
    check_node($\pgraph$, $\tgraph$, $\prec$, $v'$, $\psi_{inv} \wedge \psi_{k+1}$, $\Ctx$);
  case (b) $v \in \tgvertices \setminus \{\loops, \loope\}$: // $v$ is a threshold guard
    if $v \in \PrecondU$ then: // $v$ is an unlocking guard, e.g., $x \ge t+1-f$
      push_segment($\psi_{inv}$);@\label{line:guard-unlock}@ // one rule unlocks $v$
      SMT_assert($\counters^{\mathsf{fn}}, \vars^{\mathsf{fn}}, \param \models v$); // $v$ is unlocked
      push_segment($\psi_{inv}$);@\label{line:guard-segment}@ // execute all unlocked rules
      $v' := \min_\prec (\pgvertices \cup \tgvertices) \cap \{w: v \prec w \}$; // the next node after $v$
      check_node($\pgraph$, $\tgraph$, $\prec$, $v'$, $\psi_{inv}$, $\Ctx \cup \{v\}$);
    else: /* $v \in \PrecondL$, e.g., $x < f$, similar to the locking case: use $\neg v$ */
  case (c) $v = \loops$:
    fs := fn; // the loop starts at the current configuration
    push_segment($\psi_{inv}$); // execute all unlocked rules
    $v' := \min_\prec (\pgvertices \cup \tgvertices) \cap \{w: v \prec w \}$; // the next node after $v$
    check_node($\pgraph$, $\tgraph$, $\prec$, $v'$, $\psi_{inv}$, $\Ctx$);
  case (d) $v = \loope$:
    SMT_assert($\counters^{\mathsf{fn}} = \counters^{\mathsf{fs}} \wedge \vars^{\mathsf{fn}} = \vars^{\mathsf{fs}}$); @\label{line:close-loop}@// close the loop
    if SMT_sat() then:
      return witness(SMT_model())

// Encode a segment of rules as prescribed by @\cite{KVW15:CAV}@ and Theorems @\ref{thm:main}--\ref{thm:andor}@.
procedure push_segment($\psi_{inv}$):
  // find the number of schedules to repeat in (d) of Theorems @\ref{thm:main}, \ref{thm:andor}@
  nrepetitions := compute_repetitions($\psi_{inv}$);
  $r_1, \dots, r_k$ := compute_rules($\Ctx$); // use $\xschema_\Ctx$ from @\cite{KVW15:CAV}@
  for _ from 1 to nrepetitions:
    for j from 1 to $k$:
      SMT_assert($\counters^{\mathsf{fn}}, \vars^{\mathsf{fn}}, \param \models \psi_{inv}$);
      SMT_assert($T(\mathsf{fn}, r_j)$); // modify the counters as in Equation @\ref{eq:smt-rule}@
      fn := fn + 1; // move to the next configuration
\end{lstlisting}
\caption{Checking one topological order with SMT.}
\label{fig:pseudosmt}
\end{figure}

The high-level structure of the verification algorithm is given in
    Figure~\ref{fig:pseudo} on page~\pageref{fig:pseudo}. 
In this section, we give the details of the
    procedure~\texttt{check\_one\_order}, whose pseudo code is given
    in Figure~\ref{fig:pseudosmt}. 
It receives as the input the following parameters: a threshold
    automaton~$\TA$, an $\ELTLTB$ formula~$\varphi$, a cut graph $\pgraph$ of
    $\varphi$, a threshold graph~$\tgraph$ of~$\TA$, and a topological
    order~$\prec$ on the vertices of the graph $\pgraph \cup \tgraph$.

The procedure~\texttt{check\_one\_order} constructs SMT assertions
    about the configurations of the lassos that correspond to the
    order~$\prec$.
As explained in Section~\ref{sec:smtencodings}, an ith configuration is defined
    by the vectors of SMT variables $(\counters^i, \vars^i, \param)$. 
We use two global variables: the number~$\mathsf{fn}$ of the configuration
    under construction, and the number~$\mathsf{fs}$ of the configuration that
    corresponds to the loop start. 
Thus, with the expressions $\counters^{\mathsf{fn}}$ and $\vars^{\mathsf{fn}}$
    we refer to the SMT variables of the configuration whose number is stored
    in~$\mathsf{fn}$.

In the pseudocode in Figure~\ref{fig:pseudosmt}, we call
     \texttt{SMT\_assert($\counters^{\mathsf{fn}}$,
     $\vars^{\mathsf{fn}}$, $\param \models \psi$)} to add an
     assertion~$\psi$ about the configuration
     $(\counters^{\mathsf{fn}}, \vars^{\mathsf{fn}}, \param)$ to the
     SMT query.
Finally, the call \texttt{SMT\_sat()} returns true, only if there is a
     satisfying assignment for the assertions collected so far.
Such an assignment can be accessed with \texttt{SMT\_model()} and
     gives the values for the configurations and acceleration factors,
     which together constitute a witness lasso.

The procedure~\texttt{check\_one\_order} creates the assertions about the
    initial configurations. 
The assertions consist of: the assumptions~$\mathsf{init}(0)$ about the initial
    configurations of the threshold automaton, the top-level propositional
    formula~$\psi_0$, and the invariant propositional formula~$\psi_{k+1}$ that
    should hold from the initial configuration on. 
By writing \texttt{assume($\psi = \psi_0 \wedge \ltlF \wedge \psi_1 \dots \ltlF
    \psi_k \wedge \ltlG \psi_{k+1}$)}, we extract the subformulas of a
canonical formula~$\psi$ (see Section~\ref{sec:shapelasso}). 
The procedure finds the minimal node in the order~$\prec$ on the nodes of the
    graph~$\pgraph \cup \tgraph$ and calls the procedure~\texttt{check\_node}
    for the initial node, the initial invariant~$\psi_{k+1}$, and the empty
    context~$\emptyset$.

\begin{figure*}[t]
    \begin{center}
    \includegraphics[width=.45\textwidth]{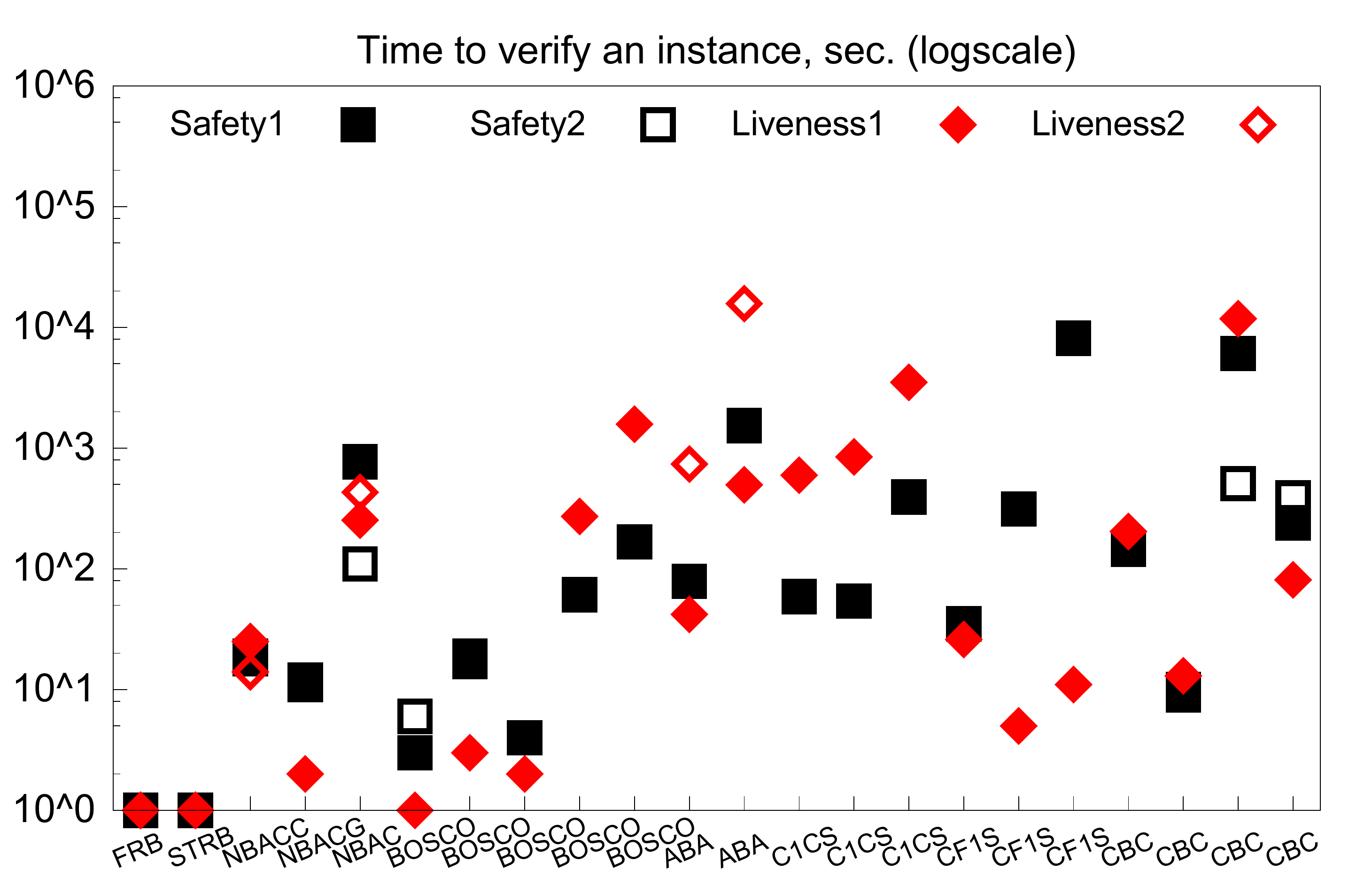}
    \includegraphics[width=.45\textwidth]{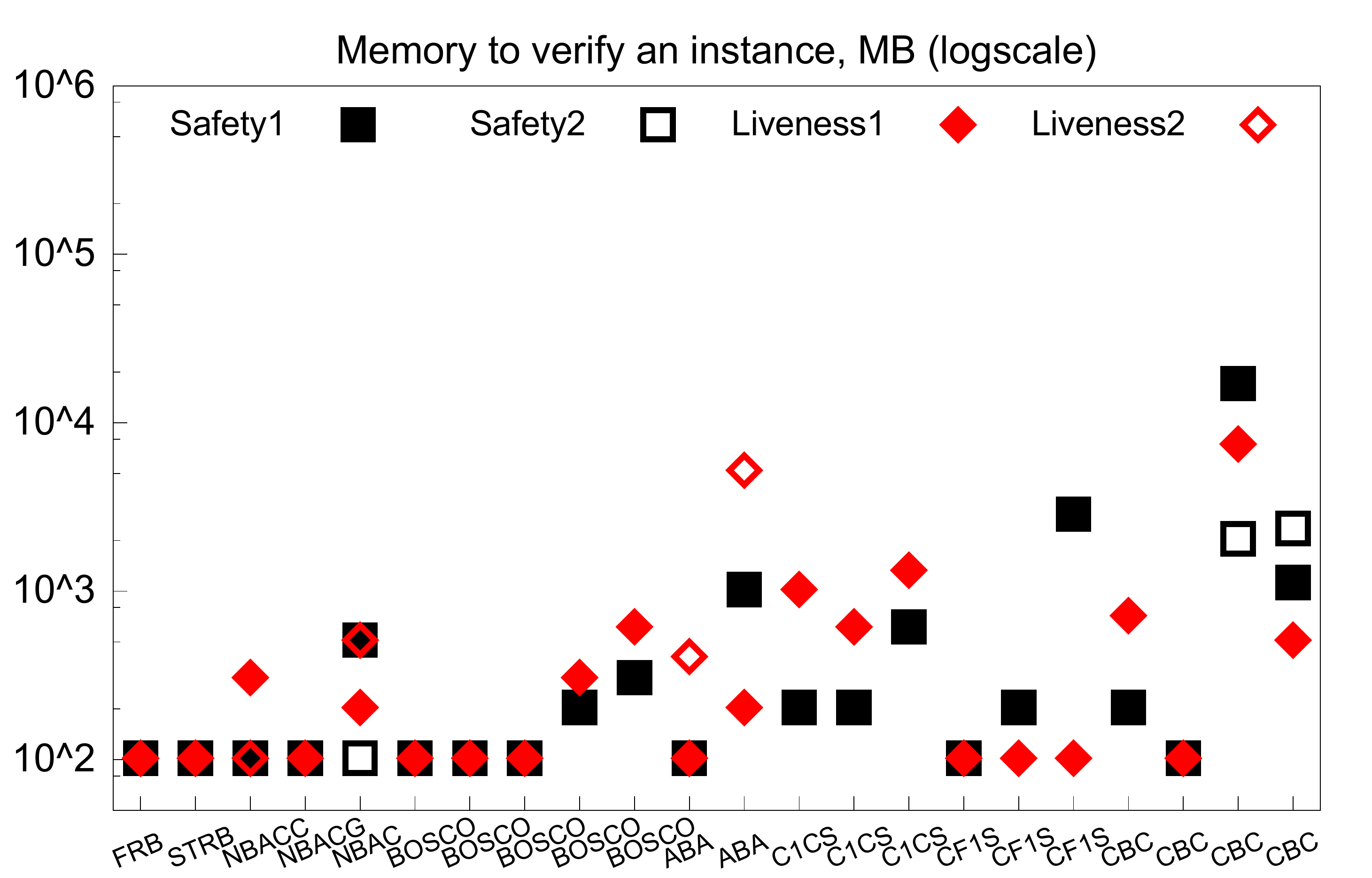}
    \includegraphics[width=.45\textwidth]{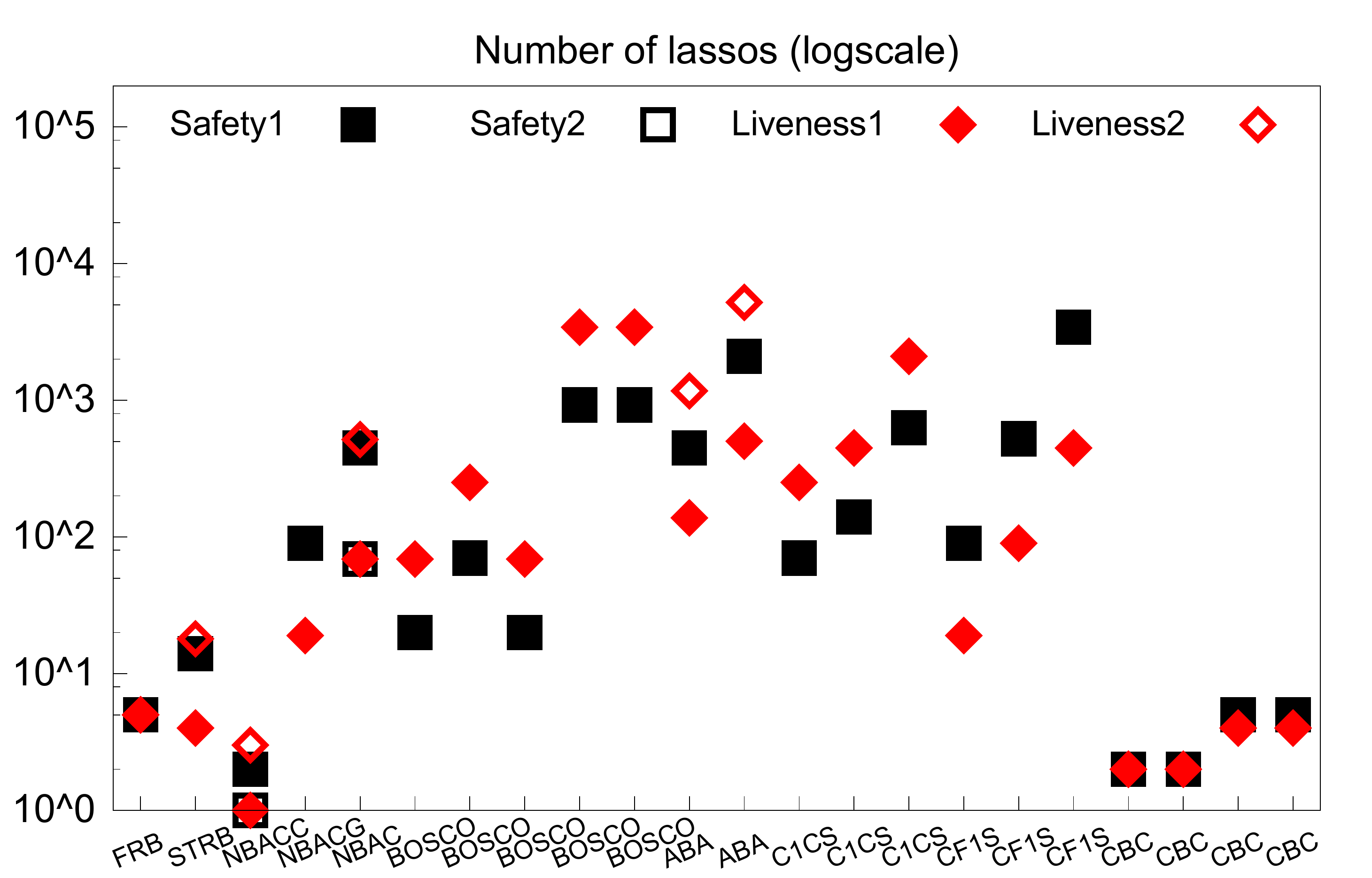}
    \includegraphics[width=.45\textwidth]{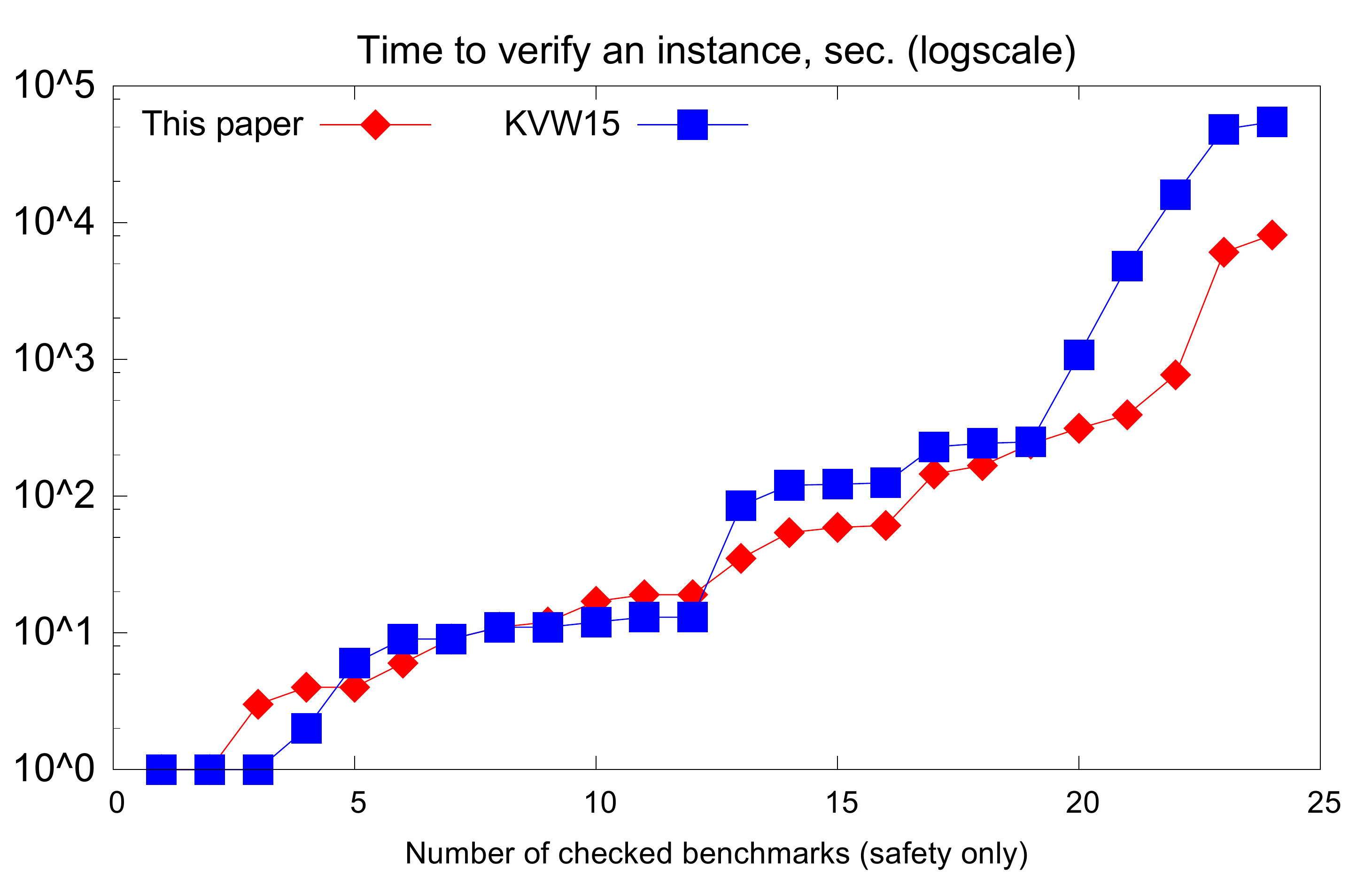}
    \end{center}

    \caption{The plots summarize the following results of running
        our implementation on all benchmarks: used time in seconds (top left),
        used memory in megabytes (top right),
        the number of checked lassos (bottom left),
        time used both by our implementation and~\cite{KVW15:CAV}
        to check \emph{safety only} (bottom right).
        Several occurrences
        of the same benchmark correspond to different cases,
        such as $f > 1$, $f=1$, and $f=0$.
        Symbols $\blacksquare$ and $\square$ correspond to the safety properties of
        each benchmark, while symbols $\blacklozenge$ and $\lozenge$ correspond to the liveness
        properties.
    }
\label{fig:plotsky}
\end{figure*}

The procedure \texttt{check\_node} is called with a node $v$ of the
     graph $\pgraph \cup \tgraph$ as a parameter.
It adds assertions that encode a finite path and constraints on the
     configurations of this path.
The finite path leads from the configuration that corresponds to the
     node $v$ to the configuration that corresponds to  $v$'s
     successor in the order~$\prec$.
The constraints depend on $v$'s origin: (a)~$v$ labels a
     formula~$\ltlF \psi$ in the syntax tree of~$\varphi$, (b)~$v$
     carries a threshold guard from the set~$\PrecondU \cup
     \PrecondL$, (c)~$v$ denotes the loop start, or (d)~$v$ denotes
     the loop end.
In case~(a), we add an SMT assertion that the current configuration
     satisfies the propositional formula $\prop(\psi)$
     (line~\ref{line:Fwitness}), and add a sequence of rules that
     leads to $v$'s successor while maintaining the
     invariants~$\psi_{\mathit{inv}}$ of the preceding nodes and the
     $v$'s invariant~$\psi_{k+1}$ (line~\ref{line:FGinv}).
In case~(b), in line~\ref{line:guard-unlock}, we add a sequence of
     rules, one of which should unlock (resp.
lock) the threshold guard in~$v \in \PrecondU$ (resp.
$v \in \PrecondL$).
Then, in line~\ref{line:guard-segment}, we add a sequence of rules
     that leads to a configuration of  $v$'s successor.
All added configurations are required to satisfy the current
     invariant~$\psi_{\mathit{inv}}$.
As the threshold guard in~$v$ is now unlocked (resp.
locked), we include the guard (resp.
its negation) in the current context~$\Ctx$.
In case~(c), we store the current configuration as the loop start in
     the variable~$\mathrm{fs}$ and, as in~(a) and~(b), add a sequence
     of rules leading to $v$'s successor.
Finally, in case~(d), we should have reached the ending configuration
     that coincides with the loop start.
To this end, in line~\ref{line:close-loop}, we add the constraint that
     forces the counters of both configurations to be equal.
At this point, all the necessary SMT constraints have been added, and
     we call \texttt{SMT\_sat} to check whether there is an assignment
     that satisfies the constraints.
If there is one, we report it as a lasso witnessing the
     $\ELTLTB$-formula~$\varphi$ that consists of: the concrete
     parameter values, the values of the counters and shared variables
     for each configuration, and the acceleration factors.
Otherwise, we report that there is no witness lasso for the
     formula~$\varphi$.

The procedure~\texttt{push\_segment} constructs a sequence of currently
    unlocked rules, as in the case of reachability~\cite{KVW15:CAV}. 
However, this sequence should be repeated several times, as required by
    Theorems~\ref{thm:main} and~\ref{thm:andor}. 
Moreover, the freshly added configurations are required to satisfy the current
    invariant~$\psi_{\mathit{inv}}$.

\subsection{Experiments}
\label{sec:subexp}

We extended the tool ByMC~\cite{KVW15:CAV} with our  technique and conducted
    experiments\footnote{The details on the experiments and the artifact
    are available at:\\
\url{http://forsyte.at/software/bymc/popl17-artifact}} with the freely
available benchmarks from~\cite{KVW15:CAV}: folklore reliable broadcast
(FRB)~\cite{CT96}, consistent broadcast (STRB)~\cite{ST87:abc}, asynchronous
Byzantine agreement (ABA)~\cite{BrachaT85}, condition-based consensus
(CBC)~\cite{MostefaouiMPR03}, non-blocking atomic commitment (NBAC and
NBACC~\cite{Raynal97} and NBACG~\cite{Gue02}), one-step consensus with zero
degradation (CF1S~\cite{DobreS06}), consensus in one communication step
(C1CS~\cite{BrasileiroGMR01}), and one-step Byzantine asynchronous consensus
(BOSCO~\cite{SongR08}). 
These threshold-guarded fault-tolerant distributed algorithms are encoded in a
    parametric extension of Promela.

Negations of the safety and liveness specifications of our benchmarks\dash---written in~$\ELTLTB$\dash---follow three patterns: unsafety
    $\ltlE (p \wedge \ltlF q)$, non-termination $\ltlE (p \wedge \ltlG
    \ltlF r \wedge \ltlG q)$, and non-response $\ltlE (\ltlG \ltlF r \wedge
    \ltlF (p \wedge \ltlG q))$. 
    The propositions $p$, $q$, and $r$ follow the syntax of~$\mathit{pform}$ (cf. Table~\ref{table:eltlft-syntax}), e.g., $p \equiv \bigwedge_{\ell
    \in \critical_1} \kappa[\ell] = 0$ and $q \equiv \bigvee_{\ell \in \critical_2} \kappa[\ell]
    \ne 0$ for some sets of locations~$\critical_1$ and~$\critical_2$.

The results of our experiments are summarized in Figure~\ref{fig:plotsky}. 
Given the properties of the distributed algorithms found in the literature, we
    checked for each benchmark one or two safety properties (depicted
    with~$\blacksquare$ and~$\square$) and one or two liveness properties
(depicted with~$\blacklozenge$ and~$\lozenge$). 
For each benchmark, we display the running times and the memory used together
    by~ByMC and the SMT solver~Z3~\cite{DeMouraB08}, as well as the number of
    exercised lasso shapes as discussed in Section~\ref{sec:structguards}.

For safety properties, we compared our implementation against the
    implementation of~\cite{KVW15:CAV}. 
The results are summarized the bottom right plot
    in Figure~\ref{fig:plotsky}, which shows that there is no clear winner.
For instance, our implementation is 170 times faster on BOSCO for the case $n >
    5t$. 
However, for the benchmark ABA we experienced a tenfold slowdown. 
In our experiments, attempts to improve the SMT encoding for liveness
    usually impaired safety results.

Our implementation has verified safety and liveness of all ten
     parameterized algorithms in less than a day.
Moreover, the tool reports counterexamples to liveness of CF1S and
     BOSCO exactly for the cases predicted by the distributed
     algorithms literature, i.e., when there are not enough correct
     processes to reach consensus in one communication step.
Noteworthy, liveness of only the two simplest benchmarks (STRB and
     FRB) had been automatically verified
     before~\cite{JohnKSVW13:fmcad}.

\makeatletter{}\section{Conclusions}
\label{sec:concl}

Parameterized verification approaches the problem
    of verifying systems of thousands of processes
    by proving correctness for all system sizes.
Although the literature predominantly deals with safety,
      parameterized verification for liveness is of growing
     interest, and has been addressed mostly in the context of
     programs that solve mutual exclusion or dining
     philosophers~\cite{Atig2012,FarzanKP16,PXZ02,FPPZ06}.
These techniques do not apply to fault-tolerant distributed algorithms
     that have arithmetic conditions on the fraction of faults,
     threshold~guards, and typical specifications that evaluate
     a global system state.

Parameterized verification is in general undecidable~\cite{AK86}.
As recently surveyed by Bloem et al.~\cite{2015Bloem}, one can escape
     undecidability by restricting, e.g., communication semantics,
     local state space, the local control flow, or the temporal logic
     used for specifications.
Hence, we make explicit the required
     restrictions. 
On the one hand, these restrictions still allow us to model
     fault-tolerant distributed algorithms and their specifications,
     and on the other hand, they give rise to a practical verification method.
The restrictions are on the local control flow (loops) of processes
     (Section~\ref{sec:TA}), as well as on the temporal operators and
     propositional formulas (Section~\ref{sec:reach-and-live}).
We conjecture that lifting these restrictions quite quickly leads to
     undecidability again.
In addition, we justify our restrictions with the considerable number
     of benchmarks~\casest\ that fit into our fragment, and with the
     convincing experimental results from Figure~\ref{fig:plotsky}.

Our main technical contribution is to combine and extend several
     important techniques: First, we extend the ideas by Etessami et
     al.~\cite{EtessamiVW02} to reason about shapes of infinite
     executions of lasso shape.
These executions are counterexample candidates.
Then we extend reductions introduced by Lipton~\cite{Lipton75}
 to deal with $\ELTLTB{}$
     formulas.
(Techniques that extend Lipton's in other directions can be found in~\movers.)
Our reduction is specific to threshold guards which are typical for
     fault-tolerant distributed algorithms and are found in
     domain-specific languages.
Using on our reduction we apply
     acceleration~\cite{BardinFLP08,KVW16:IandC} in order to arrive at
     our short counterexample property.

Our short counterexample property implies a completeness threshold,
     that is, a bound $b$ that ensures that if no lasso
     of length up to $b$ is
     satisfying an~$\ELTLTB$ formula, then there is no
     infinite path satisfying this formula.
For linear temporal logic with the $\ltlF$ and $\ltlG$ operators,
     Kroening et al.~\cite{KroeningOSWW11} prove bounds on the
     completeness thresholds on the level of B\"uchi automata.
Their bound involves the recurrence diameter of the transition
     systems, which is prohibitively large for counter
     systems.
Similarly, the general method to transfer liveness with fairness to
     safety checking by Biere et al.~\cite{BIERE2002160} leads to an
     exponential growth of the diameter, and thus to too
     large values of~$b$.
Hence, we decided to conduct an analysis on the level of threshold
     automata, accelerated counter systems, and a fragment of the
     temporal logic, which allows us to exploit specifics of the
     domain, and get bounds that can be used in practice.

Acceleration has been applied for parameterized verification
     by means of regular model
     checking~\cite{Pnueli2000,Bouajjani2004,Abdulla1998,SCHUPPAN200679}.
As noted by Fisman et al.~\cite{FismanKL08}, to verify
     fault-tolerant distributed algorithms, one would have to intersect
     the regular languages that describe sets of states with
     context-free languages that enforce the resilience condition
     (e.g., $n > 3t$).
Our approach of reducing to SMT handles resilience conditions
     naturally in linear integer arithmetic.

There are two reasons for our restrictions in the temporal logic: On one hand,
    in our benchmarks, there is no need to find counterexamples that contain a
    configuration that satisfies $\counters[\ell] = 0 \vee \counters[\ell'] =
    0$ for some~$\ell, \ell' \in \local$. 
One would only need such a formula to specify requirement that at least one
    process is at location~$\ell$ and at least one process is at
    location~$\ell'$ (the disjunction would be negated in the specification),
    which is unnatural for fault-tolerant distributed algorithms. 
On the other hand, enriching our logic with $\bigvee_{i \in \critical} \counters[i] =
    0$ allows one to express tests for zero in the counter system, which leads
    to undecidability~\cite{2015Bloem}. 
For the same reason, we avoid disjunction, as it would allow one to indirectly
    express test for zero: $\counters[\ell] = 0 \vee \counters[\ell'] = 0$.

The restrictions we put on threshold automata are justified from
     a practical viewpoint of our application domain, namely,
     threshold-guarded fault-tolerant algorithms.
We assumed that all the cycles in threshold automata
     are simple (while the benchmarks have only self-loops or cycles of length 2).
As our analysis already is quite involved, these restrictions allow us
     to concentrate on our central results without obfuscating the
     notation and theoretical results.
Still, from a theoretical viewpoint it might be interesting to relax
     the restrictions on cycles in the future.

More generally, these restrictions allowed us to develop a completely
     automated verification technique.
In general, there is a trade-off between degree of automation and
     generality.
Our method is completely automatic, but our input language cannot
     compete in generality with mechanized proof methods  that rely
     heavily on human expertise, e.g.,
     IVY~\cite{PadonMPSS16}, Verdi~\cite{WilcoxWPTWEA15},
     IronFleet~\cite{HawblitzelHKLPR15}, TLAPS~\cite{ChaudhuriDLM10}.

\bibliographystyle{abbrvnat}

\bibliography{lit}

\begin{thebibliography}{68}
\providecommand{\natexlab}[1]{#1}
\providecommand{\url}[1]{\texttt{#1}}
\expandafter\ifx\csname urlstyle\endcsname\relax
  \providecommand{\doi}[1]{doi: #1}\else
  \providecommand{\doi}{doi: \begingroup \urlstyle{rm}\Url}\fi

\bibitem[Abdulla et~al.(1998)Abdulla, Bouajjani, and Jonsson]{Abdulla1998}
P.~A. Abdulla, A.~Bouajjani, and B.~Jonsson.
\newblock On-the-fly analysis of systems with unbounded, lossy {FIFO} channels.
\newblock In \emph{CAV}, LNCS, pages 305--318, 1998.

\bibitem[Alberti et~al.(2016)Alberti, Ghilardi, and Pagani]{AlbertiGP16}
F.~Alberti, S.~Ghilardi, and E.~Pagani.
\newblock Counting constraints in flat array fragments.
\newblock In \emph{IJCAR}, volume 9706 of \emph{LNCS}, pages 65--81, 2016.

\bibitem[Apt and Kozen(1986)]{AK86}
K.~Apt and D.~Kozen.
\newblock Limits for automatic verification of finite-state concurrent systems.
\newblock \emph{IPL}, 15:\penalty0 307--309, 1986.

\bibitem[Atig et~al.(2012)Atig, Bouajjani, Emmi, and Lal]{Atig2012}
M.~F. Atig, A.~Bouajjani, M.~Emmi, and A.~Lal.
\newblock Detecting fair non-termination in multithreaded programs.
\newblock In \emph{CAV}, pages 210--226, 2012.

\bibitem[Baier and Katoen(2008)]{BK08}
C.~Baier and J.-P. Katoen.
\newblock \emph{Principles of model checking}.
\newblock MIT Press, 2008.

\bibitem[Bardin et~al.(2008)Bardin, Finkel, Leroux, and Petrucci]{BardinFLP08}
S.~Bardin, A.~Finkel, J.~Leroux, and L.~Petrucci.
\newblock Fast: acceleration from theory to practice.
\newblock \emph{STTT}, 10\penalty0 (5):\penalty0 401--424, 2008.

\bibitem[Biely et~al.(2013)Biely, Delgado, Milosevic, and Schiper]{BielyD0S13}
M.~Biely, P.~Delgado, Z.~Milosevic, and A.~Schiper.
\newblock {Distal:} a framework for implementing fault-tolerant distributed
  algorithms.
\newblock In \emph{DSN}, pages 1--8, 2013.

\bibitem[Biere et~al.(2002)Biere, Artho, and Schuppan]{BIERE2002160}
A.~Biere, C.~Artho, and V.~Schuppan.
\newblock Liveness checking as safety checking.
\newblock \emph{Electronic Notes in Theoretical Computer Science}, 66\penalty0
  (2):\penalty0 160–177, 2002.

\bibitem[Bloem et~al.(2015)Bloem, Jacobs, Khalimov, Konnov, Rubin, Veith, and
  Widder]{2015Bloem}
R.~Bloem, S.~Jacobs, A.~Khalimov, I.~Konnov, S.~Rubin, H.~Veith, and J.~Widder.
\newblock \emph{Decidability of Parameterized Verification}.
\newblock Synthesis Lectures on Distributed Computing Theory. Morgan {\&}
  Claypool Publishers, 2015.

\bibitem[Bouajjani et~al.(2004)Bouajjani, Habermehl, and Vojnar]{Bouajjani2004}
A.~Bouajjani, P.~Habermehl, and T.~Vojnar.
\newblock Abstract regular model checking.
\newblock In \emph{CAV}, LNCS, pages 372--386, 2004.

\bibitem[Bracha and Toueg(1985)]{BrachaT85}
G.~Bracha and S.~Toueg.
\newblock Asynchronous consensus and broadcast protocols.
\newblock \emph{J. ACM}, 32\penalty0 (4):\penalty0 824--840, 1985.

\bibitem[Brasileiro et~al.(2001)Brasileiro, Greve, Most{\'{e}}faoui, and
  Raynal]{BrasileiroGMR01}
F.~V. Brasileiro, F.~Greve, A.~Most{\'{e}}faoui, and M.~Raynal.
\newblock Consensus in one communication step.
\newblock In \emph{PaCT}, volume 2127 of \emph{LNCS}, pages 42--50, 2001.

\bibitem[Canfield and Williamson(1995)]{Canfield1995}
E.~R. Canfield and S.~G. Williamson.
\newblock A loop-free algorithm for generating the linear extensions of a
  poset.
\newblock \emph{Order}, 12\penalty0 (1):\penalty0 57--75, 1995.

\bibitem[Chandra and Toueg(1996)]{CT96}
T.~D. Chandra and S.~Toueg.
\newblock Unreliable failure detectors for reliable distributed systems.
\newblock \emph{J. ACM}, 43\penalty0 (2):\penalty0 225--267, 1996.

\bibitem[Charron-Bost and Merz(2009)]{Charron-BostM09}
B.~Charron-Bost and S.~Merz.
\newblock Formal verification of a consensus algorithm in the heard-of model.
\newblock \emph{IJSI}, 3\penalty0 (2--3):\penalty0 273--303, 2009.

\bibitem[Chaudhuri et~al.(2010)Chaudhuri, Doligez, Lamport, and
  Merz]{ChaudhuriDLM10}
K.~Chaudhuri, D.~Doligez, L.~Lamport, and S.~Merz.
\newblock Verifying safety properties with the {TLA+} proof system.
\newblock In \emph{IJCAR}, volume 6173 of \emph{LNCS}, pages 142--148, 2010.

\bibitem[Clarke et~al.(1999)Clarke, Grumberg, and Peled]{CGP1999}
E.~Clarke, O.~Grumberg, and D.~Peled.
\newblock \emph{Model Checking}.
\newblock MIT Press, 1999.

\bibitem[Clarke et~al.(2008)Clarke, Talupur, and Veith]{CTV2008}
E.~Clarke, M.~Talupur, and H.~Veith.
\newblock Proving {Ptolemy} right: the environment abstraction framework for
  model checking concurrent systems.
\newblock In \emph{TACAS'08/ETAPS'08}, pages 33--47. Springer, 2008.

\bibitem[Cohen and Lamport(1998)]{CohenL98}
E.~Cohen and L.~Lamport.
\newblock Reduction in {TLA}.
\newblock In \emph{CONCUR}, volume 1466 of \emph{LNCS}, pages 317--331, 1998.

\bibitem[De~Moura and Bj{\o}rner(2008)]{DeMouraB08}
L.~De~Moura and N.~Bj{\o}rner.
\newblock Z3: An efficient {SMT} solver.
\newblock In \emph{TACAS}, volume 1579 of \emph{LNCS}, pages 337--340. 2008.

\bibitem[Dobre and Suri(2006)]{DobreS06}
D.~Dobre and N.~Suri.
\newblock One-step consensus with zero-degradation.
\newblock In \emph{DSN}, pages 137--146, 2006.

\bibitem[Doeppner(1977)]{Doeppner77}
T.~W. Doeppner.
\newblock Parallel program correctness through refinement.
\newblock In \emph{POPL}, pages 155--169, 1977.

\bibitem[Dr{\u{a}}goi et~al.(2016)Dr{\u{a}}goi, Henzinger, and
  Zufferey]{DragoiHZ16}
C.~Dr{\u{a}}goi, T.~A. Henzinger, and D.~Zufferey.
\newblock {PSync:} a partially synchronous language for fault-tolerant
  distributed algorithms.
\newblock In \emph{POPL}, pages 400--415, 2016.

\bibitem[{Dr\u{a}goi} et~al.(2014){Dr\u{a}goi}, Henzinger, Veith, Widder, and
  Zufferey]{DHVWZ14}
C.~{Dr\u{a}goi}, T.~A. Henzinger, H.~Veith, J.~Widder, and D.~Zufferey.
\newblock A logic-based framework for verifying consensus algorithms.
\newblock In \emph{VMCAI}, volume 8318 of \emph{LNCS}, pages 161--181, 2014.

\bibitem[Elmas et~al.(2009)Elmas, Qadeer, and Tasiran]{Elmas09}
T.~Elmas, S.~Qadeer, and S.~Tasiran.
\newblock A calculus of atomic actions.
\newblock In \emph{POPL}, pages 2--15, 2009.

\bibitem[Emerson and Namjoshi(1995)]{EN95}
E.~Emerson and K.~Namjoshi.
\newblock Reasoning about rings.
\newblock In \emph{POPL}, pages 85--94, 1995.

\bibitem[Emerson and Kahlon(2003)]{EmersonK03LICS}
E.~A. Emerson and V.~Kahlon.
\newblock Model checking guarded protocols.
\newblock In \emph{LICS}, pages 361--370. IEEE, 2003.

\bibitem[Esparza et~al.(1999)Esparza, Finkel, and Mayr]{Esparza99}
J.~Esparza, A.~Finkel, and R.~Mayr.
\newblock On the verification of broadcast protocols.
\newblock In \emph{LICS}, pages 352--359. IEEE Computer Society, 1999.

\bibitem[Etessami et~al.(2002)Etessami, Vardi, and Wilke]{EtessamiVW02}
K.~Etessami, M.~Y. Vardi, and T.~Wilke.
\newblock First-order logic with two variables and unary temporal logic.
\newblock \emph{Inf. Comput.}, 179\penalty0 (2):\penalty0 279--295, 2002.

\bibitem[Fang et~al.(2006)Fang, Piterman, Pnueli, and Zuck]{FPPZ06}
Y.~Fang, N.~Piterman, A.~Pnueli, and L.~D. Zuck.
\newblock Liveness with invisible ranking.
\newblock \emph{{STTT}}, 8\penalty0 (3):\penalty0 261--279, 2006.

\bibitem[Farzan et~al.(2016)Farzan, Kincaid, and Podelski]{FarzanKP16}
A.~Farzan, Z.~Kincaid, and A.~Podelski.
\newblock Proving liveness of parameterized programs.
\newblock In \emph{LICS}, pages 185--196, 2016.

\bibitem[Fischer et~al.(1985)Fischer, Lynch, and Paterson]{FLP85}
M.~J. Fischer, N.~A. Lynch, and M.~S. Paterson.
\newblock Impossibility of distributed consensus with one faulty process.
\newblock \emph{J. ACM}, 32\penalty0 (2):\penalty0 374--382, 1985.

\bibitem[Fisman et~al.(2008)Fisman, Kupferman, and Lustig]{FismanKL08}
D.~Fisman, O.~Kupferman, and Y.~Lustig.
\newblock On verifying fault tolerance of distributed protocols.
\newblock In \emph{TACAS}, volume 4963 of \emph{LNCS}, pages 315--331.
  Springer, 2008.

\bibitem[Flanagan et~al.(2005)Flanagan, Freund, and Qadeer]{Flanagan:2005}
C.~Flanagan, S.~N. Freund, and S.~Qadeer.
\newblock Exploiting purity for atomicity.
\newblock \emph{IEEE Trans. Softw. Eng.}, 31\penalty0 (4):\penalty0 275--291,
  2005.

\bibitem[German and Sistla(1992)]{GS1992}
S.~M. German and A.~P. Sistla.
\newblock Reasoning about systems with many processes.
\newblock \emph{J. ACM}, 39:\penalty0 675--735, 1992.

\bibitem[Gmeiner et~al.(2014)Gmeiner, Konnov, Schmid, Veith, and
  Widder]{GKSVW14:SFM}
A.~Gmeiner, I.~Konnov, U.~Schmid, H.~Veith, and J.~Widder.
\newblock Tutorial on parameterized model checking of fault-tolerant
  distributed algorithms.
\newblock In \emph{Formal Methods for Executable Software Models}, LNCS, pages
  122--171. Springer, 2014.

\bibitem[Guerraoui(2002)]{Gue02}
R.~Guerraoui.
\newblock Non-blocking atomic commit in asynchronous distributed systems with
  failure detectors.
\newblock \emph{Distributed Computing}, 15\penalty0 (1):\penalty0 17--25, 2002.

\bibitem[Hawblitzel et~al.(2015)Hawblitzel, Howell, Kapritsos, Lorch, Parno,
  Roberts, Setty, and Zill]{HawblitzelHKLPR15}
C.~Hawblitzel, J.~Howell, M.~Kapritsos, J.~R. Lorch, B.~Parno, M.~L. Roberts,
  S.~T.~V. Setty, and B.~Zill.
\newblock Ironfleet: proving practical distributed systems correct.
\newblock In \emph{SOSP}, pages 1--17, 2015.

\bibitem[Holzmann(2003)]{H2003}
G.~Holzmann.
\newblock \emph{The SPIN Model Checker}.
\newblock Addison-Wesley, 2003.

\bibitem[John et~al.(2013)John, Konnov, Schmid, Veith, and
  Widder]{JohnKSVW13:fmcad}
A.~John, I.~Konnov, U.~Schmid, H.~Veith, and J.~Widder.
\newblock Parameterized model checking of fault-tolerant distributed algorithms
  by abstraction.
\newblock In \emph{FMCAD}, pages 201--209, 2013.

\bibitem[Killian et~al.(2007)Killian, Anderson, Braud, Jhala, and
  Vahdat]{KillianABJV07}
C.~E. Killian, J.~W. Anderson, R.~Braud, R.~Jhala, and A.~Vahdat.
\newblock Mace: language support for building distributed systems.
\newblock In \emph{{ACM} {SIGPLAN} PLDI}, pages 179--188, 2007.

\bibitem[Konnov et~al.(2015)Konnov, Veith, and Widder]{KVW15:CAV}
I.~Konnov, H.~Veith, and J.~Widder.
\newblock {SMT} and {POR} beat counter abstraction: Parameterized model
  checking of threshold-based distributed algorithms.
\newblock In \emph{CAV (Part~I)}, volume 9206 of \emph{LNCS}, pages 85--102,
  2015.

\bibitem[Konnov et~al.(2016{\natexlab{a}})Konnov, Lazi{\'c}, Veith, and
  Widder]{KLVW16:arxiv}
I.~Konnov, M.~Lazi{\'c}, H.~Veith, and J.~Widder.
\newblock A short counterexample property for safety and liveness verification
  of fault-tolerant distributed algorithms.
\newblock \emph{CoRR}, abs/1608.05327, 2016{\natexlab{a}}.
\newblock URL \url{http://arxiv.org/abs/1608.05327}.

\bibitem[Konnov et~al.(2016{\natexlab{b}})Konnov, Veith, and
  Widder]{KVW16:IandC}
I.~Konnov, H.~Veith, and J.~Widder.
\newblock On the completeness of bounded model checking for threshold-based
  distributed algorithms: Reachability.
\newblock \emph{Information and Computation}, 2016{\natexlab{b}}.
\newblock Accepted manuscript available online: 10-MAR-2016.
  \url{http://dx.doi.org/10.1016/j.ic.2016.03.006}.

\bibitem[Konnov et~al.(2016{\natexlab{c}})Konnov, Veith, and Widder]{KVW16:psi}
I.~Konnov, H.~Veith, and J.~Widder.
\newblock What you always wanted to know about model checking of fault-tolerant
  distributed algorithms.
\newblock In \emph{{PSI} 2015, Revised Selected Papers}, volume 9609 of
  \emph{LNCS}, pages 6--21. Springer, 2016{\natexlab{c}}.

\bibitem[Kroening et~al.(2011)Kroening, Ouaknine, Strichman, Wahl, and
  Worrell]{KroeningOSWW11}
D.~Kroening, J.~Ouaknine, O.~Strichman, T.~Wahl, and J.~Worrell.
\newblock Linear completeness thresholds for bounded model checking.
\newblock In \emph{CAV}, volume 6806 of \emph{LNCS}, pages 557--572, 2011.

\bibitem[Lamport and Schneider(1989)]{Lamport89pretendingatomicity}
L.~Lamport and F.~B. Schneider.
\newblock Pretending atomicity.
\newblock Technical Report~44, SRC, 1989.

\bibitem[Lesani et~al.(2016)Lesani, Bell, and Chlipala]{LesaniBC16}
M.~Lesani, C.~J. Bell, and A.~Chlipala.
\newblock Chapar: certified causally consistent distributed key-value stores.
\newblock In \emph{POPL}, pages 357--370, 2016.

\bibitem[Lincoln and Rushby(1993)]{LR93}
P.~Lincoln and J.~Rushby.
\newblock A formally verified algorithm for interactive consistency under a
  hybrid fault model.
\newblock In \emph{FTCS}, pages 402--411, 1993.

\bibitem[Lipton(1975)]{Lipton75}
R.~J. Lipton.
\newblock Reduction: A method of proving properties of parallel programs.
\newblock \emph{Commun. ACM}, 18\penalty0 (12):\penalty0 717--721, 1975.

\bibitem[Lubachevsky(1984)]{Lub84}
B.~D. Lubachevsky.
\newblock An approach to automating the verification of compact parallel
  coordination programs. {I}.
\newblock \emph{Acta Informatica}, 21\penalty0 (2):\penalty0 125--169, 1984.

\bibitem[Most{\'e}faoui et~al.(2003)Most{\'e}faoui, Mourgaya, Parv{\'e}dy, and
  Raynal]{MostefaouiMPR03}
A.~Most{\'e}faoui, E.~Mourgaya, P.~R. Parv{\'e}dy, and M.~Raynal.
\newblock Evaluating the condition-based approach to solve consensus.
\newblock In \emph{DSN}, pages 541--550, 2003.

\bibitem[Netflix(2010)]{Netflix5}
Netflix.
\newblock 5 lessons we have learned using {AWS}.
\newblock 2010.
\newblock retrieved on Nov.\ 7, 2016.
  \url{http://techblog.netflix.com/2010/12/5-lessons-weve-learned-using-aws.html}.

\bibitem[Ongaro and Ousterhout(2014)]{Ongaro2014}
D.~Ongaro and J.~Ousterhout.
\newblock In search of an understandable consensus algorithm.
\newblock In \emph{USENIX ATC}, pages 305--320, 2014.

\bibitem[Padon et~al.(2016)Padon, McMillan, Panda, Sagiv, and
  Shoham]{PadonMPSS16}
O.~Padon, K.~L. McMillan, A.~Panda, M.~Sagiv, and S.~Shoham.
\newblock Ivy: safety verification by interactive generalization.
\newblock In \emph{PLDI}, pages 614--630, 2016.

\bibitem[Pease et~al.(1980)Pease, Shostak, and Lamport]{LSP80}
M.~Pease, R.~Shostak, and L.~Lamport.
\newblock Reaching agreement in the presence of faults.
\newblock \emph{J. ACM}, 27\penalty0 (2):\penalty0 228--234, 1980.

\bibitem[Peluso et~al.(2016)Peluso, Turcu, Palmieri, Losa, and
  Ravindran]{Pel16}
S.~Peluso, A.~Turcu, R.~Palmieri, G.~Losa, and B.~Ravindran.
\newblock Making fast consensus generally faster.
\newblock In \emph{DSN}, pages 156--167, 2016.

\bibitem[Pnueli and Shahar(2000)]{Pnueli2000}
A.~Pnueli and E.~Shahar.
\newblock Liveness and acceleration in parameterized verification.
\newblock In \emph{CAV}, LNCS, pages 328--343, 2000.

\bibitem[Pnueli et~al.(2002)Pnueli, Xu, and Zuck]{PXZ02}
A.~Pnueli, J.~Xu, and L.~Zuck.
\newblock Liveness with (0,1,{$\infty$})- counter abstraction.
\newblock In \emph{CAV}, volume 2404 of \emph{LNCS}, pages 93--111. 2002.

\bibitem[Rahli et~al.(2015)Rahli, Guaspari, Bickford, and
  Constable]{RahliGBC15}
V.~Rahli, D.~Guaspari, M.~Bickford, and R.~L. Constable.
\newblock Formal specification, verification, and implementation of
  fault-tolerant systems using {EventML}.
\newblock \emph{{ECEASST}}, 72, 2015.

\bibitem[Raynal(1997)]{Raynal97}
M.~Raynal.
\newblock A case study of agreement problems in distributed systems:
  Non-blocking atomic commitment.
\newblock In \emph{HASE}, pages 209--214, 1997.

\bibitem[Schuppan and Biere(2006)]{SCHUPPAN200679}
V.~Schuppan and A.~Biere.
\newblock Liveness checking as safety checking for infinite state spaces.
\newblock \emph{Electronic Notes in Theoretical Computer Science}, 149\penalty0
  (1):\penalty0 79--96, 2006.

\bibitem[Song and van Renesse(2008)]{SongR08}
Y.~J. Song and R.~van Renesse.
\newblock Bosco: One-step {Byzantine} asynchronous consensus.
\newblock In \emph{DISC}, volume 5218 of \emph{LNCS}, pages 438--450, 2008.

\bibitem[Srikanth and Toueg(1987)]{ST87:abc}
T.~Srikanth and S.~Toueg.
\newblock Simulating authenticated broadcasts to derive simple fault-tolerant
  algorithms.
\newblock \emph{Dist. Comp.}, 2:\penalty0 80--94, 1987.

\bibitem[TLA()]{TLA}
TLA.
\newblock {TLA+} toolbox.
\newblock
  \url{http://research.microsoft.com/en-us/um/people/lamport/tla/tools.html}.

\bibitem[Vardi and Wolper(1986)]{VW86}
M.~Y. Vardi and P.~Wolper.
\newblock An automata-theoretic approach to automatic program verification.
\newblock In \emph{LICS}, pages 322--331, 1986.

\bibitem[von Gleissenthall et~al.(2016)von Gleissenthall, Bj{\o}rner, and
  Rybalchenko]{GleissenthallBR16}
K.~von Gleissenthall, N.~Bj{\o}rner, and A.~Rybalchenko.
\newblock Cardinalities and universal quantifiers for verifying parameterized
  systems.
\newblock In \emph{PLDI}, pages 599--613, 2016.

\bibitem[Wilcox et~al.(2015)Wilcox, Woos, Panchekha, Tatlock, Wang, Ernst, and
  Anderson]{WilcoxWPTWEA15}
J.~R. Wilcox, D.~Woos, P.~Panchekha, Z.~Tatlock, X.~Wang, M.~D. Ernst, and
  T.~E. Anderson.
\newblock Verdi: a framework for implementing and formally verifying
  distributed systems.
\newblock In \emph{PLDI}, pages 357--368, 2015.

\end{thebibliography}

\ifproofs

\makeatletter{}\clearpage

\appendix

\section*{APPENDIX}

\makeatletter{}\section{Specifications of fault-tolerant distributed algorithms}
    \label{sec:specs}

In this section, we summarize the specifications of fault-tolerant distributed
    algorithms, which we used to conduct the experiments. 
As our method receives a negation of the original specification, we give only
    the negated formulas in~$\ELTLTB$.

\makeatletter{}\begin{table*}[t!]

\begin{tabular}{lc|rr|rrrr|rrrr|rrrr}
        \tbh{Input}
        & \tbh{Case}
        & \multicolumn{2}{c|}{\tbh{Lasso length}}
        & \multicolumn{4}{c|}{\tbh{Nr. of lassos}}
        & \multicolumn{4}{c|}{\tbh{Time, seconds}}
        & \multicolumn{4}{c}{\tbh{Memory, GB}} \\
        \tbh{FTDA} & \scriptsize{(if more than one)}
        & \tbh{$\tbh{avg}$}
        & \tbh{$\tbh{max}$}
        & \tbh{$\tbh{S1}$}
        & \tbh{$\tbh{S2}$}
        & \tbh{$\tbh{L1}$}
        & \tbh{$\tbh{L2}$}
                & \tbh{$\tbh{S1}$}
        & \tbh{$\tbh{S2}$}
        & \tbh{$\tbh{L1}$}
        & \tbh{$\tbh{L2}$}
        & \tbh{$\tbh{S1}$}
        & \tbh{$\tbh{S2}$}
        & \tbh{$\tbh{L1}$}
        & \tbh{$\tbh{L2}$}

    \begin{filecontents*}{good-table.csv}
COMMON:01:algo,COMMON:998:rescond,POST:04:sys,S1:nschemas,S2:nschemas,L1:nschemas,L2:nschemas,S1:total-sec,S2:total-sec,L1:total-sec,L2:total-sec,S1:post-maxreskb,S2:post-maxreskb,L1:post-maxreskb,L2:post-maxreskb,avglen,maxlen
FRB,---,bcast-folklore-crash,5,--,5,5,1,--,1,1,0.1,--,0.1,0.1,17,29
STRB,---,bcast-byz,14,--,4,18,1,--,1,1,0.1,--,0.1,0.1,72,128
NBACC,---,asyn-ray97-nbac-clean,2,1,1,3,19,18,25,14,0.1,0.1,0.3,0.1,4608,4632
NBACG,---,asyn-guer01-nbac,90,90,19,--,12,11,2,--,0.1,0.1,0.1,--,148,221
NBAC,---,asyn-ray97-nbac,448,69,69,517,771,110,253,431,0.5,0.1,0.2,0.5,2489,7109
BOSCO,$\lfloor \frac{n+3t}{2} \rfloor + 1= n-t$,bosco,20,20,\highlight{69},--,3,6,\highlight{1},--,0.1,0.1,\highlight{0.1},--,423,506
BOSCO,$\lfloor \frac{n+3t}{2} \rfloor + 1> n-t$,bosco,70,70,\highlight{251},--,19,17,\highlight{3},--,0.1,0.1,\highlight{0.1},--,368,1112
BOSCO,$\lfloor \frac{n+3t}{2} \rfloor + 1< n-t$,bosco,20,20,\highlight{69},--,4,4,\highlight{2},--,0.1,0.1,\highlight{0.1},--,286,562
BOSCO,$n > 5t \wedge f = 0$,bosco,924,--,3431,--,61,--,272,--,0.2,--,0.3,--,30,8982
BOSCO,$n > 7t$,bosco,924,--,3431,--,167,--,1579,--,0.3,--,0.6,--,180,11294
ABA,$\frac{n+t}{2} = 2t+1$,asyn-byzagreement0,448,--,138,1172,79,--,42,738,0.1,--,0.1,0.4,850,1435
ABA,$\frac{n+t}{2} > 2t+1$,asyn-byzagreement0,2100,--,502,5204,1536,--,496,15720,1.0,--,0.2,5.1,2112,3548
C1CS,$f = 0$,c1cs,70,--,251,--,59,--,596,--,0.2,--,1.0,--,463,4352
C1CS,$f = 1$,c1cs,140,--,448,--,54,--,846,--,0.2,--,0.6,--,1276,7143
C1CS,$f > 1$,c1cs,630,--,2100,--,393,--,3507,--,0.6,--,1.3,--,436,4492
CF1S,$f = 0$,consensus-folklore-onestep,90,--,19,--,35,--,26,--,0.1,--,0.1,--,1191,2114
CF1S,$f = 1$,consensus-folklore-onestep,523,--,\highlight{90},--,313,--,\highlight{5},--,0.2,--,\highlight{0.1},--,787,1757
CF1S,$f > 1$,consensus-folklore-onestep,3429,--,\highlight{448},--,8125,--,\highlight{11},--,2.8,--,\highlight{0.1},--,2132,4993
CBC,$\lfloor \frac{n}{2} \rfloor < n-t \wedge f=0$,cond-consensus2,2,2,2,--,145,146,204,--,0.2,0.2,0.7,--,8168,8168
CBC,$\lfloor \frac{n}{2} \rfloor = n-t \wedge f=0$,cond-consensus2,2,2,2,--,9,10,13,--,0.1,0.1,0.1,--,1790,1790
CBC,$\lfloor \frac{n}{2} \rfloor < n-t \wedge f > 0$,cond-consensus2,5,5,4,--,6072,508,11799,--,16.7,2.0,7.3,--,10213,12236
CBC,$\lfloor \frac{n}{2} \rfloor = n-t \wedge f > 0$,cond-consensus2,5,5,4,--,240,387,81,--,1.1,2.3,0.5,--,2258,2708
\end{filecontents*}
\csvreader[before line=\\, before first line=\\\hline,        late after last line=\\\hline]    {good-table.csv}{1=\algo, 2=\rc, 4=\sonen, 5=\stwon, 6=\lonen, 7=\ltwon,
        8=\sonet, 9=\stwot, 10=\lonet, 11=\ltwot,
        12=\sonem, 13=\stwom, 14=\lonem, 15=\ltwom, 16=\savg, 17=\smax}    {\scriptsize{\textsf{\textbf{{\algo}}}}
            & \scriptsize{\rc}
            & \scriptsize{\savg} & \scriptsize{\smax}
            & \scriptsize{\sonen} & \scriptsize{\stwon}
            & \scriptsize{\lonen} & \scriptsize{\ltwon}
            & \scriptsize{\sonet} & \scriptsize{\stwot}
            & \scriptsize{\lonet} & \scriptsize{\ltwot}
            & \scriptsize{\sonem} & \scriptsize{\stwom}
            & \scriptsize{\lonem} & \scriptsize{\ltwom}
    }
   \end{tabular}
  \caption{ Summary of our experiments on Intel\textregistered\ 
      Xeon\textregistered\ E5345, 4 cores, 48~GB.
      We apply the optimizations introduced
      in~\cite[Sec.~4.4.]{KVW15:CAV}.
      A \highlight{gray box} highlights the benchmarks, for which the tool
      reports a counterexample.
    }
  \label{tab:experiments}
\end{table*}

\paragraph{Consistent broadcast (STRB).} The negated safety specification~(\textsf{S1})
    is as follows:

\begin{enumerate}
  \item[S1:] $\tlE (( \bigvee_{\ell:\ \ell.\pc = \RI} \counters[\ell]
      \ne 0 ) \wedge \ltlF \bigvee_{\ell:\ \ell.\pc = \textrm{AC}} \kappa_\ell \ne 0)$.

\end{enumerate}

All our benchmarks have similar fairness constraints, that is the property of
    reliable communication that requires the processes to eventually receive
    the messages from all other correct processes. 
The fairness constraint that encodes reliable communication for STRB is as follows:
$\varphi_{rc} \equiv 
  (\sent < t+1 \vee \bigwedge_{\ell:\ \ell \in \alpha(\rcvd < t + 1)} \counters[\ell] = 0)
    \wedge
  (\sent < n-t \vee \bigwedge_{\ell:\ \ell.\rcvd \in \alpha(x < n-t)} \counters[\ell] = 0)
$, where $\alpha(x < n-t)$ is the set of abstract values produced by the parametric
    interval data abstraction~\cite{JohnKSVW13:fmcad}.

Using~$\varphi_{rc}$, we write the liveness properties~\textsf{L1}
    and~\textsf{L2} as:

\begin{enumerate}
\item[L1:] $\ltlE \big(
    \ltlG \ltlF \varphi_{rc}
    \wedge \bigwedge_{\ell:\ \ell.\pc = \IT} \kappa_\ell = 0
    \wedge \ltlG \bigwedge_{\ell:\ \ell.\pc = \textrm{AC}} \kappa_\ell = 0 \big)$.

\item[L2:] $\ltlE \big(\ltlG \ltlF \varphi_{rc}
    \wedge \ltlF \big( \bigvee_{\ell:\ \ell.\pc = \textrm{AC}} \kappa_\ell \ne 0
    \wedge \ltlG (\bigvee_{\ell:\ \ell.\pc \ne \textrm{AC}} \kappa_\ell \ne 0)
    \big)$.

\end{enumerate}

\paragraph{Folklore Reliable Broadcast (FRB).}
    FRB has exactly the same specifications~\textsf{S1}, \textsf{L1}, and~\textsf{L2}
    as STRB, but fewer local states. 

\paragraph{Asynchronous Byzantine agreement (ABA).}
    ABA has exactly the same specifications~\textsf{S1}, \textsf{L1}, and~\textsf{L2}
    as STRB, but more local states and guards. 
In addition to~$\varphi_{\mathit{rc}}$, ABA has four fairness constraints
    that enforce local progress of enabled process transitions,
    e.g., $\ltlG \ltlF (\sent < 2t + 1
    \vee \bigwedge_{\ell:\ \ell.\pc = \textrm{RD}} \counters[\ell] = 0)$.

\paragraph{Condition-based consensus (CBC).}
CBC has two unique initial local states, where the processes
    are initialized with values 0 and 1 respectively.
    In our encoding, the numbers of processes in these states are counted
    with~$\kappa_0$ and~$\kappa_1$.
    The negation of \emph{termination} is defined as follows:
\begin{enumerate}
   \item[L1: ] $\ltlE \big(
    \ltlG \ltlF \varphi_{rc}
    \wedge |\kappa_0 - \kappa_1| > t
    \wedge \ltlG \bigvee_{\ell:\ \ell.\pc \not\in \{\textrm{AC}, \textrm{CR}\}} \kappa_\ell \ne 0 \big)$
\end{enumerate}

The negations of \emph{validity} and \emph{agreement} are as follows:

\begin{enumerate}
  \item[S1:] $\tlE (( \bigvee_{\ell:\ \ell.\pc = \RI} \counters[\ell]
      \ne 0 ) \wedge \ltlF \bigvee_{\ell:\ \ell.\pc = \textrm{AC0}} \kappa_\ell \ne 0)$

  \item[S2:] $\tlE (|\kappa_0 - \kappa_1| > t \wedge
      \ltlF (\bigvee_{\ell:\ \ell.\pc \in \{\textrm{AC0}, \textrm{AC1}\}} \kappa_\ell \ne 0))$
 
\end{enumerate}

\paragraph{Non-blocking atomic commit.}

In addition to the fairness constraint~$\varphi_{rc}$, NBAC, NBACC, NBACG use
    a fairness constraint $\varphi_{\mathit{fd}}$ on a failure detector
    defined as:
    $\bigwedge_{\ell:\ \ell.\mathit{some\_fail} = \mathit{true}}
    \kappa_\ell = 0 \wedge \ltlG (\bigwedge_{\ell:\ \ell.\pc = \textrm{CR}} \kappa_\ell = 0)$.

    The negation of \emph{termination} is defined as follows:
\begin{enumerate}
   \item[L1: ] $\ltlE (
    \ltlG \ltlF \varphi_{rc} \wedge \varphi_{fd}
    \wedge \ltlG \bigvee_{\ell:\ \ell.\pc \not\in \{\textrm{COMMIT}, \textrm{ABORT}, \textrm{CR}\}} \kappa_\ell \ne 0 )$
\end{enumerate}

The negations of \emph{abort-validity} and \emph{agreement} are as follows:

\begin{enumerate}
    \item[S1:] $\tlE (( \bigvee_{\ell:\ \ell.\pc = \textrm{NO}} \counters[\ell]
      \ne 0 ) \wedge \ltlF \bigvee_{\ell:\ \ell.\pc = \textrm{COMMIT}} \kappa_\ell \ne 0)$.

  \item[S2:] $\tlE ( \ltlF (\bigvee_{\ell:\ \ell.\pc = \textrm{ABORT}} \kappa_\ell \ne 0
      \wedge \bigvee_{\ell:\ \ell.\pc = \textrm{COMMIT}} \kappa_\ell \ne 0))$.

\end{enumerate}

\paragraph{CFCS and C1CS.}

The negation of \emph{fast termination} for value 0 is:
\begin{enumerate}
   \item[L1: ] $\ltlE (
    \ltlG \ltlF \varphi_{rc}
    \wedge \bigwedge_{\ell:\ \ell.\pc \ne \textrm{V0}} \kappa_\ell = 0
    \wedge \ltlG \bigvee_{\ell:\ \ell.\pc \not\in \{\textrm{D0}, \textrm{CR}\}} \kappa_\ell \ne 0 )$
\end{enumerate}

The negation of \emph{one-step} for value 0 is:

\begin{enumerate}
   \item[S1: ] $\ltlE (
       \bigwedge_{\ell:\ \ell.\pc \ne \textrm{V0}} \kappa_\ell = 0
    \wedge \ltlF \bigvee_{\ell:\ \ell.\pc \in \{\textrm{D1}, \textrm{U0}, \textrm{U1}\}} \kappa_\ell \ne 0 )$
\end{enumerate}

\paragraph{BOSCO~\cite{SongR08}.}

The negation of \emph{fast termination} for value 0 is:
\begin{enumerate}
   \item[L1: ] $\ltlE (
    \ltlG \ltlF \varphi_{rc}
    \wedge \bigwedge_{\ell:\ \ell.\pc \ne \textrm{V0}} \kappa_\ell = 0
    \wedge \ltlG \bigvee_{\ell:\ \ell.\pc \not\in \{\textrm{D0}, \textrm{CR}\}} \kappa_\ell \ne 0 )$
\end{enumerate}

The negations of \emph{Lemma~3} and \emph{Lemma~4}
    of~\cite{SongR08} are:

\begin{enumerate}
   \item[S1: ] $\ltlE (
       \ltlF (\bigvee_{\ell:\ \ell.\pc = \textrm{D0}} \kappa_\ell \ne 0
       \wedge \bigvee_{\ell:\ \ell.\pc = \textrm{D1}} \kappa_\ell \ne 0 ))$

   \item[S2: ] $\ltlE (
       \ltlF (\bigvee_{\ell:\ \ell.\pc = \textrm{D0}} \kappa_\ell \ne 0
       \wedge \bigvee_{\ell:\ \ell.\pc = \textrm{U1}} \kappa_\ell \ne 0 ))$
\end{enumerate}

\makeatletter{}\section{Detailed Proofs for Section~\ref{sec:counterexamples}}

\recallproposition{prop:lasso-sched}{\proplassoscheda}

\makeatletter{}\begin{proof}
    We do not give details on B\"uchi automata and the construction by
    Vardi and Wolper, since this construction is well-known and can be found
    in the original paper~\cite{VW86} as well as in a number of textbooks,
    e.g.,~\cite{CGP1999}[Ch. 9] and \cite{BK08}[Ch.~5].

Using the construction from~\cite{VW86}, we translate the formula~$\varphi$
    into a B\"uchi automaton $B=(\emph{AP}, Q, \Delta, Q^0, F)$, which has a
    finite set~$Q$ of states, a finite set~$Q_0 \subseteq Q$ of initial states,
    a finite set~$F$ of accepting states, a finite alphabet~$\emph{AP}$ of
    atomic propositions (which corresponds to the propositional formulas
    derived from~$\mathit{pform}$), and the transition relation~$\Delta \subseteq Q
    \times \emph{AP} \times Q$. 
The key property is that the automaton~$B$ recognizes exactly those sequences
    of propositions that satisfy the formula~$\varphi$.

Let $(\configs, \iconfigs, \transrel)$ be the counter system~$\Sys(\TA)$ as
    defined in Section~\ref{sec:countsys}. 
The system $\Sys(\TA)$ is a transition system, so following~\cite{VW86} we can
    construct the product B\"uchi automaton $\Sys(\TA) \otimes B$ that
    corresponds to the synchronous product of $\Sys(\TA)$ and $B$. 
Formally, $\Sys(\TA) \otimes B$ is the B\"uchi automaton $(\emph{AP}, Q_P,
    \Delta_P, Q_P^0, F_P)$ defined as follows:

\begin{itemize}
\item The set of states $Q_P$ is the Cartesian product $(\configs \cup
        \{\iota\}) \times Q$, where $\iota \not \in \configs$ is a dummy
        configuration, which is used to delay initialization of the counter
        system by one step.

\item The set of initial states $Q^0_P$ is the Cartesian product $\{\iota\}
        \times Q^0$.

\item The set of accepting states $F_P$ is the Cartesian product $\configs
        \times F$.

\item The transition relation~$\Delta_P$ includes the following triples:
        \begin{itemize}
\item an initial transition $((\iota, q_0), p, (\gst, q))$ for $q_0 \in Q_0$,
        $\gst \in \iconfigs$, and $q \in Q$ such that $(q_0, p, q) \in \Delta$
        and $\gst \models p$.

\item a transition $((\gst, q), p, (\gst', q'))$ for $q, q' \in Q$ and
        $\gst, \gst' \in \configs$ such that $(q, p, q') \in \Delta$
        and $\gst' \models p$.
    \end{itemize}

\end{itemize}

A run of the product automaton is an infinite sequence $(\iota, q_0),$
    $(\gst_1, q_1), \dots, (\gst_i, q_i), \dots$ such that $((\iota, q_0), p_0,
    (\gst_1, q_1)) \in \Delta$ and $((\gst_i, q_i), p_i, (\gst_{i+1}, q_{i+1}))
    \in \Delta$ for $i \ge 1$ and some propositions~$p_0, p_1, \dots \in
    \mathit{AP}$.
The run is accepting, if there is a state $(\gst_j, q_j) \in F_P$ that appears
    infinitely often in the run. 

In contrast to~\cite{VW86}, the product automaton $\Sys(\TA) \otimes
     B$ has infinitely many states.
However, by Proposition~\ref{prop:finite}, every path of~$\Sys(\TA)$
     visits only finitely many states, and thus every run of the
     product automaton visits finitely many states too.
Hence, in each run there are finitely many accepting states.
Due to this, and since, by assumption, $\Sys(\TA) \models \ltlE
     \varphi$, the product has an accepting run $(\iota, q_0),(\gst_1,
     q_1), \dots, (\gst_i, q_i), \dots$ with state $(\gst_j, q_j) \in
     F_P$ appearing infinitely often for some $j \ge 1$.
Hence, there is an index $k \ge 0$ such that $(\gst_{j+k+1},
     q_{j+k+1}) = (\gst_j, q_j)$.
Consequently, we construct a lasso run by taking the sequence of states
     $(\iota, q_0),(\gst_1, q_1),\dots,(\gst_{j-1}, q_{j-1})$ as a
     prefix and the sequence $(\gst_j, q_j), \dots,(\gst_{j+k},
     q_{j+k})$ as a loop, which is repeated infinitely.
This lasso run is also an accepting run of the product automaton.

It is immediate from the construction, that the infinite sequence
    $\gst_1,\dots,\gst_{j-1}, (\gst_j,\dots,\gst_{j+k})^\omega$
corresponds to a path
    of~$\Sys(\TA)$ starting from an initial configuration~$\gst_1 \in
    \iconfigs$, and this path satisfies the formula~$\varphi$. 
Thus, there are schedules $\tau = t_1, \dots, t_{j-1}$ and $\rho = t_j, \dots,
    t_{j+k}$ such that:
    \begin{enumerate}
        \item Schedule $\tau$ is applicable to~$\gst_1$ and
            the prefixes of~$\tau$ visit the intermediate configurations:
                        \begin{equation*}
                (t_1, \dots, t_i)(\gst_1) = \gst_i \mbox{ for } 1 \le i < j,
            \end{equation*}
    
        \item Schedule $\rho$ is applicable to~$\gst_j$,
            the prefixes of~$\rho$ visit the intermediate configurations,
            and~$\rho$ closes the loop:
                        \begin{equation*}
            (t_j, \dots, t_m)(\gst_j) =
                \begin{cases}
                    \gst_m, & \mbox{ when } j \le m < j+k\\
                    \gst_j, & \mbox{ when } m = j+k.
                \end{cases}
            \end{equation*}

    \end{enumerate}

The infinite schedule~$\tau \cdot \rho^\omega$ is the required schedule.
Indeed, $\infpath{\gst_1}{\tau \cdot \rho^\omega} \models \varphi$ and
    $\rho^i(\tau(\gst_1)) = \tau(\gst_1)$ for $i \ge 0$.
\end{proof}

\recallproposition{prop:unwinding}{\propunwinding}

\makeatletter{}
\begin{proof}

For each node $\left<\ltlF \psi, w\right> \in \syntree(\varphi)$, we
    define an extreme appearance~$\EA(w)$ as follows:

\begin{enumerate}
    \item\label{itm:EAomega} If there is an index
        $k \in \{|\tau|, \dots, |\tau| + |\rho| - 1\}$ such
        that $k$ cuts~$\tau\cdot\rho$ in~$\tau\cdot\rho'$ and $\rho''$, and
        it holds that
        $\infpath{(\tau \concat \rho')(\gst)}{\rho'' \concat \rho^\omega}
        \models \psi$, then we set $\EA(w)$ to the maximal such~$k
        \ge |\tau|$.

\item \label{itm:EA} Otherwise, we set $\EA(w)$ to the maximal $k < |\tau|$
        such that $k$ cuts~$\tau$ in~$\tau', \tau''$ and
        $\infpath{\tau'(\gst)}{\tau'' \concat \rho^\omega} \models \psi$.
(Such~$k$ exists, as the case~\ref{itm:EAomega} does not apply, it holds
    $\infpath{\gst}{\tau \concat \rho^\omega} \models \ltlF \psi$, and since
    temporal formulas are connected only with the conjunction~$\wedge$.)
\end{enumerate}

Consider a topologically ordered sequence $v_1, v_2, \dots, v_{|\pgvertices|}$
    of the vertices of the cut graph~$\pgraph(\syntree(\varphi)) =
    (\pgvertices, \pgedges)$, that is, the condition $(v_i, v_j) \in \pgedges$
    implies $i < j$ for $1 \le i,j \le |\pgvertices|$. 
Such a sequence exists, since the graph $\pgraph(\syntree(\varphi))$ is a
    directed acyclic graph. 
Let $\ell \in \{1,\dots,|\pgvertices|\}$ be the index of the node $\loops$,
    i.e., $v_\ell = \loops$.

We unroll the loop~$K = \ell - 1$ times. 
Formally, for $1 \le i \le |\pgvertices|$, we set the cut point~$\imap(v_i)$
     as follows:
\begin{equation*}
    \imap(v_i) =
        \begin{cases}
            |\tau| + |\rho| \cdot K,& \mbox{ if } v_i = \loops\\
            |\tau| + |\rho| \cdot (K+1) - 1,& \mbox{ if } v_i = \loope\\
            \EA(v_i),& \mbox{ if } \EA(v_i) < |\tau|\\
            \EA(v_i) + |\rho| \cdot (i-1),&
                \mbox{ if } i < \ell \mbox{ and } |\tau| \le \EA(v_i) \\
            \EA(v_i) + |\rho| \cdot K,& \mbox{ if } i \ge \ell
        \end{cases}
\end{equation*}    
It is easy to see that~$\imap$ satisfies Definition~\ref{def:cutfun}.
By the construction of extreme appearances, for a node~$\left<\ltlF
     \psi, w\right>$, the formula~$\psi$ is satisfied at the extreme
     appearance~$\EA(w)$.
Since $\imap(w) - \EA(w) = |\rho| \cdot i$ for some $i \ge 0$, it
     follows that if $\imap(w)$ cuts $(\tau \cdot \rho^K) \cdot
     \rho^\omega$ into $\pi'$ and $\pi''$, then
     $\infpath{\pi'(\gst)}{\pi'' \cdot \rho^\omega} \models \psi$
     holds.
\end{proof}

\begin{lemma}\label{lemma:prefix-Gs}
Let~$\gst$ be a configuration, $\tau \cdot \rho^\omega$ be a lasso schedule
    applicable to~$\gst$, and $\varphi$ be an $\ELTLTB$ formula. 
If an index $k < |\tau|$ cuts $\tau$ into $\pi'$ and $\pi''$ and
    $\setconf{\pi'(\gst)}{\pi'' \cdot \rho} \models \varphi$ holds, then
    $\infpath{\pi'(\gst)}{\pi'' \cdot \rho^\omega} \models \ltlG \varphi$ holds.
\end{lemma}

\begin{proof}
    From $\setconf{\pi'(\gst)}{\pi'' \cdot \rho} \models \varphi$,
    we immediately conclude that two subsets of
    $\setconf{\pi'(\gst)}{\pi'' \cdot \rho}$ also satisfy~$\varphi$:
        \begin{align}
        \setconf{\pi'(\gst)}{\pi''} & \models  \varphi \label{eq:prefix}\\
        \setconf{\tau(\gst)}{\rho} & \models \varphi \label{eq:loop}
    \end{align}

Since $\tau \cdot \rho^\omega$ is a lasso schedule, we have $\rho^i(\tau(\gst))
= \tau(\gst)$ for $i \ge 0$. 
From this and Equation~(\ref{eq:loop}), we conclude that
$\infpath{\tau(\gst)}{\rho^\omega} \models \ltlG \varphi$ holds. 
By combining this with Equation~(\ref{eq:prefix}), we arrive at the required
    property $\infpath{\pi'(\gst)}{\pi'' \cdot \rho^\omega} \models \ltlG
    \varphi$.
\end{proof}

\begin{lemma}\label{lemma:loop-Gs}
Let~$\gst$ be a configuration, $\tau \cdot \rho^\omega$ be a lasso schedule
    applicable to~$\gst$, and $\varphi$ be an $\ELTLTB$ formula. 
If an index $k: |\tau| \le k < |\tau| + |\rho|$ cuts $\tau \cdot \rho$
    into $\pi'$ and $\pi''$ and
    $\setconf{\tau(\gst)}{\rho} \models \varphi$ holds, then
    $\infpath{\pi'(\gst)}{\pi'' \cdot \rho^\omega} \models \ltlG \varphi$ holds.
\end{lemma}

\begin{proof}
Since $\tau \cdot \rho^\omega$ is a lasso schedule, we have $\rho^i(\tau(\gst))
    = \tau(\gst)$ for $i \ge 0$. 
Thus, $\setconf{\tau(\gst)}{\rho} \models \varphi$ implies
    $\infpath{\tau(\gst)}{\rho^\omega} \models \ltlG \varphi$. 
As $\infpath{\pi'(\gst)}{\pi'' \cdot \rho^\omega}$ is a subsequence of
    $\infpath{\tau(\gst)}{\rho^\omega}$, we arrive at
    $\infpath{\pi'(\gst)}{\pi'' \cdot \rho^\omega} \models \ltlG \varphi$.
\end{proof}

\recallthm{thm:witness-soundness}{\thmwitnesssoundness}

\makeatletter{}\begin{proof}
Let the cut graph $\pgraph(\syntree(\varphi))$ be $(\pgvertices, \pgedges)$.
We start with defining the notion of the parent cutpoint for a formula that has
    the form~$\ltlG \psi$. 
Given a tree node $\left<\ltlG \psi, u.j\right> \in \syntree(\varphi)$ with
    $\psi \ne \true$, we denote with~$\pnode(u.j)$ the parent node
    $\left<\psi', u\right> \in \syntree(\varphi)$. 
(By the definition of a canonical syntax tree, the formula $\ltlG \psi$
    alone cannot be the formula of the root node.)
Note that the id~$u$ always points to either the root node, or a node of the
    form $\left<\ltlF \psi'', u\right>$ for some formula~$\psi'' \in \ELTLTB$. 
We define the parent cutpoint as follows:
\begin{equation*}
    \pcutpoint(w) =
    \begin{cases}
\imap(u), \mbox{ when } u \in \pgvertices, w = u.j \mbox{ for some } j \in
    \NatZero,\\ 0, \mbox{ otherwise}.
    \end{cases}
\end{equation*}

We prove the following statements about the intermediate
 tree nodes using structural
    induction on the tree~$\syntree(\varphi)$:

\begin{enumerate}[label=(\roman*)]
\item\label{itm:G} for a node $\left<\ltlG \psi, w\right>$ with $\psi \ne
        \true$, if $\pcutpoint(w)$ cuts $\tau \cdot \rho$ into $\pi'$ and
        $\pi''$, then $\infpath{\pi'(\gst)}{\pi'' \cdot \rho^\omega} \models
        \ltlG \psi$.

\item\label{itm:F} for a node $\left<\ltlF \psi, w\right>$, if $\imap(w)$ cuts
        $\tau \cdot \rho$ into $\pi'$ and $\pi''$, then
        $\infpath{\pi'(\gst)}{\pi'' \cdot \rho^\omega} \models \psi$.
\end{enumerate}

Based on this we finally prove
\begin{enumerate}
\item[(iii)]\label{itm:root} for the root node $\left<\canform(\varphi), 0\right> \in
        \syntree(\varphi)$, it holds that $\infpath{\gst}{\tau \cdot
        \rho^\omega} \models \canform(\varphi)$.
\end{enumerate}
which establishes the theorem.

\paragraph{Proving~\ref{itm:G}.}
    Fix a tree node $\left<\ltlG \psi, w\right>$ with $\psi \ne
        \true$. Let $\pcutpoint(w)$ cut $\tau \cdot \rho$ into $\pi'$ and
        $\pi''$.
    We have to show that $\infpath{\pi'(\gst)}{\pi'' \cdot \rho^\omega} \models
        \ltlG \psi$.
Since $\psi \ne \true$, by the definition of a canonical formula, $\psi$ has
    the form $\psi_0 \wedge \ltlF \psi_1 \dots \ltlF \psi_k \wedge \ltlG \true$
    for some $k \ge 0$, a propositional formula~$\psi_0$ and canonical
    formulas~$\psi_1, \dots, \psi_k$. 
It is sufficient to show that: (a)~$\pi'(\gst) \models \ltlG \psi_0$, and
    (b)~$\infpath{\pi'(\gst)}{\pi'' \cdot \rho^\omega} \models \ltlG \ltlF
    \psi_i$ for $1 \le i \le k$.

To show (a), we consider three cases:
\begin{enumerate}
    \item Case: $\pnode(w)$ is the root and $\varphi$ is not of the form
        $\ltlF (\dots)$. 
Then $\canform(\varphi) = \dots \wedge \ltlG \psi$, and by
    Condition~\ref{assume:root} of Definition~\ref{def:prop-witness},
    we have $\setconf{\gst}{\tau
    \cdot \rho} \models \prop(\psi)$. As $\tau \cdot \rho^\omega$ is a lasso,
    i.e., $(\tau \cdot \rho^k(\gst)) = \tau(\gst)$ for $k \ge 0$,
    we have that $\gst \models \ltlG \prop(\psi)$.
    From this, and $\prop(\psi) = \psi_0$, we conclude
    that $\pi'(\gst) \models \ltlG \psi_0$, as $\pcutpoint(w) = 0$ and
    thus $\pi'$ is the empty schedule.

\item Case: $\pnode(w) = \left<\ltlF \psi'', u\right>$ for some $\psi'' \in
\ELTLTB$ and $u \in \NatZero^\omega$, and $\imap(u) < |\tau|$. 
In this case, $\psi'' = \dots \wedge \ltlG \psi$.
By Condition~\ref{assume:Gfin-prefix} of Definition~\ref{def:prop-witness},
$\setconf{\pi'(\gst)}{\pi''} \models \prop(\psi)$. 
By noticing $\prop(\psi) = \psi_0$ and applying Lemma~\ref{lemma:prefix-Gs}, we
    have $\infpath{\pi'(\gst)}{\pi'' \cdot \rho^\omega} \models \ltlG
    \psi_0$.

\item Case: $\pnode(w) = \left<\ltlF \psi'', u\right>$ for some $\psi'' \in
        \ELTLTB$ and $u \in \NatZero^\omega$, and $|\tau| \le \imap(u) < |\tau|
+ |\rho|$. 
In this case, $\psi'' = \dots \wedge \ltlG \psi$.
By Condition~\ref{assume:Gfin-loop} of Definition~\ref{def:prop-witness},
$\setconf{\tau(\gst)}{\rho} \models \prop(\psi)$. 
By noticing $\prop(\psi) = \psi_0$ and applying Lemma~\ref{lemma:loop-Gs}, we
    arrive at $\infpath{\pi'(\gst)}{\pi'' \cdot \rho^\omega} \models \ltlG
    \psi_0$.
\end{enumerate}

\bigskip

To show (b), we fix an index~$i: 1 \le i \le k$ and prove
    $\infpath{\pi'(\gst)}{\pi'' \cdot \rho^\omega} \models \ltlG \ltlF \psi_i$.
Let~$w_i$ be the node id of the $\psi$'s subformula $\ltlF \psi_i$ in
    the syntax tree~$\syntree(\varphi)$. 
Note that $\left<\ltlF \psi_i, w_i\right>$ is covered by a $\ltlG$-node, since
    it is created from a subformula of $\ltlG \psi$. 
Thus, $(\loops, w_i) \in \pgedges$, and by the definition of the cut
    function~$\imap$, we have~$\imap(w_i) \ge \imap(\loops) \ge |\tau|$. 
Let~$\imap(w_i)$ cut~$\tau \cdot \rho$ in~$\tau \cdot \beta'$ and~$\beta''$. 
By the inductive hypothesis, Point~\ref{itm:F} holds for the
    tree node~$w_i$, and thus $\finpath{(\tau \cdot \beta')(\gst)}{\beta'' \cdot \rho^\omega}
\models \psi_i$ holds.
Since $\tau \cdot \rho^\omega$ is a lasso-shaped schedule, we have $\tau(\gst)
    = (\tau \cdot \rho^j)(\gst)$ for $j \ge 0$, that is, the state~$\tau(\gst)$
    occurs infinitely often in the path $\infpath{(\tau \cdot
    \beta')(\gst)}{\beta'' \cdot \rho^\omega}$.
Hence, we arrive at:
\begin{equation*}
\infpath{(\tau \cdot \beta')(\gst)}{\beta'' \cdot \rho^\omega}
\models \ltlG \ltlF \psi_i, \mbox{ for } 1 \le i \le k.
\end{equation*}

From~(a) and~(b), and the standard \LTL{} property $(\ltlG A) \wedge (\ltlG B)
    \Rightarrow \ltlG(A \wedge B)$, Point~\ref{itm:G} follows for the tree
    node~$\left<\ltlG \psi, w\right>$.

\paragraph{Proving~\ref{itm:F}.} Fix a tree node $\left<\ltlF \psi, w\right>$, and
    let $\imap(w)$ cut $\tau \cdot \rho$ into $\pi'$ and $\pi''$. 
We have to show that $\infpath{\pi'(\gst)}{\pi'' \cdot \rho^\omega} \models
    \psi$ holds.

By the definition of a canonical formula,~$\psi$ has the form $\psi_0 \wedge
    \ltlF \psi_1 \wedge \dots \wedge \ltlF \psi_k \wedge \ltlG \psi_{k+1}$ for
    some $k \ge 0$, a propositional formula~$\psi_0$, canonical
    formulas~$\psi_1, \dots, \psi_k$, and a formula~$\psi_{k+1}$
    that is either a canonical formula, or equals~$\true$.
    We will show that: (a)~$\pi'(\gst) \models \psi_0$, and
    (b)~$\infpath{\pi'(\gst)}{\pi'' \cdot \rho^\omega} \models
    \ltlG \psi_{k+1}$, and
    (c)~$\infpath{\pi'(\gst)}{\pi'' \cdot \rho^\omega} \models
    \ltlF \psi_i$ for $1 \le i \le k$.
From (a)--(c), the required statement immediately follows.

To show~(a), we notice that there are two cases: $\imap(w) < |\tau|$, or
    $\imap(w) \ge |\tau|$. 
In these cases, either Assumption~(\ref{assume:Gfin-prefix}) or
    Assumption~(\ref{assume:Gfin-loop}) implies that $\pi'(\gst) \models
    \psi_0$.

To show~(b), we focus on the case $\psi_{k+1} \ne \true$, as the case
    $\psi_{k+1} = \true$ is trivial. 
Notice that by the definition of the syntax tree~$\syntree(\varphi)$, the
    subformula $\ltlG \psi_{k+1}$ has the id $w.j$ for $j = k+1$, and thus
    $\pcutpoint(w.j) = \imap(w)$. 
Thus,~$(b)$ follows directly from the inductive hypothesis~\ref{itm:G}, which
    has already been shown to hold
    for the tree node $\left<\ltlG \psi_{k+1}, w.j\right>$.

To show~(c), fix an index $i \in \{1, \dots, k\}$. 
Let $\imap(w.i)$ cut $\tau \cdot \rho$ into $\beta'$ and $\beta''$. 
The inductive hypothesis~\ref{itm:F} has been shown to hold for the tree node
    $\left<\ltlF \psi_i, w.i\right> \in \syntree(\varphi)$, and thus we have:
\begin{equation}\label{eq:psi-at-beta}
    \infpath{\beta'(\gst)}{\beta'' \cdot \rho^\omega} \models \psi_i.
\end{equation}    

We consider three cases, based on whether~$w$ and~$w.i$ are covered by
    a $\ltlG$-node:
\begin{enumerate}

\item Case: neither $w$, nor $w.i$ is covered by a $\ltlG$-node. 
By the definition of the cut graph, $(w, w.i) \in \pgedges$, and thus
    by the definition of the cut function~$\imap$, it holds that
    $\imap(w) \le \imap(w.i)$. 
    From this and Equation~(\ref{eq:psi-at-beta}), it follows that
    $\infpath{\pi'(\gst)}{\pi'' \cdot \rho^\omega} \models \ltlF \psi_i$.

\item Case: $w.i$ is covered by a $\ltlG$-node. 
By the definition of the cut graph, the node~$w.i$ has to be inside the
loop: $(\loops, w.i) \in \pgedges$. 
Consequently, by the definition of the cut function~$\imap$, it holds that
    $\imap(w.i) \ge |\tau|$.
Let $\beta'_l$ be the suffix of~$\beta'$ inside the loop, i.e.,
    $\beta' = \tau \cdot \beta'_l$.
As $\tau \cdot \rho^\omega$ is a lasso-shaped schedule, we have
    $(\tau \cdot \rho \cdot \beta'_l)(\gst) = \beta'(\gst)$.
Consequently, we can advance one iteration of the loop and
    derive the following from Equation~(\ref{eq:psi-at-beta}):
\begin{equation}\label{eq:one-iter-further}
\infpath{(\tau \cdot \rho \cdot \beta'_l)(\gst)}{\beta'' \cdot \rho^\omega}
    \models \psi_i \end{equation}
    
Notice that in Equation~(\ref{eq:one-iter-further}) we use $\tau \cdot \rho
\cdot \beta''_l$, not $\tau \cdot \rho$.
The definition of~$\imap$ requires that $\imap(w) < |\tau| + |\rho|$.
Since $|\tau \cdot \rho| \le |\tau \cdot \rho \cdot \beta'_l|$, we have
    $\imap(w) \le |\tau \cdot \rho \cdot \beta'_l|$, that is, the
    formula~$\psi_i$ is satisfied at the state $(\tau \cdot \rho \cdot
    \beta'_l)(\gst)$ that either coincides with the state $\pi'(\gst)$ or
    occurs after the state $\pi'(\gst)$ in the path $\infpath{\pi'(\gst)}{\pi''
    \cdot \rho^\omega}$.
From this and Equation~(\ref{eq:one-iter-further}), we conclude that
$\infpath{\pi'(\gst)}{\pi'' \cdot \rho^\omega} \models \ltlF \psi_i$
    holds.

    \item Case: $w$ is covered by a $\ltlG$-node, while $w.i$ is not.
        This case is impossible, since the node with id~$w.i$ is
        the child of the node with id~$w$ in the syntax
        tree~$\syntree(\varphi)$.

\end{enumerate}

From~(a)--(c), Point~\ref{itm:F} follows for the tree
    node~$\left<\ltlF \psi, w\right>$.

\paragraph{Proving~\ref{itm:root}.}

Let $\canform(\varphi) \equiv \psi_0 \wedge \ltlF \psi_1 \wedge \dots \wedge
    \ltlF \psi_k \wedge \ltlG \psi_{k+1}$, for some $k \ge 0$, a propositional
    formula~$\psi_0$, canonical formulas~$\psi_1, \dots, \psi_k$, and a
    formula~$\psi_{k+1}$ that is either a canonical formula, or equals to~$\true$.
We have to show $\infpath{\gst}{\tau \cdot \rho^\omega} \models
    \canform(\varphi)$. 
To this end, we will show that: (a)~$\gst \models \psi_0$, and
    (b)~$\infpath{\gst}{\tau \cdot \rho^\omega} \models \ltlG \psi_{k+1}$, and
    (c)~$\infpath{\gst}{\tau \cdot \rho^\omega} \models \ltlF \psi_i$ for $1
    \le i \le k$.  From (a)--(c), the required statement immediately follows.

Point~(a) follows directly from Condition~\ref{assume:root} of
    Definition~\ref{def:prop-witness}.

To show~(b), we focus on the case $\psi_{k+1} \ne \true$, as the case
    $\psi_{k+1} = \true$ is trivial. 
Notice that by the definition of the syntax tree~$\syntree(\varphi)$, the
    subformula $\ltlG \psi_{k+1}$ has the id $0.j$ for $j = k+1$, and thus
    $\pcutpoint(0.j) = 0$, which cuts $\tau \cdot \rho$ into
    the empty schedule and $\tau \cdot \rho$ itself.
Thus,~$(b)$ follows directly from the inductive hypothesis~\ref{itm:G}, which
    has already been shown to hold
    for the tree node $\left<\ltlG \psi_{k+1}, 0.j\right>$.

To show~(c), fix an index $i \in \{1, \dots, k\}$. 
Let $\imap(0.i)$ cut $\tau \cdot \rho$ into $\beta'$ and $\beta''$. 
The inductive hypothesis~\ref{itm:F} has been shown to hold for the tree node
    $\left<\ltlF \psi_i, 0.i\right> \in \syntree(\varphi)$, and thus we have:
\begin{equation}\label{eq:psi-at-beta-root}
    \infpath{\beta'(\gst)}{\beta'' \cdot \rho^\omega} \models \psi_i.
\end{equation}    

Since $\infpath{\beta'(\gst)}{\beta'' \cdot \rho^\omega}$ is a suffix of
    $\infpath{\gst}{\tau \cdot \rho^\omega}$, from
    Equation~\ref{eq:psi-at-beta-root}, we immediately obtain the required
    statement: $\infpath{\gst}{\tau \cdot \rho^\omega} \models \ltlF \psi_i$.

By collecting Points (a)--(c), we immediately arrive at: $\infpath{\gst}{\tau
    \cdot \rho^\omega} \models \canform(\varphi)$. 
By Definition~\ref{thm:canonical-form}, the formula $\varphi$ is
    equivalent to~$\canform(\varphi)$.
This finishes the proof.
\end{proof}

\recallthm{thm:witness-completeness}{\thmwitnesscompleteness}

\makeatletter{}\begin{proof}

We apply Proposition~\ref{prop:unwinding} to find the required
    number $K \ge 0$ and the cut function~$\imap$.
It remains to show that Conditions~\ref{assume:root}--\ref{assume:Gfin-loop}
    of Definition~\ref{def:prop-witness} are satisfied for the
    configuration~$\gst$ and the lasso~$(\tau \concat \rho^K) \concat \rho^\omega$.

\paragraph{Showing Condition~\ref{assume:root}.} This condition does not depend
    on the structure of~$\imap$. 
Let $\canform(\varphi) = \psi_0 \wedge \ltlF \psi_1 \wedge \dots \ltlF \psi_k
    \wedge \ltlG \psi_{k+1}$. 
Since $\infpath{\gst}{(\tau \concat \rho^K) \concat \rho^\omega} \models \varphi$,
    we immediately have $\gst \models \psi_0$ and $\infpath{\gst}{(\tau \concat
    \rho^K) \concat \rho^\omega} \models \ltlG \psi_{k+1}$. 
By the semantics of \LTL, the latter implies that for all configurations
    $\sigma'$ visited by the path $\infpath{\gst}{(\tau \concat \rho^K) \concat
    \rho^\omega}$, it holds that $\sigma' \models \prop(\psi_{k+1})$. 
Since $\finpath{\gst}{(\tau \concat \rho^K) \concat \rho}$ is a subsequence of
    $\infpath{\gst}{(\tau \concat \rho^K) \concat \rho^\omega}$, we immediately
    arrive at $\setconf{\gst}{(\tau \concat \rho^K) \concat \rho} \models
    \prop(\psi_{k+1})$.

\paragraph{Showing Conditions~\ref{assume:Gfin-prefix}
    and~\ref{assume:Gfin-loop}.} Let $\psi = \psi_0 \wedge \ltlF \psi_1 \wedge
    \dots \wedge \ltlF \psi_k \wedge \ltlG \psi_{k+1}$. 
Further, assume that~$\imap(v)$ cuts~$(\tau \concat \rho^K) \concat \rho$
    into~$\pi'$ and~$\pi''$. 
By Proposition~\ref{prop:unwinding}, we have $\infpath{\pi'(\gst)}{\pi''
    \concat \rho^\omega} \models \psi$. 
Thus, we have the following:
\begin{align}
\infpath{\pi'(\gst)}{\pi'' \concat \rho^\omega} &\models \psi_0 \label{eq:prop} \\
    \infpath{\pi'(\gst)}{\pi'' \concat \rho^\omega} &\models \ltlG
    \prop(\psi_{k+1}) \label{eq:always}
\end{align}
            
It remains to prove the specific statements about~\ref{assume:Gfin-prefix}
and~\ref{assume:Gfin-loop}:

\begin{enumerate}
    \item Case $\imap(v) < |\tau \concat \rho^K|$. We have to show
    Condition~\ref{assume:Gfin-prefix}.

The path $\finpath{\pi'(\gst)}{\pi''}$ is a subsequence of
$\infpath{\pi'(\gst)}{\pi'' \concat \rho^\omega}$. 
Thus, from Equation~(\ref{eq:always}), we obtain that for every
    configuration~$\gst'$ visited by the finite path
    $\finpath{\pi'(\gst)}{\pi''}$, it holds that $\gst' \models
    \prop(\psi_{k+1})$. In other words:
\begin{equation}
\setconf{\pi'(\gst)}{\pi''} \models \prop(\psi_{k+1}) \label{eq:setconf-prefix}
\end{equation}

Equations~(\ref{eq:prop}) and~(\ref{eq:setconf-prefix}) give us
    Condition~\ref{assume:Gfin-prefix}.

    \item Case $\imap(v) \ge |\tau \concat \rho^K|$. We have to show
    Condition~\ref{assume:Gfin-loop}.
    In this case, $\pi' = (\tau \concat \rho^K) \concat \pi'_l$
        for some schedule $\pi'_l$.
   
    Consider the configuration~$\gst' = 
        (\tau \concat \rho^K \concat \rho \concat \pi'_l)(\gst)$,
    that is, $\gst'$ is the result of applying to~$\gst$
    the prefix $\tau \concat \rho^K$, one iteration
    of the loop~$\rho$, and then the first part of the loop~$\pi'_l$.
    The configuration~$\gst'$ is located at the cut point~$\imap(v)$
    in the loop, and the path $\finpath{\gst'}{\pi'' \concat \pi'_l}$
    reaches the same configuration again, i.e.,
        $(\pi'' \concat \pi'_l)(\gst') = \gst'$.
    From Equation~(\ref{eq:always}), we have
    that the propositional formula~$\prop(\psi_{k+1})$ holds
    on the path $\finpath{\gst'}{\pi'' \concat \pi'_l}$.
    Since both paths $\finpath{\gst'}{\pi'' \concat \pi'_l}$ and
    $\finpath{(\tau\concat\rho^K)(\gst)}{\rho}$ visit all configurations
    of the loop, we have:
        \begin{equation}
    \setconf{(\tau\concat\rho^K)(\gst)}{\rho} \models
        \ltlG \prop(\psi_{k+1}) \label{eq:setconf-loop}
    \end{equation}

Equations~(\ref{eq:prop}) and~(\ref{eq:setconf-loop}) give us
    Condition~\ref{assume:Gfin-loop}.

\end{enumerate}

The theorem follows.
\end{proof}

\makeatletter{}

\section{Detailed Proofs for Section~\ref{sec:repr}}

From now on we fix a threshold  automaton~$\Sk = (\local, \initlocal,
     \globset,$ $\paraset,$ $\ruleset,$ $\ResCond)$ and conduct our
     analysis in Sections~\ref{subsec:cav15}
     to~\ref{subsec:liveness3} for this~$\Sk$.

\subsection{Preliminaries}\label{subsec:cav15}

We start by formalizing the notion of a \emph{thread}.

\begin{definition}[Thread]\label{def:thread}
For a configuration $\sigma$ and a schedule $\tau= \tau_1\concat t_1\concat \tau_2\concat
     \ldots\concat t_k\concat \tau_{k+1}$ applicable to $\sigma$, we define the
     sequence of transitions $\vartheta=t_1,\ldots, t_k$, $k>0$ to be
     a thread of $\sigma$ and $\tau$ if  
\begin{enumerate}
   \item $t_i.\trfactor=1$, for every $1\leq i\leq k$,
   \item $t_i.\tostate=t_{i+1}.\fromstate$, for every $1\le i < k$.
\end{enumerate}
For thread $\vartheta$, by $\vartheta.\fromstate$ and
$\vartheta.\tostate$ we denote $t_1.\fromstate$ and $t_k.\tostate$,
respectively. 
\end{definition}

\begin{definition}[Naming, Projection, and Decomposition]\label{def:Naming} 
A \emph{naming} is a function $\naming\colon
     \Natural \rightarrow \Natural$.
For a schedule $\tau$, and a set $S\subseteq \Natural$, by $\proj{\tau}{\naming,S}$ 
    we denote the sequence of transitions $\tau[j]$ satisfying $\naming(j)\in S$ 
    that preserves the order of transitions from $\tau$, i.e., 
    for all $j_1, j_2, l_1, l_2$, $j_1< j_2$, if $\proj{\tau}{\naming,S}[l_1]=\tau[j_1]$ and
    $\proj{\tau}{\naming,S}[l_2]=\tau[j_2]$, 
    then $l_1<l_2$.
If $S$ is a one-element set $\{i\}$, we write $\proj{\tau}{\naming,i}$
    instead of $\proj{\tau}{\naming,\{i\}}$.
We use the notation $\threads(\sigma,\tau,\naming)$ for the set
     $\{i \colon \proj{\tau}{\naming,i} \text{ is
     a thread of } \sigma \text{ and } \tau \}$.
For a configuration $\sigma$ and a schedule $\tau$, a naming is called
     a \emph{decomposition} of $\sigma$ and $\tau$ if
\begin{enumerate}
\item  for all
     $i\in\Natural$, $\proj{\tau}{\naming,i}$  is a thread of
     $\sigma$ and $\tau$, or $\proj{\tau}{\naming,i}$ is the empty sequence.
\item for all $\ell \in \local$, $\sigma.\counters[\ell] \ge | \{ i \colon
  i \in \threads(\sigma,\tau,\naming) \wedge \proj{\tau}{\naming,i}.\fromstate = \ell    \}  |$
\end{enumerate}
\end{definition} 

\begin{example}\label{ex:decomposition}
Let us reconsider Example~\ref{ex:slides} and
Figure~\ref{fig:movingthread}, with configuration~$\gst_1$
    where $\sigma_1.\counters[\ell_0]= \sigma_1.\counters[\ell_1]=
    \sigma_1.\counters[\ell_2] = 1$ and $\sigma_1.\counters[\ell_3]=0$, 
    and the
    schedule~$\tau=(r_1,1),$ $(r_6,1),$ $(r_4,1),$ $(r_2,1),$ $(r_4,1)$. 
The function~$\naming\colon\Natural \rightarrow \Natural$ 
    with $\naming(1)=\naming(5)=1$, $\naming(2)=\naming(4)=2$,
    $\naming(3)=3$, and $\naming(k)=4$ for every $k\ge 6$, is a
    naming. We will now see that $\naming$ is a decomposition by
    checking the two points.
    
(1) Since $\naming(1)=\naming(5)=1$, the projection
     $\proj{\tau}{\naming,1}$ consists of the  first and the fifth
     transition, in that particular order,  i.e.,
     $\proj{\tau}{\naming,1}=(r_1,1),(r_4,1)$.
This is a thread, as the factor of both transitions is~$1$,  and
     $r_1.\tostate=\ell_2=r_4.\fromstate$.
Similarly, $\proj{\tau}{\naming,2}=(r_6,1),(r_2,1)$, and
     $\proj{\tau}{\naming,3}=(r_4,1)$ are threads.
Besides, $\proj{\tau}{\naming,4}$ is the empty sequence,  as
     numbers mapping to~4 are $n\ge 6$,  and~$\tau$ has length~5,
     i.e., there is no transition~$\tau[n]$ for $n\ge 6$. Further, for
     every $i>4$, $\proj{\tau}{\naming,i}$ is the empty sequence,  as
     there is no $m\in\Natural$, with $\naming(m)=i$.
Thus, $\threads(\sigma,\tau,\naming)=\{1,2,3\}$.
Note that $\newrest{\tau}{\naming}{2}=\proj{\tau}{\naming,\{1,3\}}=
     (r_1,1),(r_4,1),(r_4,1).$ 

(2) As $\proj{\tau}{\naming,2}=r_6.\fromstate=\ell_0$,
    we obtain $\{ i \colon i \in \threads(\sigma,\tau,\naming)
    \wedge \proj{\tau}{\naming,i}.\fromstate = \ell_0\}=\{2\}$,
    and $\gst_1.\counters[\ell_0]=1 \ge |\{2\}|$.
Similarly we check this inequality for other local states, 
    and conclude that~$\naming$ is a decomposition of~$\gst_1$ and~$\tau$.
\end{example}

\begin{proposition}\label{prop:prefixdecomposition} If $\sigma$ is a configuration,
     $\tau$ is a steady schedule applicable to~$\sigma$, and $\naming$
     is a decomposition of $\sigma$ and~$\tau$, then for each prefix
     $\tau'$ of $\tau$, the naming~$\naming$ is  a
     decomposition of $\sigma$ and $\tau'$. Further
     $\threads(\sigma,\tau',\naming) \subseteq \threads(\sigma,\tau,\naming)$.
\end{proposition}

From \cite[Prop.\ 12]{KVW16:IandC} we directly obtain:

\begin{proposition}\label{prop:threadcounts}
If $\sigma$ is a configuration, $\tau$ is a steady schedule
applicable to $\sigma$, and $\naming$ is a decomposition of $\sigma$
and~$\tau$, then for all $\ell$ in $\local$:
\begin{multline*}\tau(\sigma).\counters[\ell] = \sigma.\counters[\ell] +
| \{ i \colon  i \in \threads(\sigma,\tau,\naming) 
\wedge \proj{\tau}{\naming,i}.\tostate = \ell \} |\\
-  | \{  i \colon i \in \threads(\sigma,\tau,\naming) 
\wedge \proj{\tau}{\naming,i}.\fromstate = \ell \} |
.\end{multline*}
\end{proposition}

\newcommand{\propexistsdecomp}{If $\sigma$ is a configuration, $\tau$ is a steady conventional
     schedule applicable to~$\sigma$, then there is a decomposition of
     $\sigma$ and $\tau$.}

\begin{proposition}\label{prop:existsdecomp}
\propexistsdecomp
\end{proposition}

\newcommand{\proofexistsdecomp}{
\begin{proof}
We have to prove the two properties of Definition~\ref{def:Naming}. We
do so by induction on the length of $\tau$.

\begin{itemize}
\item
$|\tau| = 1$.
Let $\tau=t_1$ for a transition~$t_1$, and let $\naming$ be the identity
     function.
Then, $\proj{\tau}{\naming,1} = t_1$ is a thread and for all $i>1$,
the sequence $\proj{\tau}{\naming,i}$ is empty.
As~$\tau$ is applicable to $\sigma$,
     $\sigma.\counters[t_1.\fromstate]\ge 1$.

\item
$|\tau| > 1$.  Let $\tau= \tau' \concat t_{|\tau|}$, and let $\naming'$ be a
decomposition of $\sigma$ and $\tau'$, which exists by the induction
hypothesis. 
We distinguish two cases for $T=\{ i \colon
i \in \threads(\sigma,\tau',\naming') \wedge \proj{\tau'}{\naming',i}.\tostate
= t_{|\tau|}.\fromstate \}$:

\begin{itemize}
\item If $T \ne \emptyset$, then for some $j\in T$
 let
$$
\naming(k) = 
\begin{cases}
j & \text{if } k= |\tau| \\
\naming'(k) & \text{otherwise,}
\end{cases}
$$
that is, we append transition $t_{|\tau|}$ to thread $j$.
Therefore, $\threads(\sigma,\tau,\naming) =
     \threads(\sigma,\tau',\naming')$ and consequently for all
     $\ell\in\local$ we have $ \{  i \colon i \in
     \threads(\sigma,\tau,\naming) \wedge
     \proj{\tau}{\naming,i}.\fromstate = \ell \} =  \{  i \colon i \in
     \threads(\sigma,\tau',\naming') \wedge
     \proj{\tau'}{\naming',i}.\fromstate = \ell \}$. 
Hence, it follows from the induction hypothesis
     that for all $\ell \in \local$,
     $\sigma.\counters[\ell] \ge | \{ i \colon i \in
     \threads(\sigma,\tau,\naming) \wedge
     \proj{\tau}{\naming,i}.\fromstate = \ell    \}  |$.

\item If $T = \emptyset$, then for some $j
  \not\in\threads(\sigma,\tau',\naming')$ let
$$
\naming(k) = 
\begin{cases}
j & \text{if } k= |\tau| \\
\naming'(k) & \text{otherwise,}
\end{cases}
$$
that is, we add a new thread consisting of $t_{|\tau|}$ only.
From applicability of $\tau$ to $\sigma$ follows that
$\tau'(\sigma).\counters[t_{|\tau|}.\fromstate] \ge 1$.
Now from
Proposition~\ref{prop:threadcounts} follows that
\begin{multline*} \sigma.\counters[t_{|\tau|}.\fromstate] \ge 1 -
| T |+\\
+ | \{  i \colon i \in \threads(\sigma,\tau',\naming') \wedge
\proj{\tau'}{\naming',i}.\fromstate = t_{|\tau|}.\fromstate \} | 
.\end{multline*}
As $| T | = 0$ in this case and since by construction 
$ | \{  i \colon i \in \threads(\sigma,\tau',\naming') \wedge
\proj{\tau'}{\naming',i}.\fromstate = t_{|\tau|}.\fromstate \} | =
 | \{  i \colon i \in \threads(\sigma,\tau,\naming) \wedge
\proj{\tau}{\naming,i}.\fromstate = t_{|\tau|}.\fromstate \} | -1$, we obtain
that  $\sigma.\counters[t_{|\tau|}.\fromstate] \ge | \{  i \colon 
i \in \threads(\sigma,\tau,\naming) \wedge
\proj{\tau}{\naming,i}.\fromstate = t_{|\tau|}.\fromstate \} |$ as required.
For the other components of $\sigma.\counters$, the proposition
follows from the induction hypothesis as $\threads(\sigma,\tau,\naming) =
     \threads(\sigma,\tau',\naming') \cup \{j\}$.
\end{itemize}

\end{itemize}

\end{proof}
}
\proofintext{\proofexistsdecomp}
\proofexistsdecomp

\newcommand{\propleftmove}{
If $\sigma$ is a configuration, $\tau = \tau_1
     \concat t_{i-1} \concat t_i \concat \tau_2$ is a steady schedule
     applicable to $\sigma$, $\naming$ is a decomposition of $\sigma$
     and $\tau$, and $\naming(i-1) \ne \naming(i)$, then  $\tau_1
     (\sigma).\counters[t_i.\fromstate] \ge 1$.
}

\newcommand{\proofleftmove}{
\begin{proof} From Proposition~\ref{prop:threadcounts} follows that
     \begin{multline*}
     \tau_1(\sigma).\counters[t_i.\fromstate] =
     \sigma.\counters[t_i.\fromstate] + | \{ k \colon
     \proj{\tau_1}{\naming,k}.\tostate = t_i.\fromstate \} |\\ - | \{k \colon
     \proj{\tau_1}{\naming,k}.\fromstate = t_i.\fromstate \} |.
     \end{multline*}
We distinguish two cases: 
\begin{itemize}
\item If $t_{i} = \proj{\tau}{\naming,\naming(i)}[1]$, that is, it is the first
     in the thread, then $| \{
     k \colon \proj{\tau}{\naming,k}.\fromstate = t_i.\fromstate \}| >
     | \{k \colon \proj{\tau_1}{\naming,k}.\fromstate = t_i.\fromstate \}|$.
By assumption, $\sigma.\counters[t_i.\fromstate] \ge | \{ k \colon
     \proj{\tau}{\naming,k}.\fromstate = t_i.\fromstate\} |$.
Thus, $\tau_1(\sigma).\counters[t_i.\fromstate] > | \{ k
     \colon   \proj{\tau_1}{\naming,k}.\tostate = t_i.\fromstate \} |
     \ge 0$, which proves the proposition in this case.

\item Otherwise, by Definition~\ref{def:Naming}~(2), we have that
     $\sigma.\counters[t_i.\fromstate] -  | \{k \colon
     \proj{\tau_1}{\naming,k}.\fromstate = t_i.\fromstate \} | \ge 0$.
Therefore, it holds that $\tau_1(\sigma).\counters[t_i.\fromstate] \ge | \{ k
     \colon  \proj{\tau_1}{\naming,k}.\tostate = t_i.\fromstate \} |$.
As $\tau_1$ contains the prefix of $\proj{\tau}{\naming,\naming(i)}$, we find
     that $| \{k \colon \proj{\tau_1}{\naming,k}.\tostate =
     t_i.\fromstate \} | \ge 1$ such that
     $\tau_1(\sigma).\counters[t_i.\fromstate] \ge 1$ as required.
\end{itemize}
\end{proof}
}
\proofintext{\proofleftmove}

\begin{proposition} \label{prop:leftmove}
\propleftmove
\end{proposition}
\proofleftmove

Now we show that in a steady
     schedule, a transition of a thread commutes with transitions of
     other threads.
This will  allow us to move whole threads.

\begin{definition}[Move]\label{def:move}
For a schedule $\tau$, and a natural number~$i$, $1<i\le |\tau|$,
  the schedule $\move{\tau}{i}$ is obtained by
  moving the $i$th transition of $\tau$ to the left, and
  naming~$\move{\naming}{i}(k)$ is defined accordingly,
  for every $k\in\Natural$,
  i.e.,
$$
\move{\tau}{i}[k] = 
\begin{cases}
\tau[i] & \text{if $k=i-1$}\\
\tau[i-1] & \text{if $k=i$}\\
\tau[k] & \text{otherwise,}
\end{cases}
$$
$$
\move{\naming}{i}(k) = 
\begin{cases}
\naming(i) & \text{if $k=i-1$}\\
\naming(i-1) & \text{if $k=i$}\\
\naming(k) & \text{otherwise.}
\end{cases}
$$
For natural numbers $n$ and $m$, where $1\le n\le m\le |\tau|$, we define 
    $\movemove\tau{n}{m}$ to be the schedule obtained from $\tau$
    by moving the $m$th transition of $\tau$ to the $n$th position
    (that is $m-n$ times to the left),
    and naming $\movemove{\naming}{n}{m}(k)$ accordingly,
    for every $k\in\Natural$, i.e.,
    $$\movemove{\tau}{n}{m}\,=\,
    \move{(\ldots (\move{(\move{\tau}{m})}{m -1})\ldots )}{n+1}\quad \mbox{ and}$$
    $$\movemove{\naming}{n}{m}(k)=
    \move{(\ldots (\move{(\move{\naming}{m})}{m -1})\ldots )}{n+1}(k).$$
\end{definition}

\begin{example}
Note that if $m=n$, then $\movemove{\tau}{n}{m}=\tau$ and 
    $\movemove{\naming}{n}{m}=\naming$.
If $\tau=t_1,t_2,\ldots,t_{|\tau|}$, then 
    for $i, n, m \in \Natural$ with $n<m\le |\tau|$, and $i\le |\tau|$, it is
    $\move{\tau}{i}=t_1,\ldots,t_{i-2},t_{i},t_{i-1},t_{i+1},\ldots,t_{|\tau|}$
    and
    $\movemove{\tau}{n}{m}=
    t_1,\ldots,t_{n-1},t_m,t_{n},t_{n+1},
    \ldots,t_{m-1},t_{m+1},\ldots,t_{|\tau|}$.
\end{example}

\newcommand{\propmovingleft}{ If $\sigma$ is a configuration, $\tau$ is a steady
     schedule applicable to $\sigma$, and $\naming$ is a decomposition
     of $\sigma$ and $\tau$, then for every $i\in\Natural$, 
     if $1< i \le |\tau|$ and $\naming(i-1) \ne \naming(i)$, then 
\begin{enumerate}
\item $\move{\tau}{i}$ is a steady schedule applicable to $\sigma$,
\item $\move{\naming}{i}$ is a decomposition of $\sigma$ and $\move{\tau}{i}$,
      and $\proj{\move{\tau}{i}}{\move{\naming}{i},j}=\proj{\tau}{\naming,j}$,
      for every $j\in\threads(\sigma,\tau,\naming)$,
\item $\move{\tau}{i}(\gst)=\tau(\gst)$.
\end{enumerate}}

\begin{proposition}\label{prop:movingleft}
\propmovingleft
\end{proposition}

\newcommand{\proofmovingleft}{
\begin{proof}
(1) To prove this we have to show that (1a) $\tau[i]$ is applicable to
$\tau^{i-2}(\sigma)$, and that (1b) $\tau[i-1]$ is applicable to
$\tau^{i-2} \concat \tau[i] (\sigma)$. Point (1) then follows from
commutativity of addition and subtraction on the counters.
\begin{enumerate}
\item[1a]
Since $\tau$ is a steady schedule, then it suffices to show that
     $\tau^{i-2}(\sigma).\counters[\tau[i].\fromstate] \ge 1$, which
 follows from Proposition~\ref{prop:leftmove}.

\item[1b] If $\tau[i].\fromstate \ne \tau[i-1].\fromstate$, then
     $\tau^{i-2}\concat \tau[i](\sigma).\counters[\tau[i-1].\fromstate]
     \ge \tau^{i-2}(\sigma).\counters[\tau[i-1].\fromstate]$ and the
     statement follows from applicability of $\tau$ to $\sigma$.
Otherwise, from applicability of  $\tau$ to $\sigma$ for the case
  $\tau[i-1].\fromstate = \tau[i-1].\tostate$ it follows that
     $\tau^{i-2}(\sigma).\counters[\tau[i-1].\fromstate] \ge 1$, and for
 $\tau[i-1].\fromstate \ne \tau[i-1].\tostate$
it follows that
 $\tau^{i-2}(\sigma).\counters[\tau[i-1].\fromstate] \ge 2$. In both cases the
     statement follows.
\end{enumerate}

(2)
We firstly show that every transition from $\proj{\move{\tau}{i}}{\move{\naming}{i},j}$
    is also in $\proj{\tau}{\naming,j}$.
Let $\move{\tau}{i}[k]$ be a transition from $\proj{\move{\tau}{i}}{\move{\naming}{i},j}$.
Thus, $\move{\naming}{i}(k)=j$. We want to show that $\move{\tau}{i}[k]$
    is also in $\proj{\tau}{\naming,j}$.
We consider three cases:
\begin{itemize}
 \item If $k=i-1$, then $\move{\tau}{i}[k]=\move{\tau}{i}[i-1]=\tau[i]$ and 
	  $\naming(i)=\move{\naming}{i}(i-1)=\move{\naming}{i}(k)=j$.
       As $\naming(i)=j$, then $\tau[i]$ belongs to $\proj{\tau}{\naming,j}$.
       Now $\tau[i]=\move{\tau}{i}[k]$ gives the required.
 \item If $k=i$, then $\move{\tau}{i}[k]=\move{\tau}{i}[i]=\tau[i-1]$ and 
	  $\naming(i-1)=\move{\naming}{i}(i)=\move{\naming}{i}(k)=j$.
       As $\naming(i-1)=j$, then $\tau[i-1]=\move{\tau}{i}[k]$ belongs to $\proj{\tau}{\naming,j}$.
 \item If $k\ne i-1$ and $k\ne i$, then by Definition~\ref{def:move} we have 
	  $\move{\tau}{i}[k]=\tau[k]$ and 
	  $\naming(k)=\move{\naming}{i}(k)=j$.
       Since $\naming(k)=j$, then $\tau[k]$ is in $\proj{\tau}{\naming,j}$.
	  Now $\tau[k]=\move{\tau}{i}[k]$ gives the required.
\end{itemize}

Proving that every transition from $\proj{\tau}{\naming,j}$
    is also in $\proj{\move{\tau}{i}}{\move{\naming}{i},j}$,
    is analogous to the previous direction.

Now we know that for every $j\in\threads(\sigma,\tau,\naming)$, 
    schedules $\proj{\move{\tau}{i}}{\move{\naming}{i},j}$ and
    $\proj{\tau}{\naming,j}$ contain same transitions.
The order of these transitions remains the same, since the only two transitions
    with different positions in $\tau$ and $\move{\tau}{i}$ are
    adjacent transitions from two different threads.
    
Now, knowing that $\naming$ is a decomposition of $\gst$ and $\tau$,
    and that all threads remain the same,
    we conclude that $\move{\naming}{i}$ is a decomposition of 
    $\sigma$ and $\move{\tau}{i}$.

(3) Follows from the step~(2) and Proposition~\ref{prop:threadcounts}.
\end{proof}
}
\proofintext{\proofmovingleft}
\proofmovingleft

\begin{proposition}\label{prop:leftmovefromto}
 Let $\sigma$ be a configuration, let $\tau$ be a steady
     schedule applicable to $\sigma$, and let $\naming$ be a decomposition
     of $\sigma$ and $\tau$.
 If for $n,m\in\Natural$ holds that $1\le n\le  m \le |\tau|$ and
     $\naming(m) \ne \naming(i)$, for every $i$ with $n\le i<m$, then 
\begin{enumerate}
\item $\movemove{\tau}{n}{m}$ is a steady schedule applicable to $\sigma$,
\item $\movemove{\naming}{n}{m}$ is a decomposition of $\sigma$ and $\movemove{\tau}{n}{m}$,
      and $\proj{\movemove{\tau}{n}{m}}{\movemove{\naming}{n}{m},j}=\proj{\tau}{\naming,j}$,
      for every $j\in\threads(\sigma,\tau,\naming)$,
\item $\movemove{\tau}{n}{m}(\gst)=\tau(\gst)$.
\end{enumerate}
\end{proposition}

\begin{proof}
 This statement is a consequence of 
    Proposition~\ref{prop:movingleft} applied inductively $m-n$ times,
    as Definition~\ref{def:move} suggests.
 In the case when $m=n$, the statement is trivially satisfied.
\end{proof}

\begin{proposition}\label{prop:movethreadtomiddle}
  Let $\sigma$ be a configuration, let $\tau$ be a steady
     schedule applicable to $\sigma$, and let $\naming$ be a decomposition
     of $\sigma$ and $\tau$.
  Fix an $i\in \threads(\sigma,\tau,\naming)$.
  Let us denote
    $\tau^*=\tau'\concat\proj{\tau}{\naming,i}\concat \tau''$, 
        such that 
    $\tau'$ is a possibly empty prefix of $\tau$ which contains
    no transitions from $\proj{\tau}{\naming,i}$, and
    $\tau'\concat\tau''=\newrest{\tau}{\naming}{i}$.
  Then 
  \begin{enumerate}
   \item $\tau^*$ is a steady schedule applicable to $\sigma$,
   \item there exists a decomposition $\naming^*$ of $\gst$ and $\tau^*$
      such that $\proj{\tau^*}{\naming^*,l}=\proj{\tau}{\naming,l}$, for every 
      $l\in \threads(\sigma,\tau,\naming)$.
   \item $\tau^*(\gst)=\tau(\gst)$.
  \end{enumerate}
     
\end{proposition}

\begin{proof}
 Let us firstly enumerate all transitions from $\proj{\tau}{\naming,i}$,
    for example, let $\proj{\tau}{\naming,i}=t_{n_1},t_{n_2},\ldots,t_{n_k}$,
    for $1\le n_1 <n_2<\cdots <n_k\le |\tau|$.
 Thus, $\tau'=t_1,\ldots,t_{s}$, for $0\le s < n_1$.
 The idea is that we move transitions from $\proj{\tau}{\naming,i}$, one by one, 
    to the left, namely $t_{n_1}$ to the place $(s+1)$ in $\tau$,
    then $t_{n_2}$ to the place $s+2$, and so on,
    by repeatedly applying Proposition~\ref{prop:movingleft}, that 
    preserves the required properties.
 Formally, $\tau^*=\movemove{(\ldots(\movemove{(\movemove{\tau}{s+1}{n_1})}{s+2}{n_2})\ldots)}{s+k}{ n_k }$.

 For every $j$ with $1\le j\le k$, we denote
 $$\tau_j=\movemove{(\ldots(\movemove{(\movemove{\tau}{1}{n_1})}{2}{n_2})\ldots)}{j}{n_j}
    \mbox{ and }$$
 $$\naming_j=\movemove{(\ldots(\movemove{(\movemove{\naming}{1}{n_1})}{2}{n_2})\ldots)}{j}{n_j}.$$
 We prove by induction that for every $j$, with $1\le j\le k$,
  it holds that:
    \begin{enumerate}
    \item[a)] $\tau_j$ is a steady schedule applicable to $\sigma$,
    \item[b)] $\naming_j$ is a decomposition of $\sigma$ and $\tau_j$, and
      $\proj{\tau_j}{\naming_j,l}=\proj{\tau}{\naming,l}$, for every 
      $l\in \threads(\sigma,\tau,\naming)$,
    \item[c)] $\tau_j(\gst)=\tau(\gst)$.
    \end{enumerate}
 
 If $j=1$, then $\tau_j=\movemove{\tau}{s+1}{n_1}$.
 Note that $\naming(n_1) \ne \naming(m)$, 
    for every $m$ with $s+1\le m<n_1$, since $t_{n_1}$ is the first
    transition in $\proj{\tau}{\naming,i}$, or in other words, the 
    smallest number mapped to $i$ by $\naming$.
  Now, as $1\le s+1\le n_1 \le |\tau|$, the required holds  
    by Proposition~\ref{prop:leftmovefromto}.
 
 Assume that the statement holds for $j$, and let us show that 
    then it holds for $j+1$ as well.
 Note that $\tau_{j+1}=\movemove{(\tau_j)}{(j+1)}{n_{(j+1)}}$.
 We show that we can apply Proposition~\ref{prop:leftmovefromto} to
    $\gst$, $\tau_j$, $\naming_j$, $s+j+1$ and $n_{j+1}$.
 By induction hypothesis, $\tau_j$ is a steady schedule applicable to $\gst$,
    and $\naming_j$ is a decomposition of $\sigma$ and $\tau_j$.
 From the assumption that $1\le s+1\le n_1 <n_2<\cdots <n_k\le |\tau|$,
    follows that $1\le s+j+1\le n_{j+1} \le |\tau|$.
 By construction, $\tau_j$ has a form 
    $\tau'\cdot t_{n_1}\cdot t_{n_2}\cdot \ldots\cdot t_{n_j}\cdot
    \rho_1 \cdot t_{n_{j+1}}\cdot \rho_2$,
    where $\rho_1\cdot\rho_2=\tau''$.
 Note that no transition from $\rho_1$ is in $\proj{\tau}{\naming,i}$,
    which is, by induction hypothesis, same as 
    $\proj{\tau_j}{\naming_j,i}$.
 Thus, $\naming_j(n_{j+1})\ne \naming_j(m)$, for every $m$ with $s+j+1 < m\le n_{j+1}$.
 Now we can apply Proposition~\ref{prop:leftmovefromto}, and obtain the required.
\end{proof}

\begin{proposition}\label{prop:movethreadtobeginning}
  Let $\sigma$ be a configuration, let $\tau$ be a steady
     schedule applicable to $\sigma$, and let $\naming$ be a decomposition
     of $\sigma$ and $\tau$.
  If $i\in \threads(\sigma,\tau,\naming)$, and we denote
     $\tau^*=\proj{\tau}{\naming,i}\cdot\newrest{\tau}{\naming}{i}$, then 
  \begin{enumerate}
   \item $\tau^*$ is a steady schedule applicable to $\sigma$,
   \item there exists a decomposition $\naming^*$ of $\gst$ and $\tau^*$
      such that $\proj{\tau^*}{\naming^*,l}=\proj{\tau}{\naming,l}$,
      for every $l\in \threads(\sigma,\tau,\naming)$,
   \item $\tau^*(\gst)=\tau(\gst)$.
  \end{enumerate}
\end{proposition}

\begin{proof}
 This Proposition is a special case of Proposition~\ref{prop:movethreadtomiddle},
  when~$\tau'$ is the empty schedule.
\end{proof}

\begin{proposition}\label{prop:movethreadsmix}
  Let $\sigma$ be a configuration, let $\tau$ be a steady
     schedule applicable to $\sigma$, and let $\naming$ be a decomposition
     of $\sigma$ and $\tau$.
  Fix $i,j\in \threads(\sigma,\tau,\naming)$.
  If $\proj{\tau}{\naming,j}$ can be written as 
    $\proj{\tau}{\naming,j}^1\cdot\proj{\tau}{\naming,j}^2$,
    for some schedules $\proj{\tau}{\naming,j}^1$ and
    $\proj{\tau}{\naming,j}^2$,
    and if we denote
    $$\tau^*=\proj{\tau}{\naming,j}^1\cdot\proj{\tau}{\naming,i}\cdot 
     \proj{\tau}{\naming,j}^2\cdot \newrest{\tau}{\naming}{i,j},$$
     then the following holds:
     \begin{enumerate}
      \item $\tau^*$ is a steady schedule applicable to $\sigma$,
      \item there exists a decomposition $\naming^*$ of 
     $\gst$ and $\tau^*$ such that $\proj{\tau^*}{\naming^*,i}=\proj{\tau}{\naming,i}$
     and $\proj{\tau^*}{\naming^*,j}=\proj{\tau}{\naming,j}$, and
      \item $\tau^*(\gst)=\tau(\gst)$.
     \end{enumerate}
\end{proposition}

\begin{proof}
 Firstly, we apply Proposition~\ref{prop:movethreadtobeginning} for 
 configuration $\gst$, schedule $\tau$, decomposition
  $\naming$ and $j\in \threads(\sigma,\tau,\naming)$.
 Then we obtain schedule $\rho=\proj{\tau}{\naming,j}\cdot\newrest{\tau}{\naming}{j}$
  and decomposition $\naming_\rho$ of~$\gst$ and~$\rho$.
 Then we apply Proposition~\ref{prop:movethreadtomiddle} for 
  configuration $\gst$, schedule $\rho$,
  decomposition $\naming_\rho$, $i\in \threads(\sigma,\tau,\naming)$, 
  and a prefix $\proj{\tau}{\naming,j}^1$ of $\rho$ (as $\tau'$ from the proposition).
\end{proof}

\makeatletter{}

Until now, we discussed when transitions can be moved.
In~\cite{KVW15:CAV}, the goal of this movings is to transform a
     schedule into a so-called \emph{representative schedule} that reaches
     the same final configuration (cf.~Example~\ref{ex:swap}).
These representative schedules are highly accelerated, and their
     length can be bounded.
After some preliminary definitions, we recall how these representative
     schedules are constructed, and then give
     Proposition~\ref{prop:representativepath} that establishes that
     representatives maintain an important trace property.

Given a threshold
    automaton~$(\local, \initlocal, \globset, \paraset, \ruleset,\ResCond)$, we
    define the \emph{precedence relation} $\relf$: for a pair of rules $r_1,
    r_2 \in \ruleset$, it holds that {$r_1 \relf r_2$} if and only if
    $r_1.\tostate = r_2.\fromstate$. 
We denote by $\relftrans$ the transitive closure of~$\relf$.
Further, we say that $r_1 \erelf r_2$, if $r_1 \relftrans r_2 \, \wedge
    \, r_2 \relftrans r_1$, or $r_1 = r_2$.
The relation $\erelf$ defines equivalence classes of
     rules.
For a given set of rules~$\ruleset$ let $\ruleclass$ be the set of
     equivalence classes defined by~$\erelf$.
We denote by $\classof{r}$ the equivalence class of rule~$r$.
For two classes $c_1$ and $c_2$ from $\ruleclass$ we write $c_1 \relc
     c_2$ iff there are two rules $r_1$ and $r_2$ in~$\ruleset$
     satisfying $\classof{r_1}=c_1$ and $\classof{r_2}=c_2$ and $r_1
     \relftrans r_2$ and $r_1 \nerelf r_2$.
As the relation~$\relc$ is a strict partial order, there are linear
     extensions of~$\relc$.
We denote by $\linrelc$ a linear
     extension of~$\relc$.

\paragraph{Construction of Representative Schedule.} \label{cons:srep} 

Given a configuration~$\gst$, and a steady schedule $\tau$ applicable
     to $\gst$, $\xrep{\gst}{\tau}$ is generated from $\tau$ by
     repeatedly swapping two neighboring transitions $t_1$ and $t_2$
     if $[t_2]  \linrelc [t_1]$ until no more such transitions exist.
Then all neighboring transitions that belong to the same rule are
     merged into a single (possibly accelerated) transition.

Then the transitions belonging to loops are replaced by a quite
     involved construction in~\cite[Prop.~5]{KVW15:CAV}.
As discussed in Section~\ref{sec:TA}, in this paper we consider the
     restriction that loops are simple (there may be self-loops).
Hence, we can have a simplified construction: If for some~$j$, the
     rules $r_1, r_2, \dots, r_j$ build a loop, then all the
     transitions in the subschedule $\tau_{\mathrm{loop}}$ that belong to the loop
     are replaced by the schedule that is constructed as follows: 
\begin{enumerate}
\item let $\sigma_{\mathrm{end}} = \tau_{\mathrm{loop}}(\sigma_0)$

\item let $\tau'=$ 
     $(r_1,f_1),$ $(r_2,f_2), \dots ,$ $(r_j,f_j),$ $(r_1,f_{j+1}),$
     $(r_2,f_{j+2}),$ $\dots$, $(r_j,f_{2j})$ be the schedule that is 
obtained by
\begin{itemize}
\item If $r_1, \dots, r_j$ appear in $\tau_{\mathrm{loop}}$:
 inductively assigning values to the
    acceleration
     factors $f_i$, for $1 \le i \le 2j$ as follows:
\begin{itemize}
\item for $1 \le i \le j$: \\
 $f_i = \sigma_{i-1}.\counters[r_{i}.\fromstate] -  \min
 (\sigma_0.\counters[r_{i}.\fromstate],
\sigma_{\mathrm{end}}.\counters[r_{i}.\fromstate])$
 and for $t_i = (r_i,f_i)$, we get $\sigma_i =
 t_i(\sigma_{i-1})$ 

\item for $j+1 \le i \le 2j$: \\
 $f_i =  \sigma_{i-1}.\counters[r_{i-j}.\fromstate]
 - \sigma_{\mathrm{end}}.\counters[r_{i-j}.\fromstate]$ and\\
 for $t_i = (r_i,f_i)$, we obtain $\sigma_i =
 t_i(\sigma_{i-1})$ 
\end{itemize}
\item otherwise, that is, if some rules are missing in the schedule, then we set their
acceleration factors to zero. Note that due to the missing rules, the
loop falls apart into several independent chains. Each of this chains
is a subschedule of~$\tau'$, we just have to sum up the acceleration
factors for the present rules. Formally, we proceed in two steps: First,
 if $r_i$ is not present in  $\tau_{\mathrm{loop}}$, and for all
$k<i$, $r_k$  is  present in  $\tau_{\mathrm{loop}}$ then $f_\ell =
0$, for all $\ell$ that satisfy $\ell \le i$ or $\ell \ge
j+i$. Second, for $\ell$ with $i<\ell \le j$, the factor $f_\ell$ is the sum of the
acceleration factors of transitions in $\tau_{\mathrm{loop}}$ with the rule $r_\ell$. For  $\ell$ with
$j < \ell < i+j$, $f_\ell$ is the sum of the acceleration factors of transitions in $\tau_{\mathrm{loop}}$  with the rule $r_{\ell-j}$.

\end{itemize}

\item $\tau''$ is obtained from $\tau'$ by removing all transitions with zero
 acceleration factors 
\item we replace $\tau_{\mathrm{loop}}$ with $\tau''$.
\end{enumerate}

\newcommand{\propLarryBird}{
Let $\tau_{\mathrm{loop}}$ be a schedule applicable to $\sigma_0$ that
consists of transitions
whose rules all belong to the same loop. For all $\ell\in\local$, if
$\sigma_0.\counters[\ell]>0$ and
$\tau_{\mathrm{loop}}(\sigma_0).\counters[\ell]>0$, 
then
$\setconf{\sigma_0}{{\sr{bla}{\sigma_0}{\tau_{\mathrm{loop}}}}} \models \counters[\ell]>0$.
}

\begin{proposition}\label{prop:LarryBird}
\propLarryBird
\end{proposition}

In this way, $\xrep{\gst}{\tau}$ contains a subset of the rules of $\tau$ but
     ordered according to the linear extension $\linrelc$ of the
     control flow of the automaton. 
Thus, from the above construction we directly obtain:

\newcommand{\propMichaelJordan}{
Let~$\gst$ be a configuration and
    let~$\tau$ be a steady schedule applicable to~$\gst$. 
The rules contained in transitions of $\sr{\Ctx}{\gst}{\tau}$ are a
    subset of the 
   rules contained in transitions of $\tau$.}

\begin{proposition}\label{prop:MichaelJordan}
\propMichaelJordan
\end{proposition}

From Proposition~\ref{prop:srep-ex}, we know that we can replace a
     schedule by its representative, and maintain the same final
     state.
In the following propositions, we show that the representative schedule
     also maintains non-zero counters.

\newcommand{\proprepresentativepath}{
Let~$\gst$ be a configuration, and 
    let~$\tau$ be a steady schedule applicable to~$\gst$.
For every $\ell\in\local$, it holds that $\setconf{\gst}{\tau}\models\counters[\ell]> 0$
    implies
    $\setconf{\gst}{\srgen}\models\counters[\ell]> 0.$}

\begin{proposition}\label{prop:representativepath}
\proprepresentativepath
\end{proposition}

\newcommand{\proofrepresentativepath}{
\begin{proof}
Schedule $\srgen$ is constructed by first swapping transitions and then
  reducing loops. 
We first show that swapping maintains $\counters[\ell]> 0$, and then
that reducing loops does so, too.

Consider the sub-path $\sigma_{i-1}, t_i, \sigma_{i},
     t_{i+1},\sigma_{i+1}$ of one schedule in the construction and
     $\sigma_{i-1}, t_{i+1}, \sigma'_{i}, t_{i},\sigma_{i+1}$ be the
     path obtained by swapping.
Assume by ways of contradiction that  $\sigma_{i-1}.\counters[\ell]>0$,
     $\sigma_{i}.\counters[\ell]>0$, and 
     $\sigma_{i+1}.\counters[\ell]>0$, but $\sigma'_{i}.\counters[\ell]
     = 0$.
As $\counters[\ell]$ reduces from $\sigma_{i-1}$ to $\sigma'_{i}$, we
     get $t_{i+1}.\fromstate = \ell$.
By similar reasoning  on $\sigma'_{i}$ and $\sigma_{i+1}$ we obtain
     $t_i.\tostate = \ell$.
It thus holds that $t_i \relf t_{i+1}$, which contradicts that these
     transitions are swapped in the construction of $\srgen$.

Let $\dots, \sigma, \tau, \sigma', \dots$ be the path before the
     loops are replaced and $\tau$ consist of all the transition
     belonging to one loop.
From the above paragraph we know that $\sigma.\counters[\ell]>0$ and
     $\sigma'.\counters[\ell]>0$.
We may thus apply Proposition~\ref{prop:LarryBird} and the proposition
     follows.
\end{proof}
}
\proofintext{\proofrepresentativepath}
\proofrepresentativepath

For threshold-guarded fault-tolerant algorithms, the restrictions we
     put on threshold automata are well justified.
In this paper we used the assumption that all the cycles in threshold
     automata are simple.
In fact this assumption is  a generalization of the TAs we
     found in our benchmarks.
The authors of~\cite{KVW15:CAV} did not make this assumption.
As a consequence, they have to explicitly treat contexts (the guards
     that currently evaluate to true), which lead to context-specific
     representative schedules.
Our restriction allows us to use only one way to construct simple
     representative schedules (cf.\ Section~\ref{cons:srep}).
In addition, with this restriction, we can easily proof
     Proposition~\ref{prop:MichaelJordan}, while this proposition not
     true under the assumptions in~\cite{KVW15:CAV}.
We conjecture that even under their assumption a proposition similar
     to our Proposition~\ref{prop:representativepath} can be proven,
     so that our results can be extended.
As our analysis already is quite involved, these restrictions allow us
     to concentrate on our central results without obfuscating the
     notation and theoretical results.
Still, from a theoretical viewpoint it might be interesting to lift
     the restrictions on loops.

We have now seen how to construct simple representative schedules.
In the following Sections~\ref{subsec:liveness1}
     to~\ref{subsec:liveness3}, we will see how we can construct
     representative schedules that maintain different forms of
     specifications.

\makeatletter{}

\subsection{Representative Schedules maintaining
 \boldmath $\bigvee_{i \in \critical}
    \counters[i] \ne 0$}\label{subsec:liveness1}

\begin{definition}[Types]\label{def:types}
Let $\sigma$ be a configuration, 
     $\tau$ be a schedule applicable to $\sigma$,
     $\vartheta=t_1,\ldots,t_n$ be a thread of $\sigma$ and $\tau$,
         $\firststate{\vartheta}=t_1.\fromstate$, 
    $\laststate{\vartheta}=t_n.\tostate$, 
    $\middlestate{\vartheta}=\{t_i.\tostate\colon 1\le i < n\}$, and $\critical \subseteq \local$.
We say that  $\vartheta$ is of $\critical$-type:
\begin{itemize}
    \item $\typeall$, 
        if $\{\firststate{\vartheta},\laststate{\vartheta}\}\cup \middlestate{\vartheta}
        \subseteq \critical$;
    \item $\typebeg$,
        if $\firststate{\vartheta}\in \critical $, 
        $\laststate{\vartheta}\not\in\critical $;
    \item $\typeend$,
        if $\firststate{\vartheta}\not\in \critical $, 
        $\laststate{\vartheta}\in\critical $;
    \item $\typemid$,
        if $\firststate{\vartheta}\not\in \critical  $, 
        $\laststate{\vartheta}\not\in\critical $,
        $\middlestate{\vartheta}\cap\critical \neq\emptyset$;
    \item $\typebegend$, 
        if $\firststate{\vartheta}\in\critical $,
        $\laststate{\vartheta}\in\critical $,
        $\middlestate{\vartheta} \nsubseteq\critical$
    \item $\typenot$,
        if $(\{\firststate{\vartheta}, \laststate{\vartheta}\} \cup 
        \middlestate{\vartheta})\cap\critical= \emptyset$.
\end{itemize}
\end{definition}

\begin{example}
 Let us consider the threshold automaton from Figure~\ref{fig:movingthread},
    and the subset of local states $\critical=\{\ell_2\}$.
 Schedule $(r_4,1)$ is of $\critical$-type~$\typebeg$,
 schedule $(r_6,1),(r_2,1)$ is of $\critical$-type~$\typeend$,
      and $(r_1,1),(r_4,1)$ is of $\critical$-type~$\typemid$.  
\end{example}

\newcommand{\proponeofatype}{Given a configuration $\gst$, a schedule $\tau$ applicable to $\sigma$, 
    and a subset of local states $\critical$, 
    every thread $\vartheta$ of $\gst$ and $\tau$ is of exactly one 
    $\critical$-type.
}

\begin{proposition}\label{prop:oneofatype}
\proponeofatype
\end{proposition}

\newcommand{\proofoneofatype}{
\begin{proof}
We consider an arbitrary thread $\vartheta$ of $\gst$ and $\tau$.
There are two possibilities for $\firststate{\vartheta}$, namely, 
    $\firststate{\vartheta}\in \critical$ or $\firststate{\vartheta}\not\in \critical$,
    and similarly for $\laststate{\vartheta}$, 
    $\laststate{\vartheta}\in \critical$ or $\laststate{\vartheta}\not\in \critical$. 
Combining these possibilities, we obtain four cases:
\begin{itemize}
\item Assume $\firststate{\vartheta}\in \critical$ and $\laststate{\vartheta}\in \critical$.
    If $\middlestate{\vartheta}\subseteq\critical$, then $\vartheta$ is of $\critical$-type $\typeall$.
    Otherwise, if $\middlestate{\vartheta}\nsubseteq\critical$, then $\vartheta$ is of $\critical$-type $\typebegend$.
\item If $\firststate{\vartheta}\in \critical$ and $\laststate{\vartheta}\not\in \critical$,
    then $\vartheta$ is of $\critical$-type $\typebeg$.
\item If $\firststate{\vartheta}\not\in \critical$ and $\laststate{\vartheta}\in \critical$,
   then $\vartheta$ is of $\critical$-type $\typeend$.
\item Finally, assume $\firststate{\vartheta}\not\in \critical$ and $\laststate{\vartheta}\not\in \critical$.
    If $\middlestate{\vartheta}\cap\critical\neq \emptyset$, then $\vartheta$ is of $\critical$-type $\typemid$.
    Otherwise, if $\middlestate{\vartheta}\cap\critical=\emptyset$, then $\vartheta$ is of $\critical$-type $\typenot$.
\end{itemize}
\end{proof}
}
\proofintext{\proofoneofatype}
\proofoneofatype

\newcommand{\propTypeA}{
Let $\sigma$ be a configuration, and let $\tau$ be a steady conventional
     schedule applicable to $\sigma$.
If there exists a decomposition $\naming$ of $\sigma$ and $\tau$ that
     satisfies $|\threads(\sigma,\tau,\naming)| = 1$ and $\tau$ is of
     $\critical$-type~$\typeall$, then $\xrep{\sigma}{\tau}$ is a thread of
     $\critical$-type $\typeall$.
}

\newcommand{\proofTypeA}{
\begin{proof}
By definition of a thread, the transitions in $\tau$ are ordered by
     the flow relation $\relf$.
Due to our restriction of loops and the construction of representative
schedules, $\xrep{\sigma}{\tau}$ does not contain rules that are not
contained in $\tau$.
Hence, no new intermediate states are added in the construction of
     $\xrep{\sigma}{\tau}$ which proves the proposition.
\end{proof}
}
\proofintext{\proofTypeA}

\begin{proposition}\label{prop:TypeA}
\propTypeA
\end{proposition}
\proofTypeA

\newcommand{\propTypeAacc}{
Let $\sigma$ be a configuration, and let $\tau$ be a steady conventional
     schedule applicable to $\sigma$. Fix a set $\critical\subseteq\local$.
If there exists a decomposition $\naming$ of $\sigma$ and $\tau$ that
     satisfies $|\threads(\sigma,\tau,\naming)| = 1$ and $\tau$ is of
     $\critical$-type~$\typeall$, then 
     $\setconf \gst{\xrep{\sigma}{\tau}}\models \bigvee_{\ell\in
       \critical} \counters[\ell]\neq 0$.}

\newcommand{\proofTypeAacc}{
\begin{proof}
Let $\xrep{\sigma}{\tau}=t_1,\ldots,t_n$, for an $n\in \Natural$.
Since $\tau$ is of $\critical$-type $\typeall$, by Proposition~\ref{prop:TypeA},
    $\xrep{\sigma}{\tau}$ is of $\critical$-type~$\typeall$,
    which yields that for all $1\le i\le n$ 
    both $t_i.\fromstate$ and $t_i.\tostate$ are in $\critical$. 
\begin{itemize}
\item Since $\xrep{\sigma}{\tau}$ is applicable to $\gst$,
        it must be the case that $\gst\models \counters[\ell^*]\neq 0$,
        where $\ell^*=t_1.\fromstate\in\critical$.
\item If $\tau'=t_1,\ldots,t_k$, $1\le k\le n$, is a nonempty prefix of $\xrep{\sigma}{\tau}$, 
    then, by definition of a counter system from Section \ref{sec:countsys},
    we have that $\tau'(\gst).\counters[t_k.\tostate] > 0$,
    and also $t_k.\tostate\in \critical$.
\end{itemize}
Therefore, $\setconf \gst{\xrep{\sigma}{\tau}}\models \bigvee_{\ell\in \critical}
    \counters[\ell]\neq 0$.
\end{proof}
}
\proofintext{\proofTypeAacc}

\begin{proposition}\label{prop:TypeA.acc}
\propTypeAacc
\end{proposition}
\proofTypeAacc

\newcommand{\lemmaOtherThreads}{
Let $\sigma$ be a configuration, 
     let $\tau$ be a steady conventional schedule
     applicable to $\sigma$,
     and let $\naming$ be a decomposition of $\sigma$
     and $\tau$.
If $k\in \threads(\sigma,\tau,\naming)$ and $n\in\Natural$ are such that 
    $t_n$ is the last transition from $\proj{\tau}{\naming,k}$, i.e.,
    $n$ is the maximal number with $\naming(n)=k$,
    then for every prefix $\tau'$ of $\tau$, 
    of length $|\tau'|\geq n$,
    we have that $\tau' (\gst).\counters[\ell]\ne 0$,
    for $\ell=\laststate{\proj{\tau}{\naming,k}}$.}

\newcommand{\proofOtherThreads}{
\begin{proof}
Fix a prefix $\tau'$ of $\tau$ of length at least $n$.
Then, by Proposition~\ref{prop:prefixdecomposition}, 
    $\naming$ is a decomposition of $\gst$ and $\tau'$.
Note that $k\in \threads(\sigma,\tau',\naming)$, and
    $\proj{\tau}{\naming,k} = \proj{\tau'}{\naming,k}$.
Therefore $\proj{\tau'}{\naming,k}.\tostate=
     \proj{\tau}{\naming,k}.\tostate=t_n.\tostate.$
Proposition~\ref{prop:threadcounts}, when applied to $\tau'$, yields 
    \begin{multline*}
     \tau'(\sigma).\counters[t_n.\tostate] = 
    \sigma.\counters[t_n.\tostate]\\ +
    | \{ i \colon   i \in \threads(\sigma,\tau',\naming) \wedge 
    \proj{\tau'}{\naming,i}.\tostate =t_n.\tostate \} |\\
    - | \{  i \colon  i \in \threads(\sigma,\tau',\naming) \wedge 
    \proj{\tau'}{\naming,i}.\fromstate =t_n.\tostate \} |
    \end{multline*} 
    By Definition~\ref{def:Naming}, we have that  
   $\sigma.\counters[t_n.\tostate]-
    | \{  i \colon  i \in \threads(\sigma,\tau',\naming) \wedge 
    \proj{\tau'}{\naming,i}.\fromstate = t_n.\tostate \}  |
    \ge 0$. 
Since $\naming (n) =k \in  \{ i \colon  i \in \threads(\sigma,\tau',\naming) \wedge
    \proj{\tau'}{\naming,i}.\tostate = t_n.\tostate\} $, 
    we conclude that $| \{ i \colon  i \in \threads(\sigma,\tau',\naming) \wedge
    \proj{\tau'}{\naming,i}.\tostate = t_n.\tostate\}|\ge 1$.
Thus $\tau'(\sigma).\counters[t_n.\tostate]\ge ~1$.
\end{proof}
}
\proofintext{\proofOtherThreads}

\begin{lemma}\label{lem:other.threads.do.not.move.process.from.the.first.thread.nonaccelerated}
\lemmaOtherThreads
\end{lemma}
\proofOtherThreads

\newcommand{\lemmawaitingprocessnonacc}{Let $\sigma$ be a configuration, let $\tau$ be a steady conventional schedule
     applicable to $\sigma$, and let $\naming$ be a decomposition of $\sigma$
     and $\tau$.
If $k\in\threads (\sigma,\tau,\naming)$ and $n\in\Natural$ are such that 
    $t_n$ is the first transition from $\proj{\tau}{\naming,k}$, i.e.,
    $n$ is the minimal number with $\naming(n)=k$,
    then for every prefix $\tau'$ of $\tau$,
     of length $|\tau'| < n$,
    we have that $\tau' (\gst).\counters[\ell]\ne 0$, 
    for $\ell=\firststate{\proj{\tau}{\naming, k}}.$
}

\newcommand{\proofwaitingprocessnonacc}{
\begin{proof}
By repeated application of Proposition~\ref{prop:movingleft}, the first
transition of $\proj{\tau}{\naming,k}$ can be moved to the beginning
of the schedule. Applying Proposition~\ref{prop:leftmove} to the
resulting schedule proves this lemma.
\end{proof}
}
\proofintext{\proofwaitingprocessnonacc}

\begin{lemma}\label{lem:waitingprocess.nonacc}
\lemmawaitingprocessnonacc
\end{lemma}
\proofwaitingprocessnonacc

\newcommand{\proptypeoptions}{Let $\sigma$ be a configuration, let $\tau=t_1,\ldots,t_{|\tau|}$ 
    be a nonempty steady conventional schedule
     applicable to $\sigma$, and let $\naming$ be a decomposition of $\sigma$
     and $\tau$. Fix a set $\critical$ of local states. 
If there is no local state $\ell\in \critical$ such that $\setconf\gst{\tau}\models \counters[\ell]\neq 0$,
    but it holds that $\setconf \gst\tau \models
    \bigvee_{\ell\in \critical} \bk[\ell]\neq 0,$
then at least one of the following cases is true:
    \begin{enumerate}
        \item There is at least one thread of $\gst$ and $\tau$, which is of $\critical$-type $\typeall$;
        \item There is a thread of $\critical$-type $\typebeg$ or $\typebegend$, 
            and an additional of $\critical$-type $\typeend$ or $\typebegend$;
        \item There is a thread of $\critical$-type $\typebegend$, 
        and one of $\critical$-type $\typemid$.
\end{enumerate}
}

\begin{proposition}\label{prop:typeoptions}
\proptypeoptions
\end{proposition}

\newcommand{\prooftypeoptions}{
\begin{proof}
Firstly, if $|\threads(\sigma,\tau,\naming)| = 1$,    
   we prove by contradiction that $\tau$ is of $\critical$-type $\typeall$. 
Namely, if we suppose the opposite, we distinguish three cases:
\begin{itemize}
 \item If $\tau$ is of $\critical$-type $\typeend$, $\typemid$ or $\typenot$, 
    then $\gst\not\models \bigvee_{\ell\in \critical} \counters[\ell]\neq 0$, and therefore 
    $\setconf \gst\tau \nmodels\bigvee_{\ell\in \critical} \bk[\ell]\neq 0$;
 \item If $\tau$ is of $\critical$-type $\typebeg$, 
    then $\tau(\gst)\not\models \bigvee_{\ell\in \critical} \counters[\ell]\neq 0$, and therefore 
    again $\setconf \gst\tau \nmodels\bigvee_{\ell\in \critical} \bk[\ell]\neq 0$;
 \item If $\tau$ is of $\critical$-type $\typebegend$, 
    and a $k$, $1\le k < |\tau|$, is such that $t_k.\tostate\not\in \critical$,
    then for the prefix $\tau'$ of $\tau$ of length $k$
    holds that $\tau'(\gst)\not\models \bigvee_{\ell\in \critical} \counters[\ell]\neq 0$.
\end{itemize}
Thus, for all three options we get a contradiction, which tells us that $\tau$ cannot be of any other 
  type, and leaves the only remaining option: that $\tau$ is of $\critical$-type $\typeall$. 
This gives us the case 1.

Otherwise, if $|\threads(\sigma,\tau,\naming)| \ge 2$, we have two options:
\begin{itemize}
\item If one of the threads is of 
    $\critical$-type $\typeall$, then this is the case $1$.
\item If there is no thread of 
    $\critical$-type $\typeall$, we consider two possibilities:
    \begin{itemize}
    \item  There is a thread $\proj{\tau}{\naming,i}$ of $\critical$-type $\typebegend$,
	      for some $i\in \threads(\sigma,\tau,\naming)$.
	   Then, by definition, there is a $k\in\Natural$ such that $\naming(k)=i$ and 
	      $t_k.\tostate\not\in \critical$.
	   Assume by contradiction that we are not is cases 2. nor 3.
           Then, among the other threads, there are no threads of $\critical$-type $\typeall$,
	      $\typebeg$, $\typeend$, $\typemid$, nor $\typebegend$.
	   In other words, all the other threads are of $\critical$-type $\typenot$.
	   Then the prefix $\tau'$ of $\tau$ of length $k$ has the property that
	      $\tau'(\gst)\not\models \bigvee_{\ell\in \critical} \counters[\ell]\neq 0$.
           This is a contradiction with the assumption that $\setconf \gst\tau \models
	      \bigvee_{\ell\in \critical} \bk[\ell]\neq 0$.
    \item There is no thread of $\critical$-type $\typebegend$.
	  Since $\gst\models \bigvee_{\ell\in \critical} \counters[\ell]\neq 0$, there exists
	    an $\ell'\in \critical$ such that $\gst\models \counters[\ell']\neq 0$.
	  From the assumption that $\setconf\gst{\tau}\nmodels \counters[\ell']\neq 0$, we obtain that
            there must exist a thread $\vartheta_1$ with $\firststate{\vartheta_1}=\ell'\in\critical$.
          Since in this case there are no threads of $\critical$-type $\typeall$ nor $\typebegend$,
	    this implies that $\vartheta_1$ is of $\critical$-type $\typebeg$.
          
	  Similarly, since $\tau(\gst)\models \bigvee_{\ell\in \critical} \counters[\ell]\neq 0$,
	    there exists an $\ell''\in \critical$ such that $\tau(\gst)\models \counters[\ell'']\neq 0$.
	  Now, from the assumption that $\setconf\gst{\tau}\nmodels \counters[\ell'']\neq 0$, we obtain that
            there exists a thread $\vartheta_2$ with $\laststate{\vartheta_2}=\ell''\in\critical$.
          Thus, $\vartheta_2$ is of $\critical$-type $\typeend$, and this case is the case $2$.
    \end{itemize}
\end{itemize} 
 Therefore, at least one of the given cases is true.
\end{proof}
}

\proofintext{\prooftypeoptions}
\prooftypeoptions

\newcommand{\proptypeoptionA}{Let $\sigma$ be a configuration, 
     let $\tau$ be a steady conventional schedule
     applicable to $\sigma$, and let $\naming$ be a decomposition of $\sigma$
     and $\tau$. 
Fix a set $\critical$ of local states, and an $i\in \threads(\gst,\tau,\naming)$.
If $\proj{\tau}{\naming,i}$ is a thread of $\critical$-type $\typeall$,
  and if we denote $\ell^*=\laststate{\proj{\tau}{\naming,i}}$, 
  then $\ell^*\in\critical$, and 
  $$\setconf {\proj{\tau}{\naming,i}(\gst)}{ \newrest{\tau}{\naming}{i}} 
  \models\counters[\ell^*]\ne 0.$$
}

\newcommand{\prooftypeoptionA}{
\begin{proof}
Firstly note that $\proj{\tau}{\naming,i}\cdot \newrest{\tau}{\naming}{i}$ is a 
    steady schedule applicable to $\gst$, and 
    $\proj{\tau}{\naming,i}\cdot \newrest{\tau}{\naming}{i}(\gst)=\tau(\gst)$,
    by Proposition~\ref{prop:movethreadtobeginning}.
Let $\tau'$ be a prefix of $\newrest{\tau}{\naming}{i}$.
Then $\proj{\tau}{\naming,i}\cdot\tau'$ is a prefix of
    $\proj{\tau}{\naming,i}\cdot \newrest{\tau}{\naming}{i}$
    of length $l\ge |\proj{\tau}{\naming,i}|$.
By Lemma~\ref{lem:other.threads.do.not.move.process.from.the.first.thread.nonaccelerated}, 
    it is $\proj{\tau}{\naming,i}\cdot\tau'(\gst).\counters[\ell^*]\neq 0$,
    where $\ell^*=\laststate{\proj{\tau}{\naming,i}}$. 
As $\proj{\tau}{\naming,i}$ is of $\critical$-type $\typeall$,
    then $\ell^*\in \critical$.
\end{proof}
}
\proofintext{\prooftypeoptionA}

\begin{proposition}\label{prop:typeoptionA}
\proptypeoptionA
\end{proposition}
\prooftypeoptionA

\newcommand{\proptypeoptionBC}{Let $\sigma$ be a configuration, let $\tau$ be a steady conventional schedule
     applicable to $\sigma$, let $\naming$ be a decomposition of $\sigma$
     and $\tau$, and let $\critical$ be a subset of $\local$. 
If $i,j\in \threads(\gst,\tau,\naming)$ are such that
    $i\neq j$,
    $\proj{\tau}{\naming,i}$ is a thread of $\critical$-type $\typebeg$ or $\typebegend$, and
    $\proj{\tau}{\naming,j}$ is a thread of $\critical$-type $\typeend$ or $\typebegend$,
    then it holds that 
    \begin{enumerate}
    \item[1)]
        $\setconf \gst{\proj{\tau}{\naming,j}} \models\counters[\ell_1]\neq 0$,
        for $\ell_1=\firststate{\proj{\tau}{\naming,i}}\in \critical$,
    \item[2)]
        $\setconf{\proj{\tau}{\naming,j}(\gst)}{\newrest{\tau}{\naming}{j}} \models 
        \counters[\ell_2]\neq 0,$
        for $\ell_2=\laststate{\proj{\tau}{\naming,j}}\in \critical$.
    \end{enumerate}}

\newcommand{\prooftypeoptionBC}{
\begin{proof}
Firstly note that $\proj{\tau}{\naming,j}\cdot \newrest{\tau}{\naming}{j}$ is a 
    steady schedule applicable to $\gst$, and 
    $\proj{\tau}{\naming,j}\cdot \newrest{\tau}{\naming}{j}(\gst)=\tau(\gst)$,
    by Proposition~\ref{prop:movethreadtobeginning}.
    
1) 
Let  $\tau'$ be a prefix of $\proj{\tau}{\naming,j}$.
Note that in this case $\tau'$ is a prefix of 
    $\proj{\tau}{\naming,j}\cdot \newrest{\tau}{\naming}{j}$
     of length $l\le |\proj{\tau}{\naming,j}|$.
From Lemma~\ref{lem:waitingprocess.nonacc} we obtain
    $\tau'(\gst)\models \counters[\ell_1]\neq 0$,
    where $\ell_1=\firststate{\proj{\tau}{\naming,i}}$.
Since $\proj{\tau}{\naming,i}$ is of 
    type $\typebeg$ or $\typebegend$,
    we have that $\ell_1\in\critical$.

2)
Let  $\tau'$ be a prefix of $\newrest{\tau}{\naming}{j}$.
In this case, $\proj{\tau}{\naming,j}\cdot\tau'$ is a prefix of
    $\proj{\tau}{\naming,j}\cdot \newrest{\tau}{\naming}{j}$ 
    of length $l\ge |\proj{\tau}{\naming,j}|$.
By Lemma~\ref{lem:other.threads.do.not.move.process.from.the.first.thread.nonaccelerated} 
    we have that $\proj{\tau}{\naming,j}\cdot\tau'(\gst)\models \counters[\ell_2]\neq 0$, 
    or, equivalently, $\tau'(\proj{\tau}{\naming,j}(\gst))\models \counters[\ell_2]\neq 0$,
    where $\ell_2=\laststate{\proj{\tau}{\naming,j}}$.
Since $\proj{\tau}{\naming,j}$ is of $\critical$-type $\typeend$ or $\typebegend$,
    then $\ell_2\in\critical$. 
\end{proof}
}
\proofintext{\prooftypeoptionBC}

\begin{proposition}\label{prop:typeoptionBC}
\proptypeoptionBC
\end{proposition}
\prooftypeoptionBC

\newcommand{\proptypeoptionED}{Let $\sigma$ be a configuration, let $\tau$ be a steady conventional schedule
     applicable to $\sigma$, and let $\naming$ be a decomposition of $\sigma$
     and $\tau$.
Fix a subset $\critical$ of $\local$. 
For $i,j\in \threads(\gst,\tau,\naming)$,
    let $\proj{\tau}{\naming,i}$ be a thread of $\critical$-type $\typebegend$, and
    let $\proj{\tau}{\naming,j}$ be a thread of $\critical$-type $\typemid$.
Let us write $\proj{\tau}{\naming,j}$ as $\proj{\tau}{\naming,j}^1\cdot\proj{\tau}{\naming,j}^2$,
    where $\laststate{\proj{\tau}{\naming,j}^1}\in \critical$.
If we denote $$\tau^*=\proj{\tau}{\naming,j}^1\cdot\proj{\tau}{\naming,i}\cdot 
    \proj{\tau}{\naming,j}^2\cdot \newrest{\tau}{\naming}{i,j},$$ then we obtain that
    \begin{enumerate}
    \item  $\setconf \gst{\proj{\tau}{\naming,j}^1} \models \counters[\ell_1]\neq 0$, 
	for $\ell_1=\firststate{\proj{\tau}{\naming,i}}\in\critical$,
    \item  $\setconf {\proj{\tau}{\naming,j}^1(\gst)}{\proj{\tau}{\naming,i}} \models
        \counters[\ell_2]\neq 0$, for 
        $\ell_2=\laststate{\proj{\tau}{\naming,j}^1}\in \critical$,
    \item $\setconf {\proj{\tau}{\naming,j}^1\cdot\proj{\tau}{\naming,i}(\gst)}
        {\proj{\tau}{\naming,j}^2\cdot \newrest{\tau}{\naming}{i,j}} \models
        \counters[\ell_3]\neq 0$, for $\ell_3=\laststate{\proj{\tau}{\naming,i}}\in\critical$. 
    \end{enumerate}}

\newcommand{\prooftypeoptionED}{
\begin{proof}
By Proposition~\ref{prop:movethreadsmix},
  $\tau^*$ is a steady schedule applicable to $\gst$, 
  $\tau^*(\gst)=\tau(\gst)$,
  and there exists a decomposition $\naming^*$ of $\gst$ and $\tau^*$
    such that $\proj{\tau}{\naming,i}=\proj{\tau^*}{\naming^*,i}$
    and $\proj{\tau}{\naming,j}=\proj{\tau^*}{\naming^*,j}$.
Let $l_1=|\proj{\tau}{\naming,j}^1|$ and $l_2=|\proj{\tau}{\naming,i}|$.

1)
Let $\tau'$ be a prefix of $\proj{\tau}{\naming,j}^1$, and therefore a prefix 
    of $\tau^*$ of length $l\le l_1$.
By Lemma~\ref{lem:waitingprocess.nonacc}, 
    we have that 
    $\tau'(\gst)\models 
    \counters[\ell_1]\neq 0$, where $\ell_1=\firststate{\proj{\tau}{\naming,i}}$.
Since $\proj{\tau}{\naming,i}$ is of $\critical$-type $\typebegend$, 
    it is $\ell_1\in \critical$. 

2) 
Let $\tau'$ be a prefix of $\proj{\tau}{\naming,i}$.
Then $\proj{\tau}{\naming,j}^1\cdot \tau'$ is a prefix 
    of $\proj{\tau}{\naming,j}^1\cdot \proj{\tau}{\naming,i}$ of length $l\ge l_1$.
We apply 
    Lemma~\ref{lem:other.threads.do.not.move.process.from.the.first.thread.nonaccelerated}
    for the configuration $\gst$, 
    the schedule $\proj{\tau}{\naming,j}^1\cdot \proj{\tau}{\naming,i}$, 
    and the decomposition $\naming^*$.
With the decomposition $\naming^*$, 
    schedule $\proj{\tau}{\naming,j}^1$ is a thread 
    of $\gst$ and $\proj{\tau}{\naming,j}^1\cdot \proj{\tau}{\naming,i}$,
    and therefore from Lemma~\ref{lem:other.threads.do.not.move.process.from.the.first.thread.nonaccelerated}
    we obtain that 
    $\proj{\tau}{\naming,j}^1\cdot\tau'(\gst)\models \counters[\ell_2]\neq 0$,
    or, equivalently, $\tau'(\proj{\tau}{\naming,j}^1(\gst))\models \counters[\ell_2]\neq 0$,
    where $\ell_2= \laststate{\proj{\tau}{\naming,j}^1}$.
From the construction of $\proj{\tau}{\naming,j}^1$ follows
    that $\ell_2 \in\critical$.
    
3)
Let  $\tau'$ be a prefix of $ \proj{\tau}{\naming,j}^2\cdot \newrest{\tau}{\naming}{i,j}$.
Then $\tau''=\proj{\tau}{\naming,i}\cdot \tau'$ is a prefix of 
    $\tau^*_1=\proj{\tau}{\naming,i}\cdot \proj{\tau}{\naming,j}^2\cdot \newrest{\tau}{\naming}{i,j}$ 
    of length $l\ge l_2$. 
We define a naming $\naming_1$ of $\proj{\tau}{\naming,j}^1(\gst)$ and
    $\tau^*_1$, 
    for every $n\in\Natural$, as follows:
     $$\naming_1(n) = \naming^*(n+l_1).$$
Note that $\naming_1$ is a decomposition of $\proj{\tau}{\naming,j}^1(\gst)$ and
    $\tau^*_1$, and  $\proj{\tau^*_1}{\naming_1,i}=\proj{\tau^*}{\naming^*,i}=\proj{\tau}{\naming,i}$.
We apply Lemma~\ref{lem:other.threads.do.not.move.process.from.the.first.thread.nonaccelerated}
    for the configuration $\proj{\tau}{\naming,j}^1(\gst)$, 
    the schedule $\tau^*_1$, 
    and the decomposition $\naming_1$,
    and obtain that for the prefix $\tau''$ of $\tau^*_1$ holds
    $\tau''(\proj{\tau}{\naming,j}^1(\gst)).\counters[\ell_3]\ge~1$,
    or, equivalently, 
    $$\tau'(\proj{\tau}{\naming,i}(\proj{\tau}{\naming,j}^1(\gst))).\counters[\ell_3]\ge~1,$$
    where $\ell_3 =\laststate{\proj{\tau^*_1}{\naming_1,i}}=
    \laststate{\proj{\tau}{\naming,i}}$.
Again, as $\proj{\tau}{\naming,i}$ is of $\critical$-type $\typebegend$, 
    then $\ell_3 \in\critical$.
\end{proof}
}
\proofintext{\prooftypeoptionED}

\begin{proposition}\label{prop:typeoptionED}
\proptypeoptionED
\end{proposition}
\prooftypeoptionED

\newcommand{\propboring}{Let $\gst$ be a confi\-gu\-ration, and 
    let $\tau$ be a steady conventional schedule applicable to $\gst$.
Fix a set $\critical\subseteq\local$.
If there exist a local state $\ell^*\in\critical$ such that 
    $\setconf \gst\tau \models \counters[\ell^*]\ne 0$,
    then  $$\setconf \gst\srgen \models
    \bigvee_{\ell\in \critical} \counters[\ell]\neq 0.$$ }

\newcommand{\proofboring}{
\begin{proof}
If there is a local state $\ell^*\in \critical$ such that $\setconf\gst{\tau}\models \counters[\ell^*]\neq 0$,
    then we have $\setconf \gst\srgen \models\counters[\ell^*]\neq 0$,
    by Proposition~\ref{prop:representativepath}.
Therefore, $\setconf \gst\srgen \models  \bigvee_{\ell\in \critical} \counters[\ell]\neq 0$. 
\end{proof}
}
\proofintext{\proofboring}

\begin{proposition}\label{prop:boring}
\propboring
\end{proposition}
\proofboring

\begin{proposition}\label{prop:srep-ex-concatenation-basic}
 
Let~$\gst$ be a confi\-gu\-ration,
    let $\tau=\tau_1\cdot\ldots\cdot\tau_n$, for $n\ge 1$, 
    be a steady conventional schedule applicable to $\gst$.
Fix a set $\critical\subseteq\local$. 
If we denote 
    $\tau^*= \sr\Ctx\gst{\tau_1}\cdot \sr\Ctx{\tau_1(\gst)}{\tau_2}\cdot
    \ldots\cdot\sr\Ctx{\tau_1\cdot\ldots\cdot\tau_{n-1}(\gst)}{\tau_n}$,
    then the following holds:
    \begin{itemize}
     \item[a)] $\tau^*$ is applicable to $\gst$, and $\tau^*(\gst)=\tau(\gst)$,
     \item[b)] $|\tau^*|\le 2\cdot n\cdot |\ruleset |.$
    \end{itemize}
\end{proposition}

\begin{proof}
 Firstly note that for every $k$ with $1\le k\le n$, $\tau_k$
    is a steady conventional schedule applicable to
    $\tau_1\cdot\ldots\cdot\tau_{k-1}(\gst)$.
 Therefore, by Proposition~\ref{prop:srep-ex}, for every $k$ with 
    $1\le k\le n$ holds that 
    \begin{itemize}
     \item $\sr\Ctx{\tau_1\cdot\ldots\cdot\tau_{k-1}(\gst)}{\tau_k}$ is 
    applicable to $\tau_1\cdot\ldots\cdot\tau_{k-1}(\gst)$, 
     \item $\sr\Ctx{\tau_1\cdot\ldots\cdot\tau_{k-1}(\gst)}{\tau_k}
    (\tau_1\cdot\ldots\cdot\tau_{k-1}(\gst))=
    \tau_k(\tau_1\cdot\ldots\cdot\tau_{k-1}(\gst))=
    \tau_1\cdot\ldots\cdot\tau_{k}(\gst)$,
     \item $|\sr\Ctx{\tau_1\cdot\ldots\cdot\tau_{k-1}(\gst)}{\tau_k}|
    \le 2\cdot |\ruleset |$.
    \end{itemize}
 The first two observations imply the statement {\it a)}, and the third one 
    implies {\it b)}.
 \end{proof}

\begin{proposition}\label{prop:srep-ex-concatenation-formula}
Let $\gst$ be a confi\-gu\-ration,
    let $\tau_1\cdot\ldots\cdot\tau_n$, for $n\ge 1$, 
    be a steady conventional schedule applicable to $\gst$
    and
    let $\psi \equiv \bigvee_{\ell\in \critical} \counters[\ell]\neq 0$.
Fix a set $\critical\subseteq\local$. 
If for every $k$ with $1\le k\le n$ holds at least one
    of the following:
    \begin{enumerate}
     \item[a)] $\tau_k$ is a thread of $\gst$ and $\tau_1\cdot\ldots\cdot\tau_n$
	   of $\critical$-type~$A$,
     \item[b)] $\setconf{\tau_1\cdot\ldots\cdot\tau_{k-1}(\gst)}{\tau_k}
	       \models \counters[\ell]\neq 0$,
	       for some $\ell\in\critical$,
    \end{enumerate}
    then $\setconf\gst{\tau^*}\models \psi,$
    for $\tau^*= \sr\Ctx\gst{\tau_1}\cdot \sr\Ctx{\tau_1(\gst)}{\tau_2}\cdot
        \ldots\cdot\sr\Ctx{\tau_1\cdot\ldots\cdot\tau_{n-1}(\gst)}{\tau_n}$.
\end{proposition}

\begin{proof}
 For every $k$, $1\le k\le n$, we know that $\tau_k$
    is a steady conventional schedule applicable to
    $\tau_1\cdot\ldots\cdot\tau_{k-1}(\gst)$.
 We prove the statement by showing that for every $k$ 
    with $1\le k\le n$ holds that
    $$\setconf{\tau_1\cdot\ldots\cdot\tau_{k-1}(\gst)}
    {\sr\Ctx{\tau_1\cdot\ldots\cdot\tau_{k-1}(\gst)}{\tau_k}}
    \models \psi.$$
 If we fix one such $k$, then there are two cases:
 \begin{itemize}
  \item If $\tau_k$ is a thread of $\gst$ and $\tau_1\cdot\ldots\cdot\tau_n$
	of $\critical$-type~$A$, then Proposition~\ref{prop:TypeA.acc} yields 
	the required.
  \item If there exists an $\ell\in\critical$ such that 
	$\setconf{\tau_1\cdot\ldots\cdot\tau_{k-1}(\gst)}{\tau_k}
	\models \counters[\ell]\neq 0$,
	then by Proposition~\ref{prop:representativepath} we know that
	$\setconf{\tau_1\cdot\ldots\cdot\tau_{k-1}(\gst)}
	{\sr\Ctx{\tau_1\cdot\ldots\cdot\tau_{k-1}(\gst)}{\tau_k}}
	\models \counters[\ell]\neq 0,$ which implies the required.
 \end{itemize}

\end{proof}

Proposition~\ref{prop:typeoptions} provides us with a case distinction.
To prove the following theorem, for each of the cases we construct a
     representative schedule.
We do so by repeatedly using Proposition~\ref{prop:movingleft}, to
     reorder transitions in the following way:  In Case~1 we move the
     thread of $\critical$-type~$A$ to the beginning of the schedule.
Then, the representative schedule is obtained by applying
     Proposition~\ref{prop:srep-ex} to the thread of
     $\critical$-type~$A$ and then to the rest.
In Case~2 we move the thread of $\critical$-type~$C$ or $E$ to the
     beginning, and again apply Proposition~\ref{prop:srep-ex} to the
     thread and the rest.
Case~3 is the most involved construction.
A prefix of the $\critical$-type~$D$ thread is moved to the beginning
     followed by the complete $\critical$-type~$E$ thread.
Proposition~\ref{prop:srep-ex} is applied to the prefix, the thread
     and the rest.
If the assumption of the proposition is not satisfied (i.e., there is
     a local state $\ell\in \critical$ such that
     $\setconf\gst{\tau}\models \counters[\ell]\neq 0$), then we just
     apply Proposition~\ref{prop:srep-ex} to~$\tau$.

\newcommand{\thmsteadyonedisjunction}{Fix a threshold automaton~$\Sk = (\local,$ 
    $\initlocal,$ $\globset,$ $\paraset,$ $\ruleset,$ $\ResCond)$, and a set 
    $\critical\subseteq\local$. 
Let $\gst$ be a confi\-gu\-ration such that $\omega(\gst)=\Ctx$,
    and let $\psi \equiv \bigvee_{\ell\in \critical} \counters[\ell]\neq 0$.
Then for every steady conventional schedule $\tau$, applicable to $\gst$,
    with $\setconf \gst\tau \models \psi$,
    there is a steady schedule $\srogen$ with the properties:
\begin{enumerate}
        \item[a)] $\srogen$ is applicable to $\gst$, and $\srogen (\gst)=\tau(\gst)$,
        \item[b)] $|\srogen|\leq 6\cdot |\ruleset |$,
        \item[c)] $\setconf \gst\srogen \models \psi$,
        \item[d)] 
        there exist $\tau_1$, $\tau_2$ and $\tau_3$, (not necessarily nonempty) subschedules of  $\tau$, such that
        $\tau_1\cdot\tau_2\cdot\tau_3$ is applicable to $\gst$, it holds that 
        $\tau_1\cdot\tau_2\cdot\tau_3(\gst)=\tau(\gst)$, 
        and $\srogen=\sr\Ctx\gst{\tau_1}\cdot \sr\Ctx{\tau_1(\gst)}{\tau_2}\cdot
        \sr\Ctx{\tau_1\cdot\tau_2(\gst)}{\tau_3}$.
\end{enumerate}}

\begin{theorem}\label{thrm:steady,one.disjunction}
\thmsteadyonedisjunction
\end{theorem}

\newcommand{\proofsteadyonedisjunction}{
\begin{proof}
We give a constructive proof, and therefore $\tau_1$, $\tau_2$, $\tau_3$ and its 
    properties will be obvious from the construction.

If there is a local state $\ell^*\in \critical$ such that $\setconf\gst{\tau}\models \counters[\ell^*]\neq 0$,
    by Proposition~\ref{prop:boring} we have $\setconf \gst\srgen \models\psi$.
Using properties of $\srgen$ described in Proposition~\ref{prop:srep-ex},
    we see that the required schedule is $$\srogen=\srgen.$$
    
If this is not the case, and $\naming$ is a decomposition of $\gst$ and $\tau$,
    then, since $\setconf \gst\tau \models \psi$, by Proposition~\ref{prop:typeoptions}
    at least one of the following cases is true:
\begin{itemize}
\item[(1)] 
Assume there is an $i\in \threads(\sigma,\tau,\naming)$
    such that $\proj{\tau}{\naming,i}$ is of $\critical$-type $\typeall$. 
We claim that the required schedule is
    $$\srogen=\sr \Ctx\gst{\proj{\tau}{\naming,i}}\cdot 
    \sr \Ctx{\proj{\tau}{\naming,i}(\gst)}{\newrest{\tau}{\naming}{i}}.$$

By Proposition~\ref{prop:movethreadtobeginning}, 
     $\proj{\tau}{\naming,i}\cdot\newrest{\tau}{\naming}{i}$
     is a steady schedule applicable to $\sigma$, and
     $\proj{\tau}{\naming,i}\cdot\newrest{\tau}{\naming}{i}(\gst)=\tau(\gst)$.
Therefore, we can apply Proposition~\ref{prop:srep-ex-concatenation-basic}
    to obtain {\it a)} and {\it b)}.
Since $\proj{\tau}{\naming,i}$ is a thread of $\gst$ and $\tau$, of $\critical$-type $\typeall$,
    and by Proposition~\ref{prop:typeoptionA} there is an $\ell^*\in \critical$ such that 
    $\setconf{\proj{\tau}{\naming,i}(\gst)}{\newrest{\tau}{\naming}{i}} \models\counters[\ell^*]\ne 0$,
    then {\it c)} holds by Proposition~\ref{prop:srep-ex-concatenation-formula}.

\item[(2)] Here we assume there exist $i,j\in \threads(\sigma,\tau,\naming)$ 
    such that $i\neq j$, 
    $\proj{\tau}{\naming,j}$ is of $\critical$-type $\typebeg$ or $\typebegend$, and 
    $\proj{\tau}{\naming,i}$ is of $\critical$-type $\typeend$ or $\typebegend$.
We show that the required schedule is $$\srogen=
      \sr \Ctx\gst{\proj{\tau}{\naming,j}} \cdot
      \sr\Ctx{\proj{\tau}{\naming,j}(\gst)}{\newrest{\tau}{\naming}{j}}.$$
Again, by Proposition~\ref{prop:movethreadtobeginning}, 
     $\proj{\tau}{\naming,j}\cdot\newrest{\tau}{\naming}{j}$
     is a steady schedule applicable to $\sigma$, and
     $\proj{\tau}{\naming,j}\cdot\newrest{\tau}{\naming}{j}(\gst)=\tau(\gst)$.
Therefore, we can apply Proposition~\ref{prop:srep-ex-concatenation-basic}
    to obtain {\it a)} and {\it b)}.
By Proposition~\ref{prop:typeoptionBC}, there exist $\ell_1,\ell_2\in \critical$ such that 
    $\setconf \gst{\proj{\tau}{\naming,j}} \models\counters[\ell_1]\neq 0$
    and $\setconf{\proj{\tau}{\naming,j}(\gst)}{\newrest{\tau}{\naming}{j}} \models 
    \counters[\ell_2]\neq 0.$
Thus, {\it c)} holds by Proposition~\ref{prop:srep-ex-concatenation-formula}.

\item[(3)]
For the last case we assume there exist 
    $i,j\in \threads(\sigma,\tau,\naming)$ 
    such that 
    $\proj{\tau}{\naming,i}$ is of $\critical$-type $\typebegend$, and 
    $\proj{\tau}{\naming,j}$ is of $\critical$-type $\typemid$.
We represent $\proj{\tau}{\naming,j}$ as 
    $\proj{\tau}{\naming,j}^1\cdot\proj{\tau}{\naming,j}^2$,
    where $\laststate{\proj{\tau}{\naming,j}^1}\subseteq \critical$.
With a similar idea as in the previous cases, we show that the required schedule is
\begin{eqnarray*}
 \srogen &=&\sr \Ctx\gst{\proj{\tau}{\naming,j}^1} \cdot \\
    &&\cdot \sr \Ctx{\proj{\tau}{\naming,j}^1(\gst)}{\proj{\tau}{\naming,i}} \cdot \\
    &&\cdot \sr \Ctx{\proj{\tau}{\naming,j}^1\cdot 
    \proj{\tau}{\naming,i}(\gst)}{\proj{\tau}{\naming,j}^2\cdot
    \newrest{\tau}{\naming}{i,j}}.
\end{eqnarray*}
Again, statements {\it a)} and {\it b)} follow from Proposition~\ref{prop:movethreadsmix} 
  and Proposition~\ref{prop:srep-ex-concatenation-basic}.
By Proposition~\ref{prop:typeoptionED}, there exist 
    $\ell_1,\ell_2,\ell_3\in\critical$ such that
    \begin{itemize}
        \item $\setconf \gst{\proj{\tau}{\naming,j}^1} \models \counters[\ell_1]\neq 0$,
        \item $\setconf {\proj{\tau}{\naming,j}^1(\gst)}{\proj{\tau}{\naming,i}} \models
            \counters[\ell_2]\neq 0$, and 
        \item $ \setconf {\proj{\tau}{\naming,j}^1\cdot\proj{\tau}{\naming,i}(\gst)}
            {\proj{\tau}{\naming,j}^2\cdot \newrest{\tau}{\naming}{i,j}}\models
            \counters[\ell_3]\neq 0$.
    \end{itemize}
Using these three facts and Proposition~\ref{prop:srep-ex-concatenation-formula}, we obtain {\it c)}.
\end{itemize}
\end{proof}
}
\proofintext{\proofsteadyonedisjunction}

\proofsteadyonedisjunction

\subsection{Representative Schedules maintaining
 \boldmath $\bigwedge_{\critical \in Y}
    \bigvee_{i \in \critical} \counters[i] \ne 0$}\label{subsec:liveness3}

The construction we give in this section requires us to apply the same
     schedule twice.
We confirm that we can do that by proving in
     Proposition~\ref{prop:dublesize}  that if a counterexample exists
     in a small system, there also exists one in a bigger system.
In the context of counter systems we formalize this using a
     multiplier:   
 
\begin{definition}[Multiplier]
 A \emph{multiplier}~$\multipl$ of a threshold automaton is a number~$\multipl\in \Natural$,
    such that for every guard $\varphi$, if $(\gst.\counters, \gst.\vars, \gst.\param)\models\varphi$,
    then also $({\multipl}\cdot{\gst}.\counters,{\multipl}\cdot{\gst}.\vars,
    {\multipl}\cdot{\gst}.\param)\models\varphi$, and
    ${\multipl}\cdot{\gst}.\param \in \AdmP$.
\end{definition}

For specific pathological threshold automata, such multipliers may not
     exist.
However, all our benchmarks have multipliers, and as can be seen from
     the definitions, existence of multipliers can easily be checked
     using simple queries to SMT solvers in preprocessing.

\begin{definition}
If $\gst$ is a configuration, and $\multipl\ge 1$ is a multiplier, then
    we define $\multist{\gst}{\multipl}$ to be the configuration with 
    $(\multist{\gst}{\multipl}).\counters=\multipl\cdot\gst.\counters$, 
    $(\multist{\gst}{\multipl}).\vars=\multipl\cdot\gst.\vars$, and
    $(\multist{\gst}{\multipl}).\param = \multipl\cdot \gst.\param$.
If $\tau$ is a conventional schedule,
    we define $\multisch{\tau}{\multipl}= 
    \underbrace{\tau \concat \ldots \concat \tau}_{\multipl \text{ times}}.$
\end{definition}

\newcommand{\propdublesize}{Let $\gst_1$, $\gst_1'$ and $\gst_2$ be configurations, let
    $\tau$ be a steady conventional schedule 
    applicable to $\gst_1$ and $\gst_2$, and let $\ell$ be an 
    arbitrary local state.
If a multiplier is $\multipl > 1$, then the following holds:
\begin{enumerate}
 \item $\multisch{\tau}{\multipl}$ is applicable to $\multist{\gst_1}{\multipl}$,
      and if $\tau(\gst_1)=\gst_1'$ then 
      $\multisch{\tau}{\multipl}(\multist{\gst_1}{\multipl})= \multist{\gst_1'}{\multipl}$,
 \item for every propositional formula $\psi$, if $\gst\models \psi$, 
      then $\multist{\gst}{\multipl}\models \psi$,
 \item if $\gst_1.\counters[\ell]<\gst_2.\counters[\ell]$, then
     $\tau(\gst_1).\counters[\ell]<\tau(\gst_2).\counters[\ell]$,
 \item if $\gst_1.\counters[\ell]>0$ then
     $\multist{\gst_1}{\multipl}.\counters[\ell] > \gst_1.\counters[\ell]$.
\end{enumerate}
}
\begin{proposition}\label{prop:dublesize}
\propdublesize
\end{proposition}

\newcommand{\proofdublesize}{
\begin{proof}
 All properties follow directly from definition of counter systems.
\end{proof}
}
\proofintext{\proofdublesize}
\proofdublesize

\newcommand{\propdublemaintains}{
Let~$\gst$ be a confi\-gu\-ration, let~$\tau$ be a steady conventional schedule
    applicable to $\gst$, and let~$\psi$ be a propositional formula.
If~$\multipl$ is a multiplier and if $\setconf \gst\tau \models \psi$, then
    $\setconf {\multist{\gst}{\multipl}}{\multisch{\tau}{\multipl}} \models \psi$.
}
\begin{proposition}\label{prop:dublemaintains}
\propdublemaintains
\end{proposition}

\newcommand{\proofdublemaintains}{
\begin{proof}
Paths~$\setconf \gst\tau$ and~$\setconf {\multist{\gst}{\multipl}}{\multisch{\tau}{\multipl}}$
    are trace equivalent by Proposition~\ref{prop:dublesize} 1 and 2.
Therefore, by~\cite[Corollary 3.8]{BK08}, they satisfy the same linear temporal properties.
\end{proof}
}
\proofintext{\proofdublemaintains}
\proofdublemaintains

To prove the following theorem, we use the schedule $\multisch{\tau}\multipl$.
As in the previous cases, we divide it in two parts, namely
     to $\tau$ and $\multisch{\tau}{(\multipl-1)}$, and then apply
     Proposition~\ref{prop:srep-ex} to both of them separately.
For the proof, we use statements (3) and (4) from Proposition~\ref{prop:dublesize}.

\recallthm{thm:andor}{\thmandor}

\newcommand{\proofandor}{
\begin{proof}
We show that the required schedule is 
    $$\gsrogen= \sr\Ctx{\multist{\gst}{\multipl}}\tau \cdot 
    \sr\Ctx{\tau(\multist{\gst}{\multipl})}{\multisch{\tau}{(\multipl-1)}}.$$
By the properties of $  \sr\Ctx{\multist{\gst}{\multipl}}\tau$ and 
    $\sr\Ctx{\tau(\multist{\gst}{\multipl})}{\multisch{\tau}{(\multipl-1)}}$
    from Proposition~\ref{prop:srep-ex},  we see that $\gsrogen$ 
    is a steady schedule applicable to $\multist{\gst}{\multipl}$, 
    that $\gsrogen (\multist{\gst}{\multipl})=\multisch{\tau}{\multipl}(\multist{\gst}{\multipl})$, 
    and finally $|\gsrogen|\leq  4\cdot |\ruleset |$.
Now it remains just to show that {\it c)} holds. 

 For every $\gcri\le \gcritical$, we denote $\bigvee_{\ell\in \mathit{Locs}_\gcri} \counters[\ell]\neq 0$ 
    by $\psi_\gcri$.
Since $\psi=\bigwedge_{1\le\gcri\le \gcritical} \psi_\gcri$, 
    we prove that for every $\gcri\le \gcritical$, holds 
    $\setconf {\multist{\gst}{\multipl}}\gsrogen \models \psi_\gcri$.
Let us fix an $\gcri\le \gcritical$.
Since $\setconf \gst\tau \models \psi$, it is also true that
    $\setconf \gst\tau \models \psi_\gcri$. 
Therefore, we have that
\begin{itemize}
\item $\gst\models\psi_\gcri$, which implies that there exist an $\ell_\gcri^1\in\mathit{Locs}_\gcri$ 
    with $\gst.\counters[\ell_\gcri^1]\ge 1$, and
\item $\tau(\gst)\models\psi_\gcri$, which implies that there is an $\ell_\gcri^2\in\mathit{Locs}_\gcri$ 
    with $\tau(\gst).\counters[\ell_\gcri^2]\ge 1$.
    \end{itemize}

Now we show that:\\
i) $\setconf{\multist{\gst}{\multipl}}{\tau}\models \counters[\ell_\gcri^1]\ge 1$, and\\
ii) $\setconf{\tau(\multist{\gst}{\multipl})}{\multisch{\tau}{(\multipl-1)}}\models \counters[\ell_\gcri^2]\ge 1$.

i) Let $\tau'$ be an arbitrary prefix of $\tau$.
From the assumption and Proposition~\ref{prop:dublesize}~(4) we have that
    $1\le \gst.\counters[\ell_\gcri^1]< (\multist{\gst}{f}).\counters[\ell_\gcri^1]$.
Then, from Proposition~\ref{prop:dublesize}~(3) we see that 
    $\tau'(\gst).\counters[\ell_\gcri^1]< \tau'(\multist{\gst}{\multipl}).\counters[\ell_\gcri^1]$.
Now, since $\tau'$ is applicable to $\gst$, and therefore it is $\tau'(\gst).\counters[\ell_\gcri^1]\ge 0$, 
    we obtain that $ \tau'(\multist{\gst}{\multipl}).\counters[\ell_\gcri^1]\ge 1$.
Hence, we have $\setconf{\multist{\gst}{\multipl}}{\tau}\models \counters[\ell_\gcri^1]\ge 1$.
    
ii) Let us denote $\{i \colon i \in \threads(\multist{\gst}\multipl,\tau,\naming) \wedge 
    \proj{\tau}{\naming,i}.\tostate = \ell_\gcri^2\}$ by $T$, and 
    $\{i \colon i \in \threads(\multist{\gst}\multipl,\tau,\naming) \wedge 
    \proj{\tau}{\naming,i}.\fromstate = \ell_\gcri^2\}$ by $F$.
 By Proposition~\ref{prop:threadcounts} we have that
     $0< \tau(\gst).\counters[\ell_\gcri^2]= \gst.\counters[\ell_\gcri^2] + |T| -|F|.$
 This implies that $ \gst.\counters[\ell_\gcri^2] > |F| -|T|.$
By Proposition~\ref{prop:threadcounts}, we also obtain 
\begin{eqnarray*}
 \tau(\multist{\gst}{\multipl}).\counters[\ell_\gcri^2]&= &
     \multist{\gst}{\multipl}.\counters[\ell_\gcri^2] + |T| -|F|=\\
 &=& \multist{\gst}{(\multipl-1)}.\counters[\ell_\gcri^2] + \gst.\counters[\ell_\gcri^2] - (|F| -|T|),
\end{eqnarray*}
      which combined with $\gst.\counters[\ell_\gcri^2] > |F| -|T|$ yields 
      $$ \multist{\gst}{(\multipl-1)}.\counters[\ell_\gcri^2] < 
      \tau(\multist{\gst}{\multipl}).\counters[\ell_\gcri^2].$$
Let now $\tau'$ be an arbitrary prefix of $\multisch{\tau}{(\multipl-1)}$. 
Using the fact that $\tau'$ is applicable to $\multist{\gst}{(\multipl-1)}$, 
    and Proposition~\ref{prop:dublesize}~(3),
    we obtain that 
    $$0\le \tau'(\multist{\gst}{(\multipl-1)}).\counters[\ell_\gcri^2] < 
    \tau'(\tau(\multist{\gst}{\multipl})).\counters[\ell_\gcri^2].$$
Therefore, we obtain that $ \tau'(\tau(\multist{\gst}{\multipl})).\counters[\ell_\gcri^2]\ge 1$, and
    hence $\setconf{\tau(\multist{\gst}{\multipl})}{\multisch{\tau}{(\multipl-1)}}
    \models \counters[\ell_\gcri^2]\ge 1$.
   
Now, when the statements i) and ii) are proved, we can apply Proposition~\ref{prop:srep-ex-concatenation-formula}
This gives us that
$\setconf {\multist{\gst}{\multipl}}\gsrogen \models \psi_\gcri$, for an arbitrary 
    $\gcri\le \gcritical$, which implies that {\it c)} is true, and concludes the proof.
\end{proof}
}
\proofintext{\proofandor}
\proofandor

\subsection{Representative Schedules maintaining
 \boldmath $\bigwedge_{i \in \critical} \counters[i] = 0$}
\label{subsec:liveness3}

This case is the simplest one, so that $\sr{\Ctx}{\gst}{\tau}$ from
     Section~\ref{cons:srep} can directly be used as representative schedule.

\newcommand{\thmblabla}{
Fix a threshold automaton~$\Sk = (\local,$ $\initlocal,$ $\globset,$
     $\paraset,$ $\ruleset,$ $\ResCond)$, and a configuration $\sigma$.
If     $\psi \equiv
    \bigwedge_{i\in \critical} \bk[i] = 0,$ 
for $\critical\subseteq\local$,
then for every steady schedule $\tau$ applicable to $\gst$,
     and with $\setconf \gst\tau \models \psi$, schedule~$\srgen$ satisfies:
\begin{enumerate}
        \item[a)] $\srgen$ is applicable to $\gst$, and $\srgen(\gst)=\tau(\gst)$,
        \item[b)] $|\srgen|\leq 2\cdot |\ruleset|$,
        \item[c)] $\setconf \gst\srgen\models \psi$.
\end{enumerate}}

\begin{theorem}\label{thm:blabla}
\thmblabla
\end{theorem}

\newcommand{\proofblabla}{
\begin{proof}
Since $\setconf \gst\tau \models \psi$, we know that for every
    transition $t$ from $\tau$ and for every local state $\ell\in \critical$
     it holds $t.\fromstate\neq\ell$ and $t.\tostate\neq\ell$.
Let $\srogen = \sr{\Ctx}{\gst}{\tau}$.
By Proposition~\ref{prop:MichaelJordan},  $\sr{\Ctx}{\gst}{\tau}$ contains a
     subset of the rules that appear in~$\tau$.
Hence, $\srogen$ does not change counters of states in~$\critical$.
Other properties follow from Proposition~\ref{prop:srep-ex}
\end{proof}
}
\proofintext{\proofblabla}
\proofblabla

\subsection{Proof of Theorem~\ref{thm:main}}

The theorem follows from Theorem~\ref{thrm:steady,one.disjunction}
 and~\ref{thm:blabla}.

\fi

\end{document}